\documentclass[aps,pra,reprint]{revtex4-1}

\usepackage{soul,color}
\usepackage{xcolor}
\usepackage{graphicx}
\usepackage[shortlabels]{enumitem}

\usepackage{amsfonts}
\usepackage{amsthm}

\usepackage{amssymb}
\usepackage{amsmath}
\usepackage{dcolumn}
\usepackage{physics}
\usepackage{float}
\usepackage{bm}
\usepackage{verbatim}
\usepackage{braket}
\usepackage{bbm}
\usepackage{tikz}
\usetikzlibrary{quantikz2}
\usepackage[colorlinks,linkcolor=blue,anchorcolor=blue,citecolor=blue,urlcolor=blue]{hyperref}

\newcommand{\nmg}{\mathcal{T}}
\newtheorem{theorem}{Theorem}

\newtheorem{lemma}[theorem]{Lemma}

\newtheorem{proposition}{Proposition}

\usepackage{mathtools}
\DeclarePairedDelimiter\ceil{\lceil}{\rceil}
\DeclarePairedDelimiter\floor{\lfloor}{\rfloor}
\newcommand{\sts}{\mathrm{TS}_{\rm sym}}
\newcommand{\Tsym}{\mathcal{T}_{\rm sym}(n,m)}
\newcommand{\Tsymsq}{\mathcal{T}^2_{\rm sym}(n,m)}
\newcommand{\cts}{\mathrm{TS}_{\rm cyc}}
\newcommand{\Tcyc}{\mathcal{T}_{\rm cyc}(n,m)}
\newcommand{\Tcycsq}{\mathcal{T}^2_{\rm cyc}(n,m)}
\newcommand{\rbf}{\mathbf{R}^{(T,T')}}
\newcommand{\unit}{\mathrm{U}(\mathcal{H})}
\newcommand{\symdif}{\ensuremath{\bigtriangleup}}
\newcommand{\Symdif}{\mbox{\Large \ensuremath{\bigtriangleup}}}
\newcommand{\toric}{\ensuremath{\psi_{\rm toric}}}

\begin{document}
\title{The quantum trajectory sensing problem and its solution}
\author{Zachary E. Chin} \email{zchin@mit.edu}
\affiliation{Department of Physics and Research Laboratory of Electronics, Massachusetts Institute of Technology, Cambridge, Massachusetts 02139, USA}
\author{Isaac L. Chuang}
\affiliation{Department of Physics and Research Laboratory of Electronics, Massachusetts Institute of Technology, Cambridge, Massachusetts 02139, USA}
\date{\today}

\begin{abstract}

     The quantum trajectory sensing problem seeks quantum sensor states which enable the trajectories of incident particles to be distinguished using a single measurement. 
     For an $n$-qubit sensor state to unambiguously discriminate a set of trajectories with a single projective measurement, all post-trajectory output states must be mutually orthogonal; therefore, the $2^n$ state coefficients must satisfy a system of constraints which is typically very large. 
     Given that this system is generally challenging to solve directly, we introduce a group-theoretic framework which simplifies the criteria for sensor states and exponentially reduces the number of equations and variables involved when the trajectories obey certain symmetries. These simplified criteria yield general families of trajectory sensor states and provide bounds on the particle-sensor interaction strength required for perfect one-shot trajectory discrimination. 
     Furthermore, we establish a link between trajectory sensing and quantum error correction, recognizing their common motivation to identify perturbations using projective measurements. 
     Our sensor states in fact form novel quantum codes, and conversely, a number of familiar stabilizer codes (such as toric codes) also provide trajectory sensing capabilities. 
     This connection enables noise-resilient trajectory sensing through the concatenation of sensor states with quantum error-correcting codes. 
\end{abstract}

\maketitle
\section{Introduction}
The trajectory of a particle through a medium creates spatial patterns that reveal its intrinsic properties and the conditions of its origin. By tracking the motion of collision products, detectors at high-energy particle colliders can determine quantities such as particle charge, momentum, energy, and lifetime and also elucidate production/decay mechanisms \cite{atlas,cms}. Additionally, scintillators \cite{snoplus}, bubble chambers \cite{snolab-bubbles}, and Cherenkov detectors \cite{icecube} are used to investigate the cosmic origins of muons, neutrinos, and potential dark matter candidates by inferring their paths intercepted near Earth's surface. Furthermore, the spatial distribution of emitted or transmitted particles is the basis for imaging in diverse applications such as telescopy \cite{webb}, diagnostic medicine \cite{spect}, and electron microscopy \cite{TEM}. 

Because realistic particles often only interact weakly with a sensor, particle trajectory sensing stands to benefit from quantum resources, which are known to improve sensitivities for measurements of forces and fields \cite{field,gravimeter}. A model for a quantum trajectory sensor might consist of an array of networked quantum systems (qubits or qudits) such that an incident particle applies the same local unitary operation to each system coincident with its path. The different trajectories of the particle would produce various perturbation patterns on the array which could ideally be distinguished using a single projective measurement. Although quantum sensor networks have been developed for estimating spatially-distributed continuous variables, such as multiple local parameters \cite{multi-phase} or linear functions of local parameters \cite{lin-func}, the set of possible trajectory patterns within the above model is discrete. In the direction of discrete sensing, previous work has produced quantum schemes which help to localize a perturbation that affects just a single sensor in a network \cite{guptaPRA,gupta}. These advances naturally motivate the exploration of quantum sensor networks which instead aim to discriminate spatial patterns spread over multiple sensors.

Although a quantum trajectory sensor can be understood as a sensor network, it can also be viewed as a kind of quantum code where various spatial perturbations are treated as errors to be distinguished by a syndrome measurement. 
However, the errors corresponding to trajectories are quite different from those considered conventionally in quantum error correction (QEC). While QEC typically focuses on independent and identically-distributed single-qubit Pauli errors \cite{Nielsen_Chuang_2010}, trajectories may involve highly correlated non-Pauli operations affecting many qubits. 
Nevertheless, special quantum codes have been introduced to recover correlated erasure errors due to cosmic rays \cite{QECcosmic, muongoogle,muonmcdermott}, and their usefulness bolsters the prospect of achieving trajectory sensing through QEC. 
Other related work has also shown that repeated syndrome measurements on codes can be employed for the characterization of quantum dynamics \cite{Omkar_2015,omkar_asc} in addition to error channel parameter estimation and hypothesis testing \cite{combes2014insitu}. These techniques are promising for trajectory sensing as well, especially if they could be adapted to use only single-shot measurements.

Importantly, it is unclear how to efficiently find a suitable $n$-qubit sensor state for distinguishing trajectories because the $2^n$-dimensional Hilbert space of possible states is prohibitively large for a brute-force search. Similar challenges are encountered when searching for new quantum codes. To overcome this obstacle in the context of QEC, numerous families of codes have been introduced which use symmetry to simplify their description and construction. For example, a rich class of permutation-invariant codes has been developed which correct for spontaneous decay errors as well as arbitrary qubit errors \cite{ruskai, ouyang, multi-qubit-pi,pi-qudit}; any $n$-qubit permutation-invariant code state can be parameterized with only $O(n)$ variables \cite{schur} and expressed as a superposition of Dicke states, which are known to be useful for metrology \cite{dicke,dicke-metrology}. Along similar lines, there also exist codes displaying cyclic permutation symmetry which correct trajectory-like Pauli burst errors \cite{cyclic-code}. On the other hand, the widely-used stabilizer codes \cite{gottesman1997stabilizer} instead harness symmetries of the Pauli group to efficiently describe a $2^n$-dimensional code space using only $O(n)$ stabilizer group generators. The success of all of the aforementioned codes suggests that the search for trajectory sensors may be analogously simplified by leveraging permutation and Pauli symmetries to reduce the space of prospective sensor states.

In a companion paper \cite{chin2024quantum}, we posed the following trajectory sensing (TS) problem: given a sensor array and set of trajectories, how strongly must the incident particle interact with the array for there to exist sensor states which can unambiguously distinguish the possible trajectories via a single projective measurement? Addressing this challenge, we showed that sensor entanglement can reduce the particle-sensor interaction strength $\theta$ required for perfect trajectory identification. In particular, we introduced families of sensor states, called TS states, which can discriminate trajectories with zero error provided that $\theta$ exceeds a certain threshold and the trajectories obey certain symmetries.

In this paper, we develop a formal mathematical framework for constructing these TS states. Specifically, we initially consider a noiseless, idealized setting of the TS problem where the individual sensors are qubits and the incident particle rotates each qubit along its path by the same angle $\theta\in[0,\pi]$ around a fixed axis of the Bloch sphere; here, $\theta$ parameterizes the strength of the particle-sensor interaction. After establishing basic criteria which determine whether a TS state exists for a given value of $\theta$, we use permutation and Pauli symmetry groups to dramatically simplify these criteria. We subsequently reframe the TS state existence problem as a linear programming problem and derive closed-form bounds on the intervals of $\theta$ over which two broad families of TS states are guaranteed to exist. Switching to the quantum error correction viewpoint, we then build codes from TS states that encode one logical qubit and correct the error channel of applying a random trajectory to the sensor. We then derive criteria that determine whether familiar stabilizer code states are TS states. We subsequently use these criteria to show how existing stabilizer codes, such as toric codes, can be repurposed for sensing trajectories. Lastly, we demonstrate that concatenating TS states with existing error-correcting codes can allow for perfect trajectory sensing even in the presence of environmental noise.

The rest of this paper is organized as follows. In Section \ref{sec:theory}, we formalize the TS problem and develop a group-theoretic approach to simplifying the criteria for the existence of TS states. Then, in Section \ref{sec:simp-criteria}, we use these tools to derive a reduced system of equations which determines whether a TS state exists at a particular $\theta$. We ultimately apply this result for two families of TS problems in Section \ref{sec:results} to determine the intervals of $\theta$ where TS states exist. Finally, in Section \ref{sec:qecc}, we develop a correspondence between TS states and stabilizer codes and discuss the construction of TS states resilient against noise.

\section{Quantum trajectory sensing: theory}\label{sec:theory}

We begin by assembling a theoretical formalism within which a given TS scenario can be discussed in a mathematically concrete fashion. First, in Section \ref{sec:formalism}, we introduce some foundational assumptions about the TS scenario and formally define the TS problem. Then, in Section \ref{sec:symmetries}, we use symmetry groups of the TS scenario to systematically simplify the criteria for useful TS states.

\subsection{Trajectory sensing formalism}\label{sec:formalism}

A quantum trajectory sensor consists of an array of $n$ individual quantum systems each uniquely labeled with an integer index $1,\ldots,n$. In this work, we choose the systems to be qubits, although a more general treatment might consider qudits instead. Let $\mathcal{H}$ denote the Hilbert space spanned by the possible states of the array, and note that $\dim \mathcal{H} = 2^n$. Additionally, let $\unit$ denote the set of unitary operators on $\mathcal{H}$.

We assume that the interaction between particle and sensor qubits is short-range. In particular, only those qubits which are exactly coincident with the particle path are rotated, and they are all rotated by the same local unitary operator in $\mathrm{SU}(2)$. Any operator in $\mathrm{SU}(2)$ can be described as a rotation about some axis of the Bloch sphere; without loss of generality, we choose the qubit rotation axis to be $Z$, since all axes are equivalent up to a change of computational basis. Thus, assume the incident particle applies the operator
    \begin{align}
        R_Z(\theta) = \exp(-\frac{i\theta}{2} Z)
    \end{align}
to each qubit along its path, where $Z$ is the Pauli-$Z$ operator. The parameter $\theta$ represents the particle-qubit interaction strength and is allowed to take values in the interval $[0,\pi]$. 

The assumption of a short-range interaction allows us to meaningfully define a trajectory. Note that any two particle paths which intersect the same set of qubits induce an identical perturbation on the sensor array. Since such paths cannot reasonably be distinguished by this sensor model, we propose that they represent equivalent trajectories. Accordingly, we define the \textit{trajectory} $T$ associated with a path to be the set of qubits intercepted by the particle; $T$ is thus a subset of $[n]\coloneq \{1,\ldots n\}$. The perturbation induced on the whole array by a trajectory $T$ is then represented with the $n$-qubit unitary $R^{(T)}(\theta)$, which corresponds to applying $R_Z(\theta)$ to each qubit in $T$:
\begin{align}\label{eq:RS}
        R^{(T)}(\theta) = \bigotimes_{j=1}^n R_Z(\theta \cdot \mathbbm{1}_T(j)),
\end{align}
where $\mathbbm{1}_T(j)$ is an indicator function that returns $1$ if $j\in T$ and $0$ otherwise.

A TS scenario is characterized by the array size $n$, the interaction strength $\theta$, and the set $\mathcal{T}$ of allowed trajectories to be discriminated. Note that $\mathcal{T}\subseteq \mathbb{P}([n])$, where $\mathbb{P}([n])$ is the power set (i.e., the set of subsets) of $[n]$. In this work, we assume that all allowed trajectories in $\mathcal{T}$ are of equal size, where the size of a trajectory $T$ is defined as $\abs{T}$, the number of qubits it contains. Given an input sensor state $\ket{\psi}$, each of the trajectories $T\in\mathcal{T}$ yield a distinct output state $R^{(T)}(\theta)\ket{\psi}$ (see Figure \ref{fig:scheme}). 

For $\ket{\psi}$ to be a useful trajectory sensor, we require that the various trajectories are distinguishable via a single projective measurement. Equivalently, the all trajectory output states must be mutually orthogonal, thereby imposing the following conditions on the sensor state $\ket{\psi}$:
\begin{align}
    \bra{\psi}R^{\dagger(T)}(\theta)R^{(T')}(\theta)\ket{\psi} = \delta_{T,T'}
\end{align}
for all $T,T'\in\nmg$, where $\delta_{T,T'} = 1$ if $T = T'$ and 0 otherwise. To underscore the fact that each orthogonality condition above involves a pair of trajectories, we define the \textit{orthogonality operator}
\begin{align}
    \rbf(\theta) = R^{\dagger(T)}(\theta)R^{(T')}(\theta)
\end{align}
so that the above criteria take the simpler form 
\begin{align}\label{eq:ortho}
    \bra{\psi}\rbf(\theta)\ket{\psi} = \delta_{T,T'}
\end{align}
for all $T,T'\in\mathcal{T}$. 
\begin{figure}[htbp]
  \includegraphics[width = \linewidth]{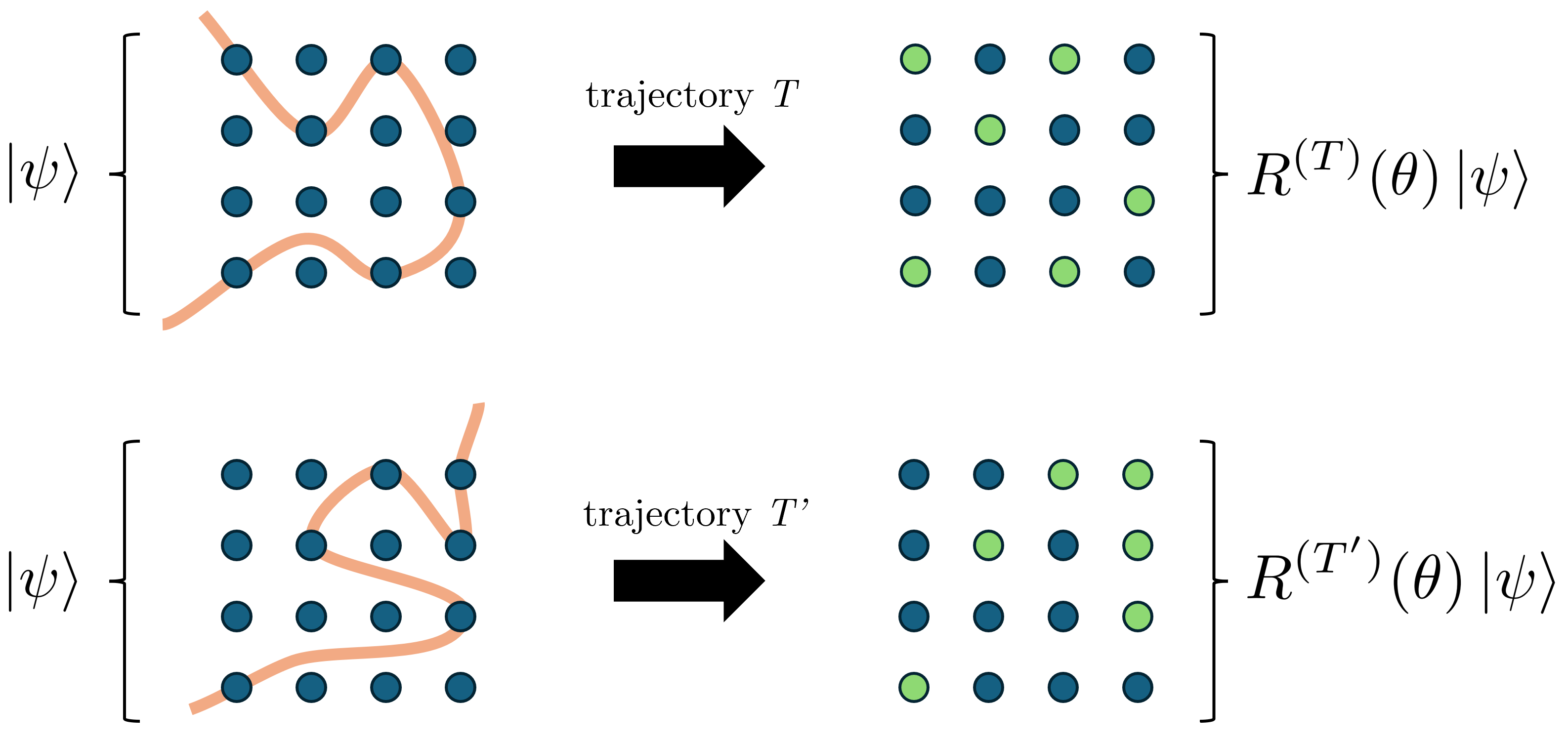}
  \caption{Two trajectories $T$ and $T'$ of the incident particle induce different perturbations on a sensor state $\ket{\psi}$. The small circles represent sensor qubits; the particle rotates each green qubit along its path by $R_Z(\theta)$. Trajectories $T$ and $T'$ are distinguishable by a single projective measurement if they produce orthogonal output states $R^{(T)}(\theta)\ket{\psi}$ and $R^{(T')}(\theta)\ket{\psi}$.} 
  \label{fig:scheme}
\end{figure}

Given $n$ and $\mathcal{T}$, any $\ket{\psi}$ satisfying Eq. (\ref{eq:ortho}) at a particular value of $\theta$ is called a \textit{TS state}. Note that a TS state satisfying these criteria at one value of $\theta$ generally need not satisfy them at other values of $\theta$. Furthermore, since $\ket{\psi}$ and $R^{(T)}(\theta)\ket{\psi}$ are not required to be orthogonal, these TS states generally do not detect the mere presence of a particle and only discriminate trajectories provided the particle has already interacted with the sensor.

For a given $n$ and $\mathcal{T}$, the \textit{TS problem} asks for which $\theta\in[0,\pi]$ there exists a TS state satisfying Eq. (\ref{eq:ortho}). In principle, a TS problem is straightforward, although computationally expensive, to solve via naive means. One can determine whether a TS state exists at a particular value of $\theta$ by substituting the general ansatz $\ket{\psi} = \sum_{j=0}^{2^n-1}a_j\ket{{j}}$ (where $a_j\in\mathbb{C}$ and the $\ket{j}$ are $Z$-eigenbasis vectors) into Eq. (\ref{eq:ortho}) and solving for $a_j$. A TS state exists if and only if there is a valid solution of $a_j$. The difficulty in this approach stems from the fact that the system usually contains an enormous number of variables and equations---there are $2^n$ complex variables and $\abs{\nmg}^2$ equations. For any reasonably large qubit array and set of allowed trajectories, this method of solving the TS problem is intractable.

Fortunately, the TS scenario described above is designed such that the symmetries involved can enable dramatic simplification of a given TS problem via group theory. In particular, symmetries arise from the assumptions that (1) every qubit in a trajectory is perturbed by the same single-qubit rotation and (2) this rotation takes place about the $Z$-axis of the Bloch sphere. These two assumptions respectively lead to groups of permutations and Pauli matrices under which a given TS problem remains invariant. Thus, in the following section, we deploy group theory to explicitly describe how these permutation and Pauli symmetries can be utilized to simplify the search for TS states. 

\begin{figure*}[t]
  \includegraphics[width = 0.75\linewidth]{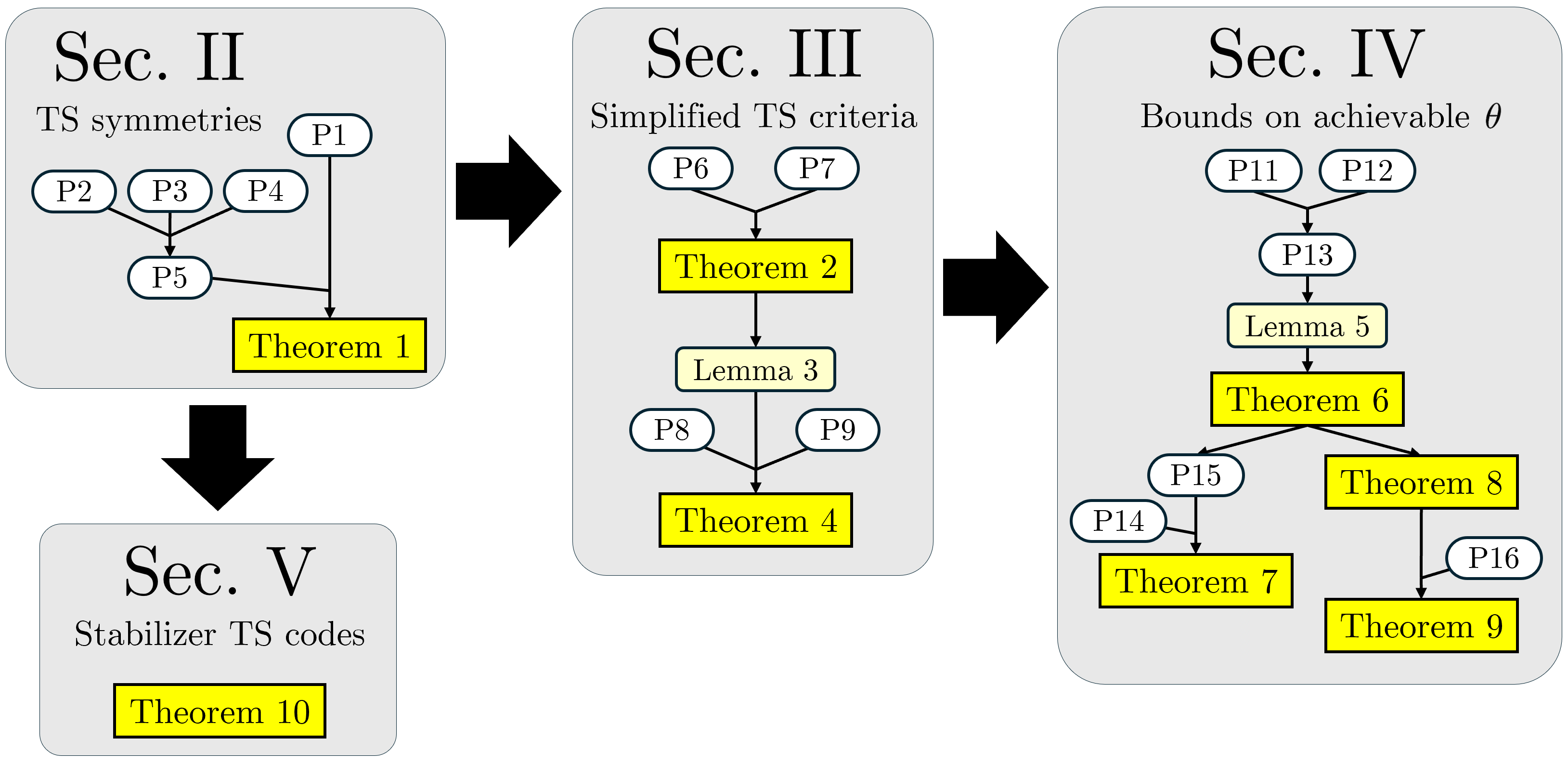}
  \caption{Logical flow of Sections \ref{sec:theory}-\ref{sec:qecc}. The small white bubbles represent propositions (e.g., P1 represents Proposition 1).} 
  \label{fig:theorem-flow}
\end{figure*}

\subsection{Permutation groups and Pauli stabilizer groups}\label{sec:symmetries}
We iteratively build a group-theoretic framework for solving the TS problem through the following steps. First, in Section \ref{sec:rbf-sym}, we formally explore how permutation and Pauli symmetry groups naturally emerge for a given TS problem from the above two assumptions. The goal is ultimately to use these symmetries to reduce the TS state criteria, that is, the system of equations in Eq. (\ref{eq:ortho}). In other words, we seek to understand which of these equations may become equivalent due to symmetry. In Section \ref{sec:index-transform}, we thus develop tools to describe how the orthogonality operators and sensor states constituting these equations transform under the action of permutation and Pauli matrices. We finally invoke these tools in Section \ref{sec:redundant} to establish a reduced set of criteria for symmetry-invariant TS states.

\subsubsection{Symmetries of the TS problem}\label{sec:rbf-sym}
Any operations which leave the criteria of Eq. (\ref{eq:ortho}) unchanged are symmetries of the corresponding TS problem. For a given sensor state $\ket{\psi}$, these criteria are determined by the set $\mathcal{R}$ of allowed orthogonality operators, where
\begin{align}\label{eq:R-inv}
    \mathcal{R} = \{\rbf(\theta)\ :\ T,T'\in\mathcal{T}\}.
\end{align}
Suppose there exists a unitary $U\in \unit$ such that $\mathcal{R}$ is invariant under conjugation by $U$; that is, for every $\rbf \in \mathcal{R}$,
\begin{align}
    U\rbf U^\dagger \in \mathcal{R}.
\end{align}
Note that conjugation by $U$ implements a bijection from $\mathcal{R}$ to itself. Hence, conjugating all of the orthogonality operators in Eq. (\ref{eq:ortho}) by $U$ leads to identical criteria and produces a physically equivalent TS scenario.

Such symmetries can enable simplification of a TS problem. Suppose $\mathcal{R}$ is invariant under conjugation by some $U\in\unit$. The following proposition then asserts that fewer criteria from Eq. (\ref{eq:ortho}) are needed to confirm a state as a TS state, assuming the state is also invariant under $U$. Specifically, the proposition shows how some criteria may become redundant due to symmetry:
\begin{proposition}\label{prop:U-simp}
    Suppose $\theta \neq 0$ and there exists $U\in\unit$ and $T_1,T_1', T_2,T_2'\in\mathcal{T}$ such that 
\begin{align}\label{eq:U-conj}
    U\mathbf{R}^{(T_1,T_1')}(\theta)U^\dagger = \mathbf{R}^{(T_2,T_2')}(\theta).
\end{align}

Then for any $\ket{\psi}$ satisfying $U\ket{\psi} = \ket{\psi}$, Eq. (\ref{eq:ortho}) holds for $T=T_1$ and $T'=T_1'$ if and only if it holds for $T=T_2$ and $T'=T_2'$.
\end{proposition}
\begin{proof}
    $U\ket{\psi}=\ket{\psi}$ implies that $U^\dagger\ket{\psi}=\ket{\psi}$. Taken with Eq. (\ref{eq:U-conj}), this fact implies that
    \begin{align}\label{eq:U-conj2}
    \bra{\psi}\mathbf{R}^{(T_1,T_1')}(\theta)\ket{\psi} = \bra{\psi}\mathbf{R}^{(T_2,T_2')}(\theta)\ket{\psi}.
\end{align}
To prove the desired result, it is sufficient to show that $\delta_{T_1,T'_1} = \delta_{T_2,T'_2}$. Note that any $\rbf$ equals the identity matrix if and only if $T=T'$ since $\theta \neq 0$. It then follows from Eq. (\ref{eq:U-conj}) that $T_2=T_2'$ if and only if $T_1=T_1'$, thereby proving the claim. 
\end{proof}

Two symmetries of $\mathcal{R}$ arise due to the assumption that every qubit in the trajectory is rotated by the same local unitary. Because of this assumption, each orthogonality operator in Eq. (\ref{eq:ortho}) can be written as $\rbf = \bigotimes_{j=1}^nU_j$, where $U_j\in\{I,R_Z,R^{\dagger}_Z\}$. It follows that $\mathcal{R} \subseteq \{I,R_Z,R^\dagger_Z\}^{\otimes n}$, that is, the set of $n$-fold tensor products of the operators $\{I,R_Z,R^{\dagger}_Z\}$. We now provide two observations. Firstly, on a global level, the repeated tensor product structure implies that $\mathcal{R}$ may be invariant under certain permutations of the qubits. Secondly, since each local set $\{I,R_Z,R^{\dagger}_Z\}$ is closed under conjugation by the Pauli matrices $\{I,X,Y,Z\}$, there may be tensor products of Paulis which also leave $\mathcal{R}$ unchanged under conjugation.

These two types of symmetry in fact have useful group structure. A permutation of the qubit indices can be described by a bijection from the set $[n]$ to itself. Define a permutation group $G$ to be a group formed by a set of these permutations. The qubits in an orthogonality operator can be permuted by any $\pi\in G$ via conjugation by the qubit permutation matrix
\begin{align} \label{eq:P-pi}
    P_\pi &=\sum_{j_1,\ldots,j_n\in \{0,1\}}\ketbra{j_{\pi^{-1}(1)}\ldots j_{\pi^{-1}(n)}}{j_1\ldots j_n},
\end{align}
where the $\ket{j_1\ldots j_n}$ are $Z$-eigenbasis states. The function $P$, which maps a permutation to its permutation matrix, is a faithful unitary representation of $G$ on $\mathcal{H}$. Since we are interested in unitaries under which $\mathcal{R}$ is invariant, it will be useful to define the subgroup $\mathcal{G}$ of $\unit$ which consists of the permutation matrices $P_\pi$ for all $\pi\in G$. Since $\mathcal{G}$ is the image of $G$ under $P$, we can write $\mathcal{G} = P(G)$; note that $G$ and $\mathcal{G}$ are isomorphic because $P$ is faithful. 

Like qubit permutations, the tensor products of Paulis can also form a group. Specifically, we use the notation $\mathcal{S}$ to denote a subgroup of the Pauli group $\mathcal{P}_n$ on $n$ qubits. $\mathcal{P}_n$ consists of all $n$-fold tensor products of $\{I,X,Y,Z\}$ with multiplicative factors $\pm 1$ and $\pm i$, that is, the set of all operators
\begin{align}\label{eq:pauli-def}
    D = e^{i\phi}\bigotimes_{j=1}^n D_j
\end{align}
where $\phi\in\{0,\frac{\pi}{2},\pi,\frac{3\pi}{2}\}$ and $D_j\in\{I,X,Y,Z\}$.
Since the Pauli matrices are unitary, any $\mathcal{S}$ is a subgroup of $\unit$ as well.

To simplify a given TS problem, we therefore seek a Pauli subgroup $\mathcal{S}$ and permutation group $G$ such that $\mathcal{R}$ is simultaneously invariant under conjugation by both $\mathcal{S}$ and $\mathcal{G}=P(G)$. For any unitary subgroup $K\leq\unit$, let $\mathcal{H}_K$ denote the simultaneous $+1$ eigenspace of all the operators $U\in K$. It subsequently follows from Proposition \ref{prop:U-simp} that if there exist nontrivial $\mathcal{S}$ and $\mathcal{G}$ under which $\mathcal{R}$ is invariant, then fewer criteria are required to determine whether a state in ${\mathcal{H}_{\mathcal{S}}\cap\mathcal{H}_{\mathcal{G}}}$ is a TS state. Note that for $\mathcal{H}_{\mathcal{S}}$ to be nontrivial, $\mathcal{S}$ must be abelian and not contain the operator $-I^{\otimes n}$ \cite{Nielsen_Chuang_2010}.

Observe that if $\mathcal{R}$ is invariant under some $\mathcal{S}$ and $\mathcal{G}$ individually, then $\mathcal{R}$ is also invariant under products of unitaries from either $\mathcal{S}$ or $\mathcal{G}$. Define $\langle\mathcal{S},\mathcal{G}\rangle$ to be the subgroup of $\unit$ whose elements can be written as a product of various $D\in\mathcal{S}$ and $P_\pi\in\mathcal{G}$:
\begin{align}
    \langle\mathcal{S},\mathcal{G}\rangle = \left\{\prod_{j=1}^kU_j\ :\ U_j\in\mathcal{S}\cup\mathcal{G},\  k\in\mathbb{N}\right\},
\end{align}
where $\mathbb{N}$ denotes the set of natural numbers. Then it is readily verified that $\mathcal{R}$ is invariant under conjugation by $\langle\mathcal{S},\mathcal{G}\rangle$ if and only if it is invariant under both $\mathcal{S}$ and $\mathcal{G}$.
Similarly, it easily follows that the space of states invariant under $\langle\mathcal{S},\mathcal{G}\rangle$ equals the space invariant under $\mathcal{S}$ and $\mathcal{G}$, that is, $\mathcal{H}_{\langle\mathcal{S},\mathcal{G}\rangle}=\mathcal{H}_\mathcal{S}\cap\mathcal{H}_\mathcal{G}$. Therefore, two separate symmetries $\mathcal{S}$ and $\mathcal{G}$ of a given TS problem can be equivalently described as a single joint symmetry $\langle\mathcal{S},\mathcal{G}\rangle$.

\begin{figure}[htbp]
\centering
\resizebox{\columnwidth}{!}{
\begin{quantikz}[row sep = 1mm, column sep = 2mm]
    &\permute{2,1,3}\gategroup[4,steps=1,style={line width = 0.1mm,inner xsep = 0mm}]{$P_{\pi_1}$}&&\gate{X}\gategroup[4,steps=1,style={line width = 0.1mm,inner xsep = 0mm}]{$D_1$}&&\permute{1,3,2}\gategroup[4,steps=1,style={line width = 0.1mm,inner xsep = 0mm}]{$P_{\pi_2}$}&&\gategroup[4,steps=1,style={line width = 0.1mm,inner xsep = 0mm}]{$D_2$}&\\
    &&&\gate{X}&&&&&\\
    &&&&&&&\gate{X}&\\
    &&&&&&&\gate{X}&
\end{quantikz}
$=$\ \begin{quantikz}[row sep = 1mm, column sep = 2mm]
    &\permute{3,1,2}\gategroup[4,steps=1,style={line width = 0.1mm,inner xsep = 0mm}]{$P_{\pi_3}$}&&\gate{X}\gategroup[4,steps=1,style={line width = 0.1mm,inner xsep = 0mm}]{$D_3$}&\\
    &&&\ghost{X}&\\
    &&&\ghost{X}&\\
    &&&\gate{X}&
\end{quantikz}
}

\caption{Consolidation of an element of $\langle\mathcal{S},\mathcal{G}\rangle$ into the product of one Pauli operator and one permutation. Let $\mathcal{G}=\langle P_{\pi_1},P_{\pi_2}\rangle$ and $\mathcal{S} = \langle D_1,D_2\rangle$. Observe that $D_2P_{\pi_2}D_1P_{\pi_1}\in \langle\mathcal{S},\mathcal{G}\rangle$ can be rewritten as $D_3P_{\pi_3}$, where $P_{\pi_3}\in \mathcal{G}$ but $D_3\notin\mathcal{S}$. However, $D_3\in\mathcal{S}_\mathcal{G}$ instead.}
\label{fig:network}
\end{figure}
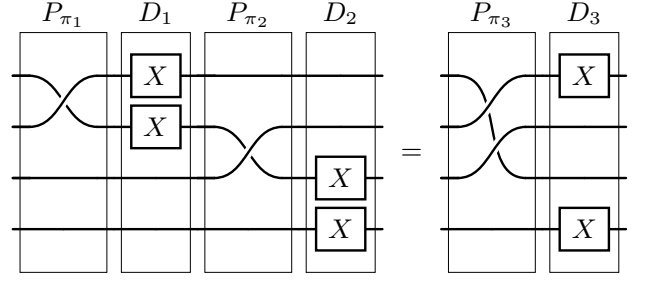

The action of any $U\in\langle\mathcal{S},\mathcal{G}\rangle$ on operators and states can be succinctly characterized by separating the resulting transformation into its local and global components. In particular, note that Pauli operators effect local transformations (due to their tensor product structure) while permutations effect global rearrangements of the underlying qubits. 
Accordingly, each $U\in\langle\mathcal{S},\mathcal{G}\rangle$ can in fact be decomposed simply as the product of only one Pauli operator $D$ and one permutation matrix $P_\pi$; Figure \ref{fig:network} illustrates how such a consolidation is possible. 

However, as seen in Figure \ref{fig:network}, the Pauli operator $D$ obtained from this decomposition need not necessarily belong to $\mathcal{S}$. Intuitively, $D$ must be the product of some Pauli operators in $\mathcal{S}$ whose qubits have been permuted by $G$, but note that these permutations potentially take the constituent Pauli operators out of $\mathcal{S}$. Thus, defining $\mathcal{S}_\mathcal{G}$ to be the subgroup of $\mathcal{P}_n$ generated by
\begin{align}\label{eq:SsubG-def}
    \mathcal{S}_\mathcal{G} =\langle P_{\pi'} D' P_{\pi'}^\dagger\ :\ D'\in \mathcal{S},P_{\pi'}\in\mathcal{G}\rangle, 
\end{align}
it follows that $D\in\mathcal{S}_\mathcal{G}$ even if $D\notin\mathcal{S}$. In general, $\mathcal{S}\subseteq\mathcal{S}_\mathcal{G}$ with equality if $\mathcal{S}$ is closed under permutations of the qubits by $G$, or equivalently, if $\mathcal{G}$ normalizes $\mathcal{S}$. Hence, whether $D$ indeed belongs to $\mathcal{S}$ depends on whether $\mathcal{G}$ normalizes $\mathcal{S}$. We subsequently define the \textit{normalizer} $\mathcal{N}(\mathcal{S})$ of a given $\mathcal{S}$ to be the set of $U\in\unit$ such that $UD'U^\dagger\in\mathcal{S}$ for all $D'\in\mathcal{S}$.

The following proposition then summarizes these results:
\begin{proposition}\label{prop:SG-decomp}
    The following are true for any $\mathcal{S}\leq \mathcal{P}_n$ and permutation matrix group $\mathcal{G}$:
    \begin{enumerate}[label=(\alph*), ref=\ref{prop:SG-decomp}\alph*]
        \item $\mathcal{G} \subseteq \mathcal{N}(\mathcal{S}_\mathcal{G})$.\label{prop:SG-decomp1}
        \item Each element of $\langle\mathcal{S},\mathcal{G}\rangle$ can be uniquely expressed as $DP_\pi$ for some $D\in\mathcal{S}_\mathcal{G}$ and $P_\pi\in\mathcal{G}$. Furthermore, $\langle\mathcal{S},\mathcal{G}\rangle = \mathcal{S}_\mathcal{G}\mathcal{G}$, where 
        \begin{align}
            \mathcal{S}_\mathcal{G}\mathcal{G}=\{DP_\pi:D\in\mathcal{S}_\mathcal{G}, P_\pi\in \mathcal{G}\}.
        \end{align}\label{prop:SG-decomp2}
        \item If $\mathcal{G}\subseteq \mathcal{N}(\mathcal{S})$, then $\mathcal{S}_\mathcal{G} = \mathcal{S}$.\label{prop:SG-decomp3}
    \end{enumerate}
\end{proposition}
\begin{proof}
    Part (a). Write any $D\in\mathcal{S}_\mathcal{G}$ as $D=\prod_{j=1}^kP_j D_j P_j^\dagger$ for some $P_j\in\mathcal{G}$, $D_j\in S$, and $k\in\mathbb{N}$. Then for any $P_\pi\in \mathcal{G}$, we have $P_\pi D P_\pi^\dagger = P_\pi\left(\prod_{j=1}^kP_j D_j P_j^\dagger\right) P_\pi^\dagger= \prod_{j=1}^kP_\pi \left(P_j D_j P_j^\dagger\right)P_\pi^\dagger = \prod_{j=1}^kP'_j D_j P'^\dagger_j$ where $P'_j=P_\pi P_j\in\mathcal{G}$. It follows that $P_\pi D P_\pi^\dagger\in\mathcal{S}_\mathcal{G}$, so $\mathcal{G}\subseteq\mathcal{N}(\mathcal{S}_\mathcal{G})$.

    Part (b). Since $\mathcal{S}\subseteq \mathcal{S}_\mathcal{G}$, any $U\in \langle\mathcal{S},\mathcal{G}\rangle$ can be written as a product of elements from $\mathcal{S}_\mathcal{G}$ and $\mathcal{G}$. Hence, $U$ can be written as $U=\prod_{j=1}^kU_j$ for some $U_j\in\mathcal{S}_\mathcal{G}\cup\mathcal{G}$ and $k\in\mathbb{N}$. For every $P_\pi\in\mathcal{G}$ and $D\in\mathcal{S}_\mathcal{G}$, part (b) implies that $P_\pi D=  D'P_\pi $ for some $D'\in\mathcal{S}_\mathcal{G}$. Repeatedly using this rule, the factors of $U$ can be reordered to write $U=\left(\prod_{j=1}^{l}D_{j}\right)\left(\prod_{j'=1}^{k-l}P_{j'}\right)$ for some $D_{j}\in\mathcal{S}_\mathcal{G}$, $P_{j'}\in\mathcal{G}$, and $l\leq k$. It follows that $U$ can be expressed as the product of one element of $\mathcal{S}_\mathcal{G}$ and one element of $\mathcal{G}$. The uniqueness of this expression follows from the fact that $\mathcal{S}_\mathcal{G}\cap\mathcal{G}$ contains only the identity matrix.
    
    It remains to show that $\langle\mathcal{S},\mathcal{G}\rangle= \mathcal{S}_\mathcal{G}\mathcal{G}$. The inclusion $\langle\mathcal{S},\mathcal{G}\rangle\subseteq \mathcal{S}_\mathcal{G}\mathcal{G}$ follows immediately from above. Conversely, since each element of $\mathcal{S}_\mathcal{G}$ is a product of elements from $\mathcal{S}$ and $\mathcal{G}$, we have $\langle\mathcal{S},\mathcal{G}\rangle\supseteq\mathcal{S}_\mathcal{G}\mathcal{G}$, which implies the result.

    Part (c). Assume $\mathcal{G}\subseteq \mathcal{N}(\mathcal{S})$. Then $P_\pi D P_\pi^\dagger \in \mathcal{S}$ for any $D\in\mathcal{S}$ and $P_\pi\in G$. It follows from Eq. (\ref{eq:SsubG-def}) that $\mathcal{S}_\mathcal{G}=\mathcal{S}$.
\end{proof}

Hence, any $U\in\langle\mathcal{S},\mathcal{G}\rangle$ can in general be written as the product of one Pauli operator in $\mathcal{S}_\mathcal{G}$ and one permutation matrix in $\mathcal{G}$. We refer to $\mathcal{S}_\mathcal{G}$ as a \textit{$\mathcal{G}$-invariant Pauli subgroup} since $\mathcal{S}_\mathcal{G}$ is invariant under conjugation by $\mathcal{G}$ due to Proposition \ref{prop:SG-decomp1} above. Furthermore, in the special case that $\mathcal{G}\subseteq\mathcal{N}(\mathcal{S})$, $U$ can be written simply as the product of one element of $\mathcal{S}$ and one element of $\mathcal{G}$.

Thus far, we have established the general joint symmetry group $\langle\mathcal{S},\mathcal{G}\rangle$ of $\mathcal{R}$ and suggested that this symmetry should reduce the number of criteria required to verify a TS state in $\mathcal{H}_\mathcal{S}\cap\mathcal{H}_\mathcal{G}$. However, it remains unclear exactly which criteria become redundant when this symmetry is applied. Generally speaking, we showed in Proposition \ref{prop:U-simp} that two equations from Eq. (\ref{eq:ortho}) become equivalent if:
\begin{enumerate}
    \item A symmetry operation transforms the orthogonality operator of one equation into that of the other.
    \item The sensor state of interest is invariant under the same symmetry.
\end{enumerate}
Therefore, in the next subsection, we develop tools to concretely describe how orthogonality operators and sensor states transform under a given $\langle\mathcal{S},\mathcal{G}\rangle$.

\subsubsection{Induced transformations of trajectory pairs and bit-strings}\label{sec:index-transform}

The general TS problem summarized by Eq. (\ref{eq:ortho}) involves two fundamental mathematical objects: orthogonality operators in $\unit$ and sensor states $\ket{\psi}$ in $\mathcal{H}$. Each of these operators and states is each associated with a simpler, more abstract index object. For example, each orthogonality operator $\rbf$ is indexed by an ordered pair $(T,T')$ of trajectories. The indices of all $\rbf\in\mathcal{R}$ thus constitute $\mathcal{T}^2$, the set of all pairs of allowed trajectories. Note that the map from $\mathcal{T}^2$ to $\mathcal{R}$ is not injective; for example, even if $T\neq T'$, then we still have $\mathbf{R}^{(T,T)} = \mathbf{R}^{(T',T')} = I^{\otimes n}$. On the other hand, any sensor state $\ket{\psi}$ can be written as a superposition of $Z$-eigenbasis vectors $\ket{j_1\ldots j_n}$, which are indexed by bit-strings $j_1\ldots j_n$. We use the symbol $\mathbb{Z}_2^n=\{0,1\}^n$ to denote the set of all length-$n$ bit-strings, and it is obvious that a bijection exists between $\mathbb{Z}_2^n$ and the $Z$-eigenbasis.

Many transformations of orthogonality operators and $Z$-eigenbasis states can equivalently be understood as transformations of their associated index objects. Suppose there exists a unitary $U$ which transforms each orthogonality operator and basis state into another such that for any $T_1,T_1'\subseteq[n]$ and $j_1\ldots j_n\in \mathbb{Z}_2^n$,
\begin{align}
    U\mathbf{R}^{(T_1,T_1')} U^\dagger&= \mathbf{R}^{(T_2,T_2')}\ \mathrm{and}\nonumber\\
    U\ket{j_1\ldots j_n } &= \ket{j'_1\ldots j'_n}\label{eq:transform}
\end{align}
for some $T_2,T_2'\subseteq[n]$ and $j'_1\ldots j'_n\in \mathbb{Z}_2^n$. It follows that $U$ also induces a transformation of the indices, that is, the application of $U$ maps $(T_1,T_1')\rightarrow(T_2,T_2')$ and $j_1\ldots j_n\rightarrow j'_1\ldots j'_n$. Subsequently, we may rewrite Eq. (\ref{eq:transform}) as 
\begin{align}
    U\mathbf{R}^{(T_1,T_1')} U^\dagger&= \mathbf{R}^{u(T_1,T_1')}\ \mathrm{and}\nonumber\\
    U\ket{j_1\ldots j_n } &= \ket{u(j_1\ldots j_n)},\label{eq:transform2}
\end{align}
where $u(\cdot)$ is some function of the indices satisfying $u(T_1,T_1') = (T_2,T_2')$ and $u(j_1\ldots j_n) = j'_1\ldots j'_n$.

This observation suggests a duality between groups which act naturally on the Hilbert space $\mathcal{H}$ and groups which act naturally on sets of indices. Suppose there exists a group of unitary transformations on $\mathcal{H}$ such that every $U$ in the group satisfies Eq. (\ref{eq:transform2}) for some function $u(\cdot)$ of the indices. We then expect that the various $u$ might also form a group, and we furthermore anticipate that the map from each $U$ to its corresponding $u$ may constitute a group homomorphism.

\begin{center}
\begin{table*}[t]
\begin{tabular}{|c|c||c|c||c|}
\hline
\multicolumn{2}{|c||}{\textbf{Object associated with $\mathcal{H}$}} & \multicolumn{2}{c||}{\textbf{Index object}} & \textbf{Mapping} 
\\
\hline\hline
trajectory perturbation &$R^{(T)}$ &trajectory & $T$  & \\
\hline
perturbation set & $\{R^{(T)}:T\in\mathcal{T}\}$& trajectory set &$\mathcal{T}$ & $\leftrightarrow$\\
\hline
orthogonality operator &$\mathbf{R}^{(T,T')}$ & trajectory pair &$(T,T')$ & \\
\hline
orthogonality operator set & $\mathcal{R}$ & trajectory pair set & $\mathcal{T}^2$& $\leftarrow$\\
\hline
$Z$-eigenbasis state & $\ket{j_1\ldots j_n}$ &bit-string&$j_1\ldots j_n$& \\
\hline
$Z$-eigenbasis & $\{\ket{j_1\ldots j_n} : j_1\ldots j_n\in\mathbb{Z}_2^n\}$ & bit-string set & $\mathbb{Z}_2^n$ & $\leftrightarrow$\\
\hline
permutation matrix & $P_\pi$& permutation & $\pi$ & \\
\hline
permutation matrix group & $\mathcal{G}$ & permutation group & $G$ & $\leftrightarrow$ (isomorphism $P$)\\
\hline
Pauli operator & $D$ & qubit swap & $\sigma_D$ & \\
\hline
Pauli subgroup &$\mathcal{S}$ & swap group & $S$ & $\rightarrow$ (homomorphism $\sigma$)\\
\hline
$\mathcal{G}$-invariant Pauli subgroup &$\mathcal{S}_\mathcal{G}$ & $G$-invariant swap group & $S_G$ & $\rightarrow$ (homomorphism $\sigma$)\\
\hline
group generated by $\mathcal{S},\mathcal{G}$& $\langle \mathcal{S},\mathcal{G} \rangle=\mathcal{S}_\mathcal{G}\mathcal{G}$ & semidirect product & $S_G\rtimes G$ & $\rightarrow$ (homomorphism $\Phi$)\\
\hline
\end{tabular}
\caption{Mathematical objects associated with the Hilbert space $\mathcal{H}$ and their corresponding index objects. In the ``Mapping" column, an arrow is included if the two objects in that row are both sets; a double-headed arrow ``$\leftrightarrow$" indicates a bijection, while a single-headed arrow ``$\rightarrow$" or ``$\leftarrow$" indicates a (not necessarily injective) surjection. If the two sets are groups, then the mapping is specified to be a homomorphism and/or isomorphism.}
\label{table:indices}
\end{table*}
\end{center}

Table \ref{table:indices} summarizes how each operator, state, and symmetry group relevant for trajectory sensing can be paired with a separate index object. We have already seen how orthogonality operators and basis states are respectively indexed by trajectory pairs and bit-strings. Moreover, by the above duality, we expect that each group which transforms operators and states (i.e., $\mathcal{S}$, $\mathcal{G}$, and $\langle \mathcal{S},\mathcal{G}\rangle$) should correspond to another group which instead realizes the induced transformations of the associated indices. In the remainder of this subsection, we will investigate the structure of these corresponding groups which act on indices.

In general, we are motivated to study orthogonality operators, basis states, and symmetry groups through their corresponding index objects for the following reasons. Firstly, the index objects succinctly consolidate the useful mathematical structure of their parent objects. Transformations of trajectory pairs and bit-strings are therefore generally simpler to analyze than transformations of the corresponding operators and states. Furthermore, note that a TS problem is fully defined by its inputs $n$ and $\mathcal{T}$; consequently, by studying $\mathcal{R}$ through its index set $\mathcal{T}^2$, one can use the given structure of $\mathcal{T}$ to determine which specific symmetries of $\mathcal{R}$ may be available to simplify a particular TS problem.

We can apply this principle to determine how orthogonality operators and basis states transform under permutations of the qubits. Suppose we are given a permutation group $G$, which consists of bijections $\pi$ from $[n]$ to itself. Recall $G$ is isomorphic to a permutation matrix group $\mathcal{G}=P(G)$. Intuitively, if we conjugate $\rbf$ by some permutation matrix $P_\pi\in\mathcal{G}$, then the qubits within each trajectory of the index $(T,T')$ will be permuted by the corresponding $\pi\in G$. Likewise, applying $P_\pi$ to a $Z$-eigenbasis state $\ket{j_1\ldots j_n}$ should permute the bits of the index $j_1\ldots j_n$ by $\pi$. The following proposition then formalizes this intuition. 
\begin{proposition}\label{prop:P-conj}
    Let $G$ be any permutation group on $[n]$. Then the following are true for any $\pi\in G$:
    \begin{enumerate}[label=(\alph*), ref=\ref{prop:P-conj}\alph*]
        \item For any $T,T' \subseteq [n]$,
        \begin{align}\label{eq:P-action}
            P_\pi\rbf P_\pi^\dagger = \mathbf{R}^{\pi(T,T')},
        \end{align}
        where $\pi(T,T')$ is the group action of $G$ on $[n]^2$ defined by
        \begin{align}\label{eq:pi-pair-action}
            \pi(T,T') = (\pi(T),\pi(T'))
        \end{align}
        with $\pi(T) = \{\pi(j)\ :\ j\in T\}$.
        \item For any $Z$-eigenbasis state $\ket{j_1\ldots j_n}$,
        \begin{align}
            P_\pi\ket{j_1\ldots j_n} = \ket{\pi(j_1\ldots j_n)},
        \end{align}
        where $\pi(j_1\ldots j_n)$ is the group action of $G$ on $\mathbb{Z}_2^n$ defined by
        \begin{align}\label{eq:pi-string-action}
            \pi(j_1\ldots j_n) = j_{\pi^{-1}(1)}\ldots j_{\pi^{-1}(n)}.
        \end{align}
    \end{enumerate}
\end{proposition}
\begin{proof}
    Part (a). The left side of Eq. (\ref{eq:P-action}) can be rewritten as 
    \begin{align}
        P_\pi\rbf P^\dagger_\pi &= P_\pi R^{\dagger(T)}R^{(T')} P^\dagger_\pi\nonumber\\
        &=\left(P_\pi R^{\dagger(T)}P^\dagger_\pi\right) \left(P_\pi R^{(T')} P^\dagger_\pi\right).
    \end{align}
    Proposition \ref{prop:P-R-conj} from Appendix \ref{appendix:perm-matrix-conjugation} implies that conjugating $R^{(T')}$ by $P_\pi$ permutes its tensor factors so that $P_\pi R^{(T')}P^\dagger_\pi = R^{(\pi(T'))}$. Thus,
    \begin{align}
        P_\pi\rbf P^\dagger_\pi = R^{\dagger(\pi(T))}R^{(\pi(T'))} = \mathbf{R}^{(\pi(T),\pi(T'))},
    \end{align}
    from which Eq. (\ref{eq:P-action}) follows. Additionally, it is easy to verify that Eq. (\ref{eq:pi-pair-action}) satisfies the group action axioms.

    Part (b). Follows trivially from Eq. (\ref{eq:P-pi}). It is also easy to verify that Eq. (\ref{eq:pi-string-action}) satisfies the group action axioms. 
\end{proof}

It follows that the action of $\mathcal{G}$ on orthogonality operators and states corresponds to the action of $G$ on indices. This proposition supports the anticipated duality between groups which act naturally on the Hilbert space $\mathcal{H}$ (e.g., $\mathcal{G}$) and groups which act naturally on the index set $[n]$ (e.g., $G$). In the case of $\mathcal{G}$ and $G$, this correspondence takes the form of the isomorphism $P$.

On the other hand, it is not immediately obvious whether a Pauli subgroup $\mathcal{S}$ also has a corresponding group on $[n]$ which realizes the induced transformations of the indices. Note that in the tensor product decomposition of $\rbf$, each qubit in $T$ receives an $R_Z$ while each qubit in $T'$ receives an $R_Z^\dagger$. Now let $\mathcal{S}$ be any subgroup of $\mathcal{P}_n$ and consider the conjugation of $\rbf$ by some $D\in\mathcal{S}$. Since $XR_ZX^\dagger = YR_ZY^\dagger = R_Z^\dagger$, if $D$ applies an $X$ or $Y$ to any qubit contained in $T$ or $T'$, then that qubit is effectively swapped between the two trajectories. It follows that the induced group on $[n]$ must transform trajectory pairs by swapping the subset of qubits which receive either $X$ or $Y$. Similarly, we anticipate that the induced group must also act on basis states by bit-flipping those qubits which receive $X$ or $Y$ (up to a phase). Accordingly, define the map $\sigma$ which returns the subset of qubits receiving $X$ or $Y$ for any general $D\in\mathcal{P}_n$:
\begin{align}
    \sigma_D = \{j\ :\ D_j=X\ \mathrm{or}\ D_j = Y\},
\end{align}
where the $D_j$ are the tensor factors of $D$ defined in Eq. (\ref{eq:pauli-def}). Next, let $S$ be the set of $\sigma_D$ for all $D\in\mathcal{S}$; equivalently, $S=\sigma(\mathcal{S})$ is the image of $\mathcal{S}$ under $\sigma$. Note that $S\leq\mathbb{P}([n])$. It is straightforward to verify that $S$ forms a group under the composition law $\symdif$, where $\symdif$ is the set symmetric difference defined by
\begin{align}
    \varsigma\symdif\varsigma' = (\varsigma\cup\varsigma')\setminus (\varsigma\cap\varsigma')
\end{align}
for any $\varsigma,\varsigma'\in S$. Furthermore, it is readily shown that $\sigma$ is a surjective homomorphism from $\mathcal{S}$ to $S$ (see Proposition \ref{prop:sigma-homo} in Appendix \ref{appendix:pauli-action}). This $S$, which we call a \textit{qubit swap group}, is in fact the desired induced group which swaps qubits between the paired trajectories and flips bit-strings, as shown in the following proposition:

\begin{proposition}\label{prop:S-conj}
    Let $\mathcal{S}$ be any subgroup of $\mathcal{P}_n$. Then the following are true for any $D\in\mathcal{S}$:
    \begin{enumerate}[label=(\alph*), ref=\ref{prop:S-conj}\alph*]
        \item For any $T,T'\subseteq[n]$,
        \begin{align}\label{eq:pauli-action}
            D\rbf D^\dagger = \mathbf{R}^{\sigma_D(T,T')},
        \end{align}
        where $\sigma_D(T,T')$ is the group action of $S = \sigma(\mathcal{S})$ on $[n]^2$ defined by
        \begin{align}\label{eq:action-s}
        \sigma_D(T_1,T_1') = (T_2,T_2')
        \end{align}
        with
        \begin{align}
        T_2 &= (T_1\setminus \sigma_D)\cup(T_1'\cap \sigma_D)\ \mathrm{and}\nonumber\\
        T_2' &= (T_1'\setminus \sigma_D)\cup(T_1\cap \sigma_D).
        \end{align}
        \item For any $Z$-eigenbasis state $\ket{j_1\ldots j_n}$,
        \begin{align}
        D\ket{j_1\ldots j_n} = e^{i\phi}\ket{\sigma_D(j_1\ldots j_n)}
        \end{align}
        for some $\phi\in[0,2\pi)$; define $\sigma_D(j_1\ldots j_n)$ to be the group action of $S$ on $\mathbb{Z}_2^n$ given by
        \begin{align}\label{eq:bit-flip-action}
            \sigma_D(j_1\ldots j_n) = 
            j'_1\ldots j'_n,
        \end{align}
        where
        \begin{align}\label{eq:j-prime-indices}
            j'_k = \begin{cases}
                j_k\oplus 1 &k\in\sigma_D\\
                j_k &k\notin\sigma_D
            \end{cases}
        \end{align}
        for $k=1,\ldots,n$. The $\oplus$ symbol indicates addition modulo 2.
    \end{enumerate}
    
\end{proposition}
\begin{proof}
    See Appendix \ref{appendix:pauli-action}.
\end{proof}
In contrast to $\mathcal{G}$ and $G$, the groups $\mathcal{S}$ and $S$ are not necessarily isomorphic. The potential non-injectivity of the homomorphism $\sigma$ is due to the fact that conjugation by $X$ or $Y$ has identical effect on the operators $R_Z$ and $R^\dagger_Z$. For example, conjugating any $\rbf$ by either $X^{\otimes n}$ or $Y^{\otimes n}$ produces an identical result. Consequently, there may be multiple Paulis in $\mathcal{S}$ which identically transform all $\rbf$ operators.

Evidently, applying products of permutations and Paulis from a given $\langle\mathcal{S},\mathcal{G}\rangle$ to orthogonality operators and states also induces transformations of the indices, so we again describe these induced transformations with a corresponding group on $[n]$. Recall the equality $\langle\mathcal{S},\mathcal{G}\rangle=\mathcal{S}_\mathcal{G}\mathcal{G}$ from Proposition \ref{prop:SG-decomp2}, which provides a convenient way to write each element of $\langle\mathcal{S},\mathcal{G}\rangle$ as the product of one Pauli operator and one permutation matrix. Now define $S_G=\sigma(\mathcal{S}_\mathcal{G})$ to be the image of $\mathcal{S}_\mathcal{G}$ under the map $\sigma$. Equivalently, $S_G$ can be defined directly in terms of $S = \sigma(\mathcal{S})$ and $G = P^{-1}(\mathcal{G})$ as the subgroup of $\mathbb{P}([n])$ generated by
\begin{align}\label{eq:SG-equiv-def}
    S_G = \langle \pi(\varsigma)\ :\ \varsigma\in S, \pi\in G\rangle
\end{align}
with the composition law $\symdif$, where $\pi(\varsigma)=\{\pi(j):j\in\varsigma\}$ (see Proposition \ref{prop:equiv-def} of Appendix \ref{appendix:semidirect}). Noting the existing correspondences $\mathcal{S}_\mathcal{G}\rightarrow S_G$ and $\mathcal{G}\rightarrow G$, we expect that the group on $[n]$ associated with $\langle\mathcal{S},\mathcal{G}\rangle$ should be some product of $S_G$ and $G$. In fact, the desired group is the semidirect product $S_G\rtimes G$, defined as the set of pairs $(\varsigma,\pi)$ for all $\varsigma\in S_G$ and $\pi\in G$ together with the composition law
\begin{align}\label{eq:SG-comp-law}
    (\varsigma,\pi)\cdot(\varsigma',\pi') = (\varsigma\symdif\pi(\varsigma'),\pi\pi').
\end{align}
Proposition \ref{prop:sigma-norm} asserts that $\pi(\cdot)$ is an automorphism of $S_G$ for all $\pi\in G$, which guarantees that $S_G\rtimes G$ is a well-defined semidirect product; additionally, due to this proposition, we call $S_G$ a \textit{G-invariant qubit swap group}.

Because there exist homomorphisms $\sigma$ from $\mathcal{S}_\mathcal{G}$ to $S_G$ and $P^{-1}$ from $\mathcal{G}$ to $G$, there also exists a natural homomorphism $\Phi$ from $\langle\mathcal{S},\mathcal{G}\rangle=\mathcal{S}_\mathcal{G}\mathcal{G}$ to $S_G\rtimes G$. In particular, define 
\begin{align}\label{eq:phi-def}
    \Phi(U) = (\sigma_D,\pi)
\end{align}
for any $U\in\langle\mathcal{S},\mathcal{G}\rangle$, where $D\in\mathcal{S}_\mathcal{G}$ and $\pi\in G$ are the unique elements of $\mathcal{S}_\mathcal{G}$ and $G$ satisfying $U=DP_\pi$. The uniqueness of $D$ and $P_\pi$ (guaranteed by Proposition \ref{prop:SG-decomp2}) ensures that $\Phi$ is well-defined. In Proposition \ref{prop:phi-homo} from Appendix \ref{appendix:semidirect}, we verify that $\Phi$ is indeed a surjective homomorphism. The following proposition then shows how transformations of orthogonality operators and basis states under $\langle\mathcal{S},\mathcal{G}\rangle$ correspond to transformations of trajectory pairs and bit-strings under $S_G\rtimes G$:

\begin{proposition}\label{prop:PS-conj}
    Given some permutation group $G$ on $[n]$ and $\mathcal{S}\leq\mathcal{P}_n$, let $\mathcal{G} = P(G)$ and $S=\sigma(\mathcal{S})$. Then the following are true for any $U \in \langle\mathcal{S},\mathcal{G}\rangle$:
    \begin{enumerate}[label=(\alph*), ref=\ref{prop:PS-conj}\alph*]
        \item For any $T,T'\subseteq[n]$,
        \begin{align}\label{eq:SG-action}
        U \rbf U^\dagger = \mathbf{R}^{(\varsigma,\pi)[(T,T')]},
        \end{align}
        where $(\varsigma,\pi) = \Phi(U)$ and $(\varsigma,\pi)[(T,T')]$ is the group action of $S_G\rtimes G$ on $[n]^2$ defined by
        \begin{align}\label{eq:sympauli-action}
            (\varsigma,\pi)[(T,T')] = \varsigma[\pi(T,T')].
        \end{align}
        \item For any $Z$-eigenbasis state $\ket{j_1\ldots j_n}$,
        \begin{align}\label{eq:SG-action-state}
            U\ket{j_1\ldots j_n} = e^{i\phi}\ket{(\varsigma,\pi)(j_1\ldots j_n)}
        \end{align}
        for some $\phi\in[0,2\pi)$; define $(\varsigma,\pi)(j_1\ldots j_n)$ to be the group action of $S_G\rtimes G$ on $\mathbb{Z}_2^n$ given by
        \begin{align}\label{eq:sympauli-action-state}
            (\varsigma,\pi)(j_1\ldots j_n) = \varsigma[\pi(j_1\ldots j_n)].
        \end{align}
    \end{enumerate}
\end{proposition}
\begin{proof}
    See Appendix \ref{appendix:semidirect}.
\end{proof}

The above proposition provides a useful way to view symmetries of $\mathcal{R}$ instead as symmetries of $\mathcal{T}^2$. In particular, it is easy to see that $\mathcal{R}$ is invariant under $\langle\mathcal{S},\mathcal{G}\rangle$ if $\mathcal{T}^2$ is invariant under $S_G\rtimes G$. 
Recall that if $\mathcal{R}$ is invariant under $\langle\mathcal{S},\mathcal{G}\rangle$, then the criteria for verifying TS states in $\mathcal{H}_\mathcal{S}\cap\mathcal{H}_\mathcal{G}$ may be simplified through Proposition \ref{prop:U-simp}. It then follows that such simplification is equivalently guaranteed if $\mathcal{T}^2$ is invariant under $S_G\rtimes G$.

\subsubsection{Redundant TS criteria}\label{sec:redundant}

Equipped with tools to describe transformations of orthogonality operators in terms of transformations of their indices, we are prepared to explicitly determine which TS state criteria of Eq. (\ref{eq:ortho}) become redundant when there exist $G$ and $\mathcal{S}$ such that $\mathcal{T}^2$ is invariant under the corresponding $S_G\rtimes G$. Proposition \ref{prop:U-simp} asserts that for two criteria of Eq. (\ref{eq:ortho}) to become redundant, their corresponding orthogonality operators must be equivalent up to conjugation by some symmetry operator. We now argue that the action of $S_G\rtimes G$ on trajectory pair indices conveniently determines whether two $\rbf$ operators are equivalent under $\langle\mathcal{S},\mathcal{G}\rangle$.

In particular, two orthogonality operators are equivalent up to conjugation by some element of $\langle\mathcal{S},\mathcal{G}\rangle$ if their indices belong to the same orbit of $\mathcal{T}^2$ under the action of $S_G\rtimes G$. Define the set of such orbits to be
\begin{align}
     \mathcal{T}^2/(S_G\rtimes G) = \left \{\operatorname{Orb}_{S_G\rtimes G}[(T,T')]\ :\ T,T'\in\mathcal{T} \rule{0pt}{2.4ex} \right \},
\end{align}
where
\begin{align}
    \operatorname{Orb}_{S_G\rtimes G}[(T,T')] = \left \{(\varsigma,\pi)[(T,T')]\ :\ (\varsigma,\pi)\in S_G\rtimes G \rule{0pt}{2.4ex} \right \}.
\end{align}
To understand this claim, suppose two trajectory pairs in $\mathcal{T}^2$ belong to the same orbit. Then there exists some $(\varsigma,\pi)\in S_G\rtimes G$ mapping one pair to the other. Since the homomorphism $\Phi(\cdot)$ is surjective, there exists some $U\in\langle\mathcal{S},\mathcal{G}\rangle$ such that $\Phi(U) = (\varsigma,\pi)$. It then follows from Proposition \ref{prop:PS-conj} that the orthogonality operators corresponding to these two trajectory pairs are equivalent up to conjugation by $U$.

The following lemma then explicitly specifies the reduced criteria required to verify whether a state in $\mathcal{H}_{\mathcal{S}}\cap\mathcal{H}_{\mathcal{G}}$ is a TS state:
\begin{theorem}\label{lemma:SG-simp}
    Given a permutation group $G$ on $[n]$, Pauli subgroup $\mathcal{S}\leq \mathcal{P}_n$, and trajectory set $\mathcal{T}\subseteq\mathbb{P}([n])$, let $\mathcal{G} = P(G)$ and $S = \sigma(\mathcal{S})$. Suppose $\mathcal{T}^2$ is invariant under the action of $S$ and $G$. Then, at a particular value of $\theta\in[0,\pi]$, a state $\ket{\psi}\in\mathcal{H}_\mathcal{S}\cap\mathcal{H}_\mathcal{G}$ is a TS state if and only if for every orbit $\Omega\in\mathcal{T}^2/(S_G\rtimes G)$, $\ket{\psi}$ satisfies Eq. (\ref{eq:ortho}) for at least one representative $(T,T')\in\Omega$.
\end{theorem}
\begin{proof}
    First note by Proposition \ref{prop:pairs-sym} of Appendix \ref{appendix:pairs-sym} that $\mathcal{T}^2$ is invariant under $S$ and $G$ if and only if $\mathcal{T}^2$ is invariant under $S_G\rtimes G$. Hence, the orbits $\mathcal{T}^2/(S_G\rtimes G)$ are well-defined. Furthermore, observe that when $\theta=0$, the theorem becomes trivial since $\ket{\psi}$ is a TS state if and only if $\abs{\mathcal{T}}=1$; we thus henceforth assume that $\theta\neq 0$.
    
    ``$\implies$" direction: If $\ket{\psi}$ is a TS state, then Eq. (\ref{eq:ortho}) is satisfied for every $(T,T')\in\mathcal{T}^2$. It follows trivially that Eq. (\ref{eq:ortho}) is satisfied for at least one $(T,T')$ in each orbit. 
    
    ``$\impliedby$" direction: Choose $\Omega$ to be any orbit in $\mathcal{T}^2/(S_G\rtimes G)$, and assume Eq. (\ref{eq:ortho}) is satisfied for some $(T_1,T_1')\in \Omega$. Then for any $(T_2,T_2')\in \Omega$, there exists some $(\varsigma,\pi)\in S_G\rtimes G$ such that $ (T_2,T_2')= (\varsigma,\pi)[(T_1,T_1')]$ by the definition of an orbit.
    Due to Proposition \ref{prop:PS-conj} and the surjectivity of $\Phi(\cdot)$, there exists some $U\in\langle\mathcal{S},\mathcal{G}\rangle$ such that Eq. (\ref{eq:U-conj}) holds. Additionally, since $\ket{\psi}\in\mathcal{H}_{\mathcal{S}}\cap\mathcal{H}_{\mathcal{G}}$ and $\mathcal{H}_{\mathcal{S}}\cap\mathcal{H}_{\mathcal{G}}=\mathcal{H}_{\langle \mathcal{S},\mathcal{G}\rangle}$, we have $U\ket{\psi}=\ket{\psi}$. It then follows from Proposition \ref{prop:U-simp} that Eq. (\ref{eq:ortho}) must also be satisfied for $(T_2,T_2')$. Thus, if Eq. (\ref{eq:ortho}) holds for one trajectory pair in an orbit, then it also holds for all pairs in the orbit. Since the orbits partition $\mathcal{T}^2$, if Eq. (\ref{eq:ortho}) is satisfied for one pair per orbit, then it is satisfied for all pairs in $\mathcal{T}^2$, which implies $\ket{\psi}$ is a TS state.
    
\end{proof}

Although this theorem implies that the TS criteria may be simplified when the orthogonality operators and sensor state are invariant under Paulis and permutation matrices, two important issues remain unresolved. Firstly, given a particular TS problem, it is not obvious how to choose suitable Pauli and permutation groups which lead to meaningful simplification; the given structure of $\mathcal{T}$ must somehow be utilized to determine which specific symmetries may be useful. Secondly, even though this lemma assists in the search for specifically Pauli- and permutation-invariant TS states, we are ultimately interested in whether \textit{any} TS state exists for a particular value of $\theta$. Hence, we must consider the possibility of TS states which satisfy no symmetry property. We address both of these issues in the next section.

\section{Simplified criteria for the existence of general TS states}\label{sec:simp-criteria}

In this section, we will expand on the result of Theorem \ref{lemma:SG-simp}, deriving simplified criteria to determine whether any general TS state exists for a particular $\theta$, given $n$ and $\mathcal{T}$. Recall that in Section \ref{sec:formalism}, we presented the following naive approach to solving a TS problem: substitute the general ansatz $\ket{\psi} = \sum_{j=0}^{2^n-1}a_j\ket{{j}}$ into the system of Eq. (\ref{eq:ortho}) at a particular $\theta$ and solve for $a_j$, since a solution implies the existence of a TS state. Note that this ansatz permits the sensor state to be any state in $\mathcal{H}$ with no restrictions of symmetry. The resulting system of equations has $\dim \mathcal{H} = 2^n$ variables and $\abs{\mathcal{T}}^2$ equations, and is generally computationally intractable to solve for large $n$ or $\mathcal{T}$. 

Hypothetically, assume now for a given $\mathcal{T}$ that it is always possible to find a TS state which obeys some permutation and Pauli symmetries, provided that a TS state exists at all. Then, a TS state exists if and only if one can be found in the subspace invariant under these symmetries. As a result, fewer variables are needed to specify the TS state ansatz used in the system, and these symmetries may also allow the number of equations to be reduced via Theorem \ref{lemma:SG-simp}. Thus, if this assumption is shown to be true for a given $\mathcal{T}$, then the general TS state existence criteria of Eq. (\ref{eq:ortho}) may be substantially simplified.

In Section \ref{sec:existence}, we confirm that this assumption holds for a broad family of trajectory sets $\mathcal{T}$; namely, there exist nontrivial permutation and Pauli stabilizer groups such that one can always find a TS state invariant under both groups, given that a TS state exists at all. Subsequently, in Section \ref{sec:var}, we show that the existence of a general TS state can be determined by solving a simplified system containing fewer than $2^n$ variables and $\abs{\mathcal{T}}^2$ equations. We then show that this system can be equivalently presented as a linear programming feasibility problem.

\subsection{Existence of symmetry-invariant TS states}\label{sec:existence}

We now investigate how to choose a suitable permutation group $G$ and Pauli stabilizer group $\mathcal{S}$ which lead to meaningful simplification. For a given $\mathcal{T}$, recall that a particular $G$ and $\mathcal{S}$ lead to simplification via Theorem \ref{lemma:SG-simp} if (1) $\mathcal{T}^2$ is invariant under $S$ and $G$ and (2) the prospective sensor state is invariant under $\mathcal{S}$ and $\mathcal{G}$, where $\mathcal{G}=P(G)$ and $S = \sigma(\mathcal{S})$. We address the former requirement first, thereby seeking $G$ and $\mathcal{S}$ such that $\mathcal{T}^2$ is invariant under both $S$ and $G$.

Permutation invariance emerges naturally for a broad family of trajectory sets $\mathcal{T}$. We will call a trajectory set $\mathcal{T}$ transitive under $G$, or \textit{$G$-transitive}, if 
\begin{align}\label{eq:G-transitive}
    \mathcal{T} = \left \{\pi(T_0)\ :\ \pi\in G 
    \rule{0pt}{2.4ex}
    \right \}
\end{align}
for some \textit{generator trajectory} $T_0\subseteq[n]$ and permutation group $G$ on $[n]$. In other words, $\mathcal{T}$ is $G$-transitive if it is the orbit of some generator trajectory under $G$. Because each $\pi\in G$ is a bijection, every trajectory in such a $\mathcal{T}$ is the same size. This construction is therefore useful because it provides a succinct way to describe a set of equally-sized trajectories in terms of a single generator and a permutation group $G$. Furthermore, $\mathcal{T}^2$ is invariant under $G$ if $\mathcal{T}$ is $G$-transitive, as desired:
\begin{proposition}\label{prop:G-transitive}
    If $\mathcal{T}\subseteq\mathbb{P}([n])$ is $G$-transitive for any permutation group $G$ on $[n]$, then $\mathcal{T}^2$ is invariant under $G$.
\end{proposition}
\begin{proof}
    Assume $\mathcal{T}$ is $G$-transitive, and pick any $(T,T')\in\mathcal{T}^2$. Hence, we can write $T$ and $T'$ in terms of a generator trajectory $T_0\subseteq[n]$ as follows: $T=\pi(T_0)$ and $T'=\pi'(T_0)$. Then for any $\pi''\in G$, we have 
    \begin{align}
        \pi''(T,T') &= (\pi''(T),\pi''(T'))\nonumber\\
        &= (\pi''\pi(T_0), \pi''\pi'(T_0)).
    \end{align}
    Since $G$ is a group, $\pi''\pi$ and $\pi''\pi'$ are in $G$. It follows that $\pi''\pi(T_0)$ and $\pi''\pi'(T_0)$ are in $\mathcal{T}$, which implies that $\pi''(T,T')\in\mathcal{T}^2$. We conclude that $\mathcal{T}^2$ is invariant under $G$.
\end{proof}

Additionally, for \textit{any} choice of $\mathcal{T}$ and $G$, there exists a nontrivial $\mathcal{S}$ such that $\mathcal{T}^2$ is guaranteed to be invariant under $S = \sigma(\mathcal{S})$. In particular, the group $\mathcal{S}=\{I^{\otimes n},X^{\otimes n}\}$ is satisfactory:
\begin{proposition}\label{prop:S-bitflip}
    Let $G$ be any permutation group on $[n]$ and let $\mathcal{S} = \{I^{\otimes n}, X^{\otimes n}\}$. Furthermore, let $\mathcal{G} = P(G)$ and $S = \sigma(\mathcal{S})$. Then the following are true:
    \begin{enumerate}[label=(\alph*), ref=\ref{prop:S-bitflip}\alph*]
        \item $\mathcal{G}\subseteq\mathcal{N}(\mathcal{S})$. In fact, every element of $\mathcal{G}$ commutes with every element of $\mathcal{S}$.\label{prop:S-bitflip1}
        \item $S_G=S = \{\varnothing, [n]\}$.\label{prop:S-bitflip2}
        \item $\mathcal{T}^2$ is invariant under $S$ for any $\mathcal{T}\subseteq\mathbb{P}([n])$.\label{prop:S-bitflip3}
    \end{enumerate}
\end{proposition}
\begin{proof}
    Part (a). Clearly $P_\pi I^{\otimes n} P_\pi^\dagger = I^{\otimes n}$ for any $\pi\in G$. Additionally, due to Proposition \ref{prop:permute-factors} of Appendix \ref{appendix:perm-matrix-conjugation},  $P_\pi X^{\otimes n} P_\pi^\dagger = X^{\otimes n}$ for any $\pi\in G$ as well. It follows that every element of $\mathcal{G}$ commutes with every element of $\mathcal{S}$, which also implies that $\mathcal{G}\subseteq\mathcal{N}(\mathcal{S})$.

    Part (b). Because $\mathcal{G}\subseteq\mathcal{N}(\mathcal{S})$, we have $\mathcal{S}_\mathcal{G} = \mathcal{S}$ by Proposition \ref{prop:SG-decomp3}. Thus, $S_G = \sigma(\mathcal{S}_\mathcal{G}) = \sigma(\mathcal{S}) = S$. Additionally, it readily follows from the definition of $\sigma$ that $S = \{\varnothing,[n]\}$.
    
    Part (c). Pick any $(T,T')\in\mathcal{T}^2$. The non-identity element $[n]$ of $S$ swaps \textit{all} of the qubits between $T$ and $T'$ such that $[n](T,T')=(T',T)$. Clearly, $(T',T)\in\mathcal{T}^2$, so $\mathcal{T}^2$ is invariant under $S$.
\end{proof}

When $\mathcal{S}=\{I^{\otimes n},X^{\otimes n}\}$, we emphasize that the single non-identity element of $S=\{\varnothing,[n]\}$ acts on trajectory pairs by swapping the positions of each trajectory, i.e. $(T,T')\leftrightarrow(T',T)$. It is worthwhile to examine this result in greater detail. Note that $\sigma\left(X^{\otimes n}\right) = [n]$. Then, by Proposition \ref{prop:S-conj}, we expect that conjugating any $\rbf$ operator by $X^{\otimes n}$ should yield $\mathbf{R}^{\sigma\left(X^{\otimes n}\right)(T,T')}= \mathbf{R}^{[n](T,T')}=\mathbf{R}^{(T',T)}$. We can now explicitly verify that this is indeed the case:
\begin{align}
    X^{\otimes n}\rbf X^{\otimes n} &= X^{\otimes n}R^{\dagger(T)}R^{(T')} X^{\otimes n}\nonumber\\
    &= \left(X^{\otimes n}R^{\dagger(T)}X^{\otimes n}\right)\left(X^{\otimes n}R^{(T')} X^{\otimes n}\right)\nonumber\\
    &= R^{(T)}R^{\dagger(T')}\nonumber\\
    &=R^{\dagger(T')}R^{(T)}\nonumber\\ 
    &= \mathbf{R}^{(T',T)}.
\end{align}
Intuitively, conjugating an orthogonality operator by $X^{\otimes n}$ yields its adjoint, and taking the adjoint of an $\rbf$ swaps the positions of $T$ and $T'$.

Provided that $\mathcal{T}$ is $G$-transitive for some permutation group $G$, Propositions \ref{prop:G-transitive} and \ref{prop:S-bitflip3} indicate that the symmetry groups $G$ and $\mathcal{S} = \{I^{\otimes n},X^{\otimes n}\}$ can be invoked within Theorem \ref{lemma:SG-simp} to simplify the search for symmetry-invariant TS states. Recall by Theorem \ref{lemma:SG-simp} that only one orthogonality condition per orbit in the set $\mathcal{T}^2/(S_G\rtimes G)$ is necessary to verify a TS state in $\mathcal{H}_{\mathcal{S}}\cap\mathcal{H}_{\mathcal{G}}=\mathcal{H}_{\langle \mathcal{S},\mathcal{G}\rangle}$. For convenience, we subsequently define the groups $\tilde{G}=S_G\rtimes G$ and $\tilde{\mathcal{G}}=\langle \mathcal{S},\mathcal{G}\rangle$ in the special case where $S=\{\varnothing,[n]\}$ and $\mathcal{S} = \{I^{\otimes n},X^{\otimes n}\}$. It then follows for $G$-transitive $\mathcal{T}$ that the search for TS states in  $\mathcal{H}_{\tilde{\mathcal{G}}}$ requires checking only $\abs{\mathcal{T}^2/\tilde{G}}$ separate conditions.

As a side note, it is useful to express $\tilde{G}$ and $\tilde{\mathcal{G}}$ in simpler and more practical forms. In particular, $\tilde{G}$ can be recast as a direct product:
\begin{align}\label{eq:tilde-G-def}
    \tilde{G} = S\times G,\ \mathrm{where}\ S = \{\varnothing,[n]\},
\end{align}
since $S_G = S$ and also $\pi(\varsigma)=\varsigma$ for all $\varsigma\in\{\varnothing,[n]\}$ and $\pi\in G$. Similarly, $\tilde{\mathcal{G}}$ can be rewritten as
\begin{align}\label{eq:tilde-m-G-def}
    \tilde{\mathcal{G}} = \mathcal{S}\mathcal{G},\ \mathrm{where}\ \mathcal{S} = \{I^{\otimes n},X^{\otimes n}\},
\end{align}
due to Proposition \ref{prop:SG-decomp} and the fact that $\mathcal{G}\subseteq\mathcal{N}(\mathcal{S})$ by Proposition \ref{prop:S-bitflip1}.

Given that the search for TS states in $\mathcal{H}_{\tilde{\mathcal{G}}}$ can be simplified when $\mathcal{T}$ is $G$-transitive, it is insightful to characterize the states in this subspace. Intuitively, the states in $\mathcal{H}_{\tilde{\mathcal{G}}}$ are invariant under bit-flips as well as permutations of the qubits under $G$. Clearly, a state is invariant under $\tilde{\mathcal{G}}$ if and only if it is invariant under both $X^{\otimes n}$ and $\mathcal{G}$; the permutation symmetry is immediately evident. To understand the bit-flip symmetry, observe that
\begin{align}
    X^{\otimes n}\ket{j_1\ldots j_n} &= \ket{\sigma\left(X^{\otimes n}\right)(j_1\ldots j_n)}\nonumber\\
    &= \ket{[n](j_1\ldots j_n)}\label{eq:X-def}
\end{align}
for all $j_1\ldots j_n\in\mathbb{Z}_2^n$, in accordance with Proposition \ref{prop:PS-conj}. Note that the $[n](\cdot)$ operation flips all of the bits in the string $j_1\ldots j_n$ such that $[n](j_1\ldots j_n) = (j_1\oplus 1)\ldots (j_n\oplus 1)$, where $\oplus$ indicates addition modulo 2. It follows that states invariant under $X^{\otimes n}$ exhibit a bit-flip symmetry.

Although the search for TS states in $\mathcal{H}_{\tilde{\mathcal{G}}}$ may be simplified, we remain interested more generally in whether \textit{any} TS state exists in the full Hilbert space $\mathcal{H}$ at a particular value of $\theta$. The following theorem provides a solution to this question. Specifically, the theorem asserts for $G$-transitive $\mathcal{T}$ that, at a particular $\theta$, a TS state exists in $\mathcal{H}$ if and only if a TS state exists in $\mathcal{H}_{\tilde{\mathcal{G}}}$.
\begin{theorem}\label{thm:exist}
    Given a permutation group $G$ on $[n]$, suppose $\mathcal{T}\subseteq\mathbb{P}([n])$ is $G$-transitive, and let $\mathcal{G} = P(G)$. Then a TS state exists at a particular value of $\theta\in[0,\pi]$ if and only if a $\tilde{\mathcal{G}}$-invariant TS state $\ket{\psi_{\tilde{\mathcal{G}}}}\in\mathcal{H}_{\tilde{\mathcal{G}}}$ exists at the same $\theta$. 
\end{theorem}
\begin{proof}
    See Appendix \ref{appendix:exist}.
\end{proof}
Hence, given 
$\mathcal{T}$ is $G$-transitive, it suffices to search only the symmetry-invariant subspace $\mathcal{H}_{\tilde{\mathcal{G}}}$ to determine whether \textit{any} TS state exists in the full space $\mathcal{H}$ at a given value of $\theta$. The situation described by the above theorem is intuitively depicted in Figure \ref{fig:venn}, which illustrates how the space of TS states intersects with the invariant subspaces $\mathcal{H}_\mathcal{S}$ and $\mathcal{H}_\mathcal{G}$ when $\mathcal{S}=\{I^{\otimes n},X^{\otimes n}\}$. Theorem \ref{thm:exist} guarantees that the region marked by ``$\star$" is non-empty if and only if the set of TS states is non-empty. Subsequently, one only needs to search the yellow region $\mathcal{H}_{\tilde{\mathcal{G}}}$ to determine whether a TS state exists at all.
\begin{figure}[htbp]
  \includegraphics[width = 0.6\linewidth]{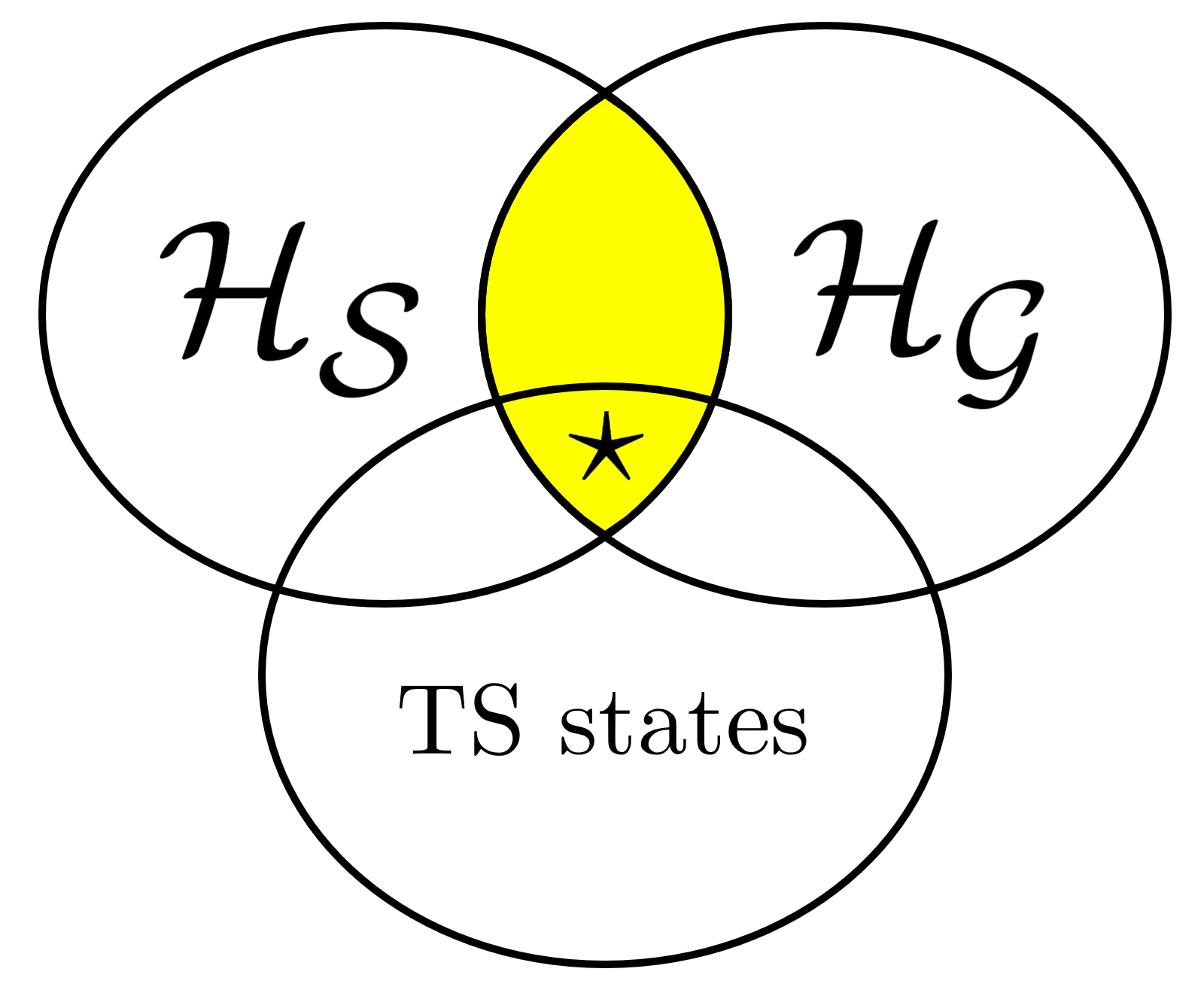}
  \caption{Venn diagram illustrating the space of TS states and the invariant subspaces $\mathcal{H}_\mathcal{S}$ and $\mathcal{H}_\mathcal{G}$ for the Pauli stabilizer $\mathcal{S}=\{I^{\otimes n},X^{\otimes n}\}$ and an arbitrary permutation matrix group $\mathcal{G}$. The yellow region depicts the joint invariant space $\mathcal{H}_{\tilde{\mathcal{G}}}$, and the region labeled with ``$\star$" is the space of $\tilde{\mathcal{G}}$-invariant TS states.} 
  \label{fig:venn}
\end{figure}

Importantly, since Theorem \ref{lemma:SG-simp} can simplify the criteria for $\tilde{\mathcal{G}}$-invariant TS states in the yellow region, Theorem \ref{thm:exist} implies that the general criteria for the existence of \textit{any} TS state can also be simplified, provided that $\mathcal{T}$ is $G$-transitive. In the next section, we formalize these reduced general criteria and demonstrate that they result in a system involving fewer equations and variables. Furthermore, we explain how this reduced system can be cast as a linear programming feasibility problem.

\subsection{Simplified TS state existence criteria as a linear programming problem}\label{sec:var}

In Section \ref{sec:existence} above, we showed for $G$-transitive  $\mathcal{T}$ that the existence of a TS state at a particular value of $\theta$ can be determined by searching only the subspace $\mathcal{H}_{\tilde{\mathcal{G}}}$; hence, we assume $\mathcal{T}$ is $G$-transitive in the remainder of this section. Additionally, in Section \ref{sec:symmetries}, we demonstrated that fewer criteria are required to determine whether a state in $\mathcal{H}_{\tilde{\mathcal{G}}}$ is a TS state. Synthesizing these results, it follows from Theorems \ref{lemma:SG-simp} and \ref{thm:exist} that a TS state exists at a particular value of $\theta$ if and only if there exists some $\ket{\psi_{\tilde{\mathcal{G}}}}\in\mathcal{H}_{\tilde{\mathcal{G}}}$ such that for every orbit $\Omega\in\mathcal{T}^2/\tilde{G}$,
\begin{align}\label{eq:G-inv-ortho}
    \bra{\psi_{\tilde{\mathcal{G}}}}\rbf\ket{\psi_{\tilde{\mathcal{G}}}}=\delta_{T,T'}
\end{align}
for at least one representative $(T,T')\in\Omega$. Thus, to determine whether a TS state exists at a particular $\theta$, we can substitute a $\tilde{\mathcal{G}}$-invariant ansatz state into Eq. (\ref{eq:G-inv-ortho}) and solve the resulting system of equations. 

In general, this reduced system involves fewer equations and variables than the naive system obtained by substituting a general ansatz directly into Eq. (\ref{eq:ortho}). In particular, Theorem \ref{lemma:SG-simp} removes those equations which become redundant due to symmetry, while Theorem \ref{thm:exist} decreases the number of variables by restricting the sensor state ansatz to a subspace of $\mathcal{H}$. Requiring fewer equations and variables, this reduced system is generally much easier to solve.

However, one practical issue arises when attempting to substitute a $\tilde{\mathcal{G}}$-invariant ansatz state $\ket{\psi_{\tilde{\mathcal{G}}}}$ into Eq. (\ref{eq:G-inv-ortho}). Note that Eq. (\ref{eq:G-inv-ortho}) requires taking the inner product of $\ket{\psi_{\tilde{\mathcal{G}}}}$ and $\rbf\ket{\psi_{\tilde{\mathcal{G}}}}$; however, while the former state resides in $\mathcal{H}_{\tilde{\mathcal{G}}}$, the latter generally does not. This discrepancy leads to computational inefficiency because both states must first be expanded in a basis for the full space $\mathcal{H}$ before their product can be evaluated. 

To circumvent this issue, we replace the operator $\rbf$ in Eq. (\ref{eq:G-inv-ortho}) with another equivalent operator which does not take states out of $\mathcal{H}_{\tilde{\mathcal{G}}}$. To construct this operator, we introduce the  projector $\Pi_{\tilde{\mathcal{G}}}$ onto the space $\mathcal{H}_{\tilde{\mathcal{G}}}$ given by
\begin{align}
    \Pi_{\tilde{\mathcal{G}}} &= \frac{1}{\abs{\tilde{\mathcal{G}}}}\sum_{U\in\tilde{\mathcal{G}}} U \nonumber\\
    &=\frac{1}{2\abs{G}}\left(I^{\otimes n}+X^{\otimes n}\right)\sum_{\pi\in G}P_\pi,\label{eq:HG-proj}
\end{align}
where the second equality follows from the fact that $\tilde{\mathcal{G}}=\mathcal{S}\mathcal{G}$. 
Then, given an ordered pair of trajectories $(T,T')\in\nmg^2$, define 
\begin{align}
    \mathbf{R}^{(T,T')}_{\tilde{\mathcal{G}}}(\theta) = \Pi_{\tilde{\mathcal{G}}} \rbf(\theta)\Pi_{\tilde{\mathcal{G}}},
\end{align}
which we refer to as a \textit{$\tilde{\mathcal{G}}$-invariant orthogonality operator}. Evidently, if the sensor state $\ket{\psi_{\tilde{\mathcal{G}}}}$ is $\tilde{\mathcal{G}}$-invariant, the $\rbf$ operator in Eq. (\ref{eq:G-inv-ortho}) can be replaced with $\rbf_{\tilde{\mathcal{G}}}$ since $\bra{\psi_{\tilde{\mathcal{G}}}} \mathbf{R}^{(T,T')}_{\tilde{\mathcal{G}}} \ket{\psi_{\tilde{\mathcal{G}}}} =\bra{\psi_{\tilde{\mathcal{G}}}} \mathbf{R}^{(T,T')} \ket{\psi_{\tilde{\mathcal{G}}}}$. Furthermore, it is easy to see that $\mathbf{R}^{(T,T')}_{\tilde{\mathcal{G}}}$ leaves states within the invariant subspace $\mathcal{H}_{\tilde{\mathcal{G}}}$, as desired.

These $\tilde{\mathcal{G}}$-invariant orthogonality operators have the nice property that two such operators are equal if their trajectory pair indices belong to the same orbit under $\tilde{G}$. To understand this claim, suppose $(T_2,T_2') = (\varsigma,\pi)[(T_1,T_1')]$ for some $(\varsigma,\pi)\in\tilde{G}$ and $T_1,T_1',T_2,T_2'\in\mathcal{T}$. Then by Proposition \ref{prop:PS-conj} and the surjectivity of the map $\Phi:\tilde{\mathcal{G}}\rightarrow\tilde{G}$, there exists $U\in\tilde{\mathcal{G}}$ such that $\mathbf{R}^{(T_2,T_2')}=U \mathbf{R}^{(T_1,T_1')} U^\dagger$. Hence,
\begin{align}
    \mathbf{R}^{(T_2,T_2')}_{\tilde{\mathcal{G}}}&= \Pi_{\tilde{\mathcal{G}}}\mathbf{R}^{(T_2,T_2')}\Pi_{\tilde{\mathcal{G}}}\nonumber\\
    &= \Pi_{\tilde{\mathcal{G}}}U \mathbf{R}^{(T_1,T_1')} U^\dagger\Pi_{\tilde{\mathcal{G}}}\nonumber\\
    &= \Pi_{\tilde{\mathcal{G}}}\mathbf{R}^{(T_1,T_1')} \Pi_{\tilde{\mathcal{G}}}\nonumber\\
    &=\mathbf{R}^{(T_1,T_1')}_{\tilde{\mathcal{G}}},
\end{align}
as desired.

Consequently, it is most useful to index $\tilde{\mathcal{G}}$-invariant orthogonality operators with orbits from the set $\mathcal{T}^2/\tilde{G}$ rather than with individual trajectory pairs. Defining $M_{\tilde{G}} = \abs{\mathcal{T}^2/\tilde{G}}$ to be the number of such orbits, it is thus helpful to assign a unique integer $\mu=0,1,\ldots M_{\tilde{G}}-1$ to each $\Omega\in\mathcal{T}^2/\tilde{G}$ such that $\Omega_\mu$ represents the orbit associated with the integer $\mu$. By convention, we choose $\Omega_0=\operatorname{Orb}_{\tilde{G}}[(T,T)]$, where $T$ can be chosen to be any trajectory in $\mathcal{T}$. Then for all $\mu$, define $\mathbf{R}_{\tilde{\mathcal{G}}}^{(\mu)}$ to be the operator equal to $\mathbf{R}_{\tilde{\mathcal{G}}}^{(T,T')}$ for any equivalent choice of $(T,T')\in\Omega_\mu$. The following lemma then formally restates the reduced criteria of Eq. (\ref{eq:G-inv-ortho}), but where the $\rbf$ operators have been replaced with these $\tilde{\mathcal{G}}$-invariant orthogonality operators:
\begin{lemma}\label{thm:red-criteria}
     Given a permutation group $G$ on $[n]$, suppose $\mathcal{T}\subseteq\mathbb{P}([n])$ is $G$-transitive, and let $\mathcal{G} = P(G)$. For a given $\theta\in[0,\pi]$, a TS state exists if and only if there exists a state $\ket{\psi_{\tilde{\mathcal{G}}}}\in\mathcal{H}_{\tilde{\mathcal{G}}}$ such that 
     \begin{align}\label{eq:red-criteria}
         \bra{\psi_{\tilde{\mathcal{G}}}} \mathbf{R}^{(\mu)}_{\tilde{\mathcal{G}}} \ket{\psi_{\tilde{\mathcal{G}}}} = \delta_{\mu,0}
     \end{align}
     for all $\mu=0,\ldots,M_{\tilde{G}}-1$, where $\delta_{\mu,0}$ equals $1$ if $\mu=0$ and zero otherwise. Furthermore, any $\ket{\psi_{\tilde{\mathcal{G}}}}\in\mathcal{H}_{\tilde{\mathcal{G}}}$ satisfying Eq. (\ref{eq:red-criteria}) is a TS state. 
\end{lemma}
\begin{proof}
    We first make the following observation: for any $\mu=0,\ldots, M_{\tilde{G}}-1$ and $(T,T')\in\Omega_\mu$, 
    \begin{align}
        \bra{\psi_{\tilde{\mathcal{G}}}} \mathbf{R}^{(\mu)}_{\tilde{\mathcal{G}}} \ket{\psi_{\tilde{\mathcal{G}}}} &=\bra{\psi_{\tilde{\mathcal{G}}}} \mathbf{R}^{(T,T')}_{\tilde{\mathcal{G}}} \ket{\psi_{\tilde{\mathcal{G}}}}\nonumber\\
        &=\bra{\psi_{\tilde{\mathcal{G}}}} \Pi_{\tilde{\mathcal{G}}}\mathbf{R}^{(T,T')}\Pi_{\tilde{\mathcal{G}}} \ket{\psi_{\tilde{\mathcal{G}}}}\nonumber\\
        &=\bra{\psi_{\tilde{\mathcal{G}}}} \mathbf{R}^{(T,T')}\ket{\psi_{\tilde{\mathcal{G}}}}
    \end{align}
    and $\delta_{\mu,0}=\delta_{T,T'}$ by convention.
    
    ``$\implies$" direction: If a TS state exists, then Theorem \ref{thm:exist} guarantees that there exists a $\tilde{\mathcal{G}}$-invariant TS state $\ket{\psi_{\tilde{\mathcal{G}}}}\in \mathcal{H}_{\tilde{\mathcal{G}}}$. For all $\mu=0,\ldots,M_{\tilde{G}}-1$, Theorem \ref{lemma:SG-simp} then implies that $\ket{\psi_{\tilde{\mathcal{G}}}}$ satisfies Eq. (\ref{eq:ortho}) for at least one $(T,T')\in\Omega_\mu$. The result then follows from the above observation.

    ``$\impliedby$" direction: If some $\ket{\psi_{\tilde{\mathcal{G}}}}\in \mathcal{H}_{\tilde{\mathcal{G}}}$ satisfies Eq. (\ref{eq:red-criteria}) for all $\mu=0,\ldots,M_{\tilde{G}}-1$, then the above observation implies that Eq. (\ref{eq:ortho}) holds for at least one trajectory pair per orbit. Consequently, by Theorem \ref{lemma:SG-simp}, $\ket{\psi_{\tilde{\mathcal{G}}}}$ is a TS state.   
\end{proof}

Using this Lemma, we now derive the explicit reduced system obtained by substituting a $\tilde{\mathcal{G}}$-invariant ansatz into Eq. (\ref{eq:red-criteria}). Concretely, this ansatz will be a superposition of basis vectors of $\mathcal{H}_{\tilde{\mathcal{G}}}$; therefore, it is necessary to find a basis for $\mathcal{H}_{\tilde{\mathcal{G}}}$.

We will determine the basis vectors of $\mathcal{H}_{\tilde{\mathcal{G}}}$ by projecting the $Z$-eigenbasis vectors onto $\mathcal{H}_{\tilde{\mathcal{G}}}$. Note that the projector $\Pi_{\tilde{\mathcal{G}}}$ maps any $Z$-eigenbasis vector $\ket{j_1\ldots j_n}$ to the equal superposition of $U\ket{j_1\ldots j_n}$ for all $U\in\tilde{\mathcal{G}}$. Due to Eq. (\ref{eq:SG-action-state}) of Proposition \ref{prop:PS-conj}, the various $U\ket{j_1\ldots j_n}$ must be $Z$-eigenbasis vectors as well, up to a phase. In fact, this phase is trivial for all $U\in\tilde{\mathcal{G}}$ such that  
\begin{align}\label{eq:tilde-mathcal-G-action}
    U\ket{j_1\ldots j_n} = \ket{(\varsigma,\pi)(j_1\ldots j_n)},
\end{align}
where $(\varsigma,\pi)=\Phi(U)$ is an element of $\tilde{G}$. To prove Eq. (\ref{eq:tilde-mathcal-G-action}), note that any $U\in\tilde{\mathcal{G}}$ can be written as $DP_\pi$ for some unique $D\in\{I^{\otimes n},X^{\otimes n}\}$ and $\pi\in G$ due to Proposition \ref{prop:SG-decomp2}. By Proposition \ref{prop:P-conj}, $U\ket{j_1\ldots j_n} =D\ket{\pi(j_1\ldots j_n)}$. Eq. (\ref{eq:X-def}) then implies the desired result.

It follows that projection of some $\ket{j_1 \ldots j_n}$ onto $\mathcal{H}_{\tilde{\mathcal{G}}}$ can be expressed as an equal superposition of $Z$-eigenbasis vectors whose indices belong to the orbit of $j_1\ldots j_n$ under the action of $\tilde{G}$. Thus, let $\mathbb{Z}_2^n/\tilde{G}$ be the set of orbits of bit-strings under the action of $\tilde{G}$: more concretely,
\begin{align}
    \mathbb{Z}_2^n/\tilde{G} = \left \{\operatorname{Orb}_{\tilde{G}}[j_1\ldots j_n]\ :\ j_1\ldots j_n\in\mathbb{Z}_2^n \rule{0pt}{2.4ex} \right \},
\end{align}
where 
\begin{align}
    \operatorname{Orb}_{\tilde{G}}[j_1\ldots j_n] = \left \{(g,\pi)[j_1\ldots j_n]\ :\ (g,\pi)\in\tilde{G} \rule{0pt}{2.4ex} \right \}.
\end{align}
Denote the number of orbits in $\mathbb{Z}_2^n/\tilde{G}$ with $N_{\tilde{G}}$, that is, $N_{\tilde{G}} = \abs{\mathbb{Z}_2^n/\tilde{G}}$. It is again useful to assign to each orbit $\omega\in\mathbb{Z}_2^n/\tilde{G}$ a unique integer $\nu=0,1,\ldots,N_{\tilde{G}}-1$ such that $\omega_\nu$ represents the orbit associated with integer $\nu$. The following proposition then provides a basis for $\mathcal{H}_{\tilde{\mathcal{G}}}$ in terms of these orbits.
\begin{proposition}\label{prop:HG-basis}
    Given a permutation group $G$ on $[n]$, let $\mathcal{G} = P(G)$. An unnormalized orthogonal basis for $\mathcal{H}_{\tilde{\mathcal{G}}}$ is given by
    \begin{align}\label{eq:G-basis-vecs}
        \left\{\ket{\overline{\nu}}=\sum_{j_1\ldots j_n\in {\omega_\nu}}\ket{{j_1\ldots j_n}}\ :\ \nu=0,\ldots,N_{\tilde{G}}-1\right\}
    \end{align}
    where the $\omega_\nu\in \mathbb{Z}_2^n/\tilde{G}$ are sets of bit-strings.
\end{proposition}
\begin{proof}
    See Appendix \ref{appendix:HGbasis}.
\end{proof}

Using this basis for $\mathcal{H}_{\tilde{\mathcal{G}}}$, a suitable $\tilde{\mathcal{G}}$-invariant ansatz $\ket{\psi_{\tilde{\mathcal{G}}}}$ is
\begin{align}\label{eq:G-ansatz}
    \ket{\psi_{\tilde{\mathcal{G}}}} = \sum_{\nu=0}^{N_{\tilde{G}}-1}b_\nu\ket{\overline{\nu}}.
\end{align}
Clearly, any state in $\mathcal{H}_{\tilde{\mathcal{G}}}$ can be expressed in the form of this ansatz. Substituting this ansatz into Eq. (\ref{eq:red-criteria}), it follows from Lemma \ref{thm:red-criteria} that a TS state exists if and only if the resulting system of $M_{\tilde{G}}$ equations in $N_{\tilde{G}}$ variables admits a solution of $b_\nu\in\mathbb{C}$. 

To expand Eq. (\ref{eq:red-criteria}) into this system, it remains to determine how $\tilde{\mathcal{G}}$-invariant orthogonality operators act on the basis states of $\mathcal{H}_{\tilde{\mathcal{G}}}$. Conveniently, the basis states $\ket{\overline{\nu}}$ of Eq. (\ref{eq:G-basis-vecs}) are eigenvectors of the $\tilde{\mathcal{G}}$-invariant orthogonality operators:
\begin{proposition}\label{prop:RG-props}
    Given a permutation group $G$ on $[n]$ and a $G$-transitive trajectory set $\mathcal{T}$, the following are true for any $\mu=0,\ldots M_{\tilde{G}}-1$:
    \begin{enumerate}[label=(\alph*), ref=\ref{prop:RG-props}\alph*]
        \item $\mathbf{R}^{(\mu)}_{\tilde{\mathcal{G}}}(\theta)$ is Hermitian.
        \item For all $\nu=0,\ldots,N_{\tilde{G}}-1$,
            \begin{align}\label{eq:R-eigval}
                \mathbf{R}^{(\mu)}_{\tilde{\mathcal{G}}}(\theta)\ket{\overline{\nu}} = \lambda_{\mu,\nu}(\theta)\ket{\overline{\nu}},
            \end{align}
             where $\lambda_{\mu,\nu}\in\mathbb{R}$ is a real eigenvalue (given in Appendix \ref{appendix:R-prop}).
    \end{enumerate}
\end{proposition}
\begin{proof}
    See Appendix \ref{appendix:R-prop}.
\end{proof}

We now derive the explicit system of Eq. (\ref{eq:red-criteria}) in terms of the TS ansatz coefficients $b_\nu$. Substituting the $\tilde{\mathcal{G}}$-invariant ansatz of Eq. (\ref{eq:G-ansatz}) into Eq. (\ref{eq:red-criteria}) and evaluating the inner products yields a system of equations which are each linear in the various squared ansatz coefficients $\abs{b_{\nu}}^2$. Letting $c_\nu = \abs{b_{_\nu}}^2$, the original system in the variables $b_{\nu}\in\mathbb{C}$ can be recast as a system in $c_\nu\in\mathbb{R}$ along with the nonnegativity constraints $c_\nu\geq 0$, per the following theorem:
\begin{theorem}\label{thm:lin-prog}
    Given a permutation group $G$ on $[n]$, suppose $\mathcal{T}\subseteq\mathbb{P}([n])$ is $G$-transitive. Let $A(\theta)$ be the $M_{\tilde{G}}\times N_{\tilde{G}}$ real matrix such that
    \begin{align}\label{eq:A-entries-def}
        A_{\mu,\nu}(\theta)=\lambda_{\mu,\nu}(\theta)\abs{\omega_\nu},
    \end{align}
    where $\lambda_{\mu,\nu}$ is the eigenvalue given in Eq. (\ref{eq:R-eigval}) and $\omega_\nu$ is the $\nu$th orbit of $\mathbb{Z}_2^n/\tilde{G}$. Then for a given $\theta\in[0,\pi]$, a TS state exists if and only if there exists $\mathbf{c}\in\mathbb{R}^{N_{\tilde{G}}}$ such that
    \begin{align}\label{eq:lin-prog}
        A(\theta)\mathbf{c}=\mathbf{d}\ \mathrm{and}\ \mathbf{c}\geq 0,
    \end{align}
    where $\mathbf{d}\in\mathbb{R}^{M_{\tilde{G}}}$ is the vector with entries $d_\mu = \delta_{\mu,0}$.
    Furthermore, for any $\mathbf{c}$ solving Eq. (\ref{eq:lin-prog}), $\ket{\psi}=\sum_{\nu=0}^{N_{\tilde{G}}-1}\sqrt{c_\nu}\ket{\overline{\nu}}$ is a satisfactory TS state.
\end{theorem}
\begin{proof}
    Note by convention that $A$, $\mathbf{c}$, and $\mathbf{d}$ will be zero-indexed. ``$\impliedby$'' direction: Suppose there exists $\mathbf{c}\in\mathbb{R}^{N_{\tilde{G}}}$ satisfying Eq. (\ref{eq:lin-prog}). Then let $\ket{\psi_{\tilde{\mathcal{G}}}} = \sum_{\nu=0}^{N_{\tilde{G}}-1}\sqrt{c_\nu}\ket{\overline{\nu}}$. By Proposition \ref{prop:HG-basis}, $\ket{\psi_{\tilde{\mathcal{G}}}}\in\mathcal{H}_{\tilde{\mathcal{G}}}$. Additionally, $\ket{\psi_{\tilde{\mathcal{G}}}}$ satisfies Eq. (\ref{eq:red-criteria}) for all $\mu=0,\ldots,M_{\tilde{G}}-1$:
    \begin{align}
        \bra{\psi_{\tilde{\mathcal{G}}}}\mathbf{R}_{\tilde{\mathcal{G}}}^{(\mu)}\ket{\psi_{\tilde{\mathcal{G}}}} &= \bra{\psi_{\tilde{\mathcal{G}}}}\left(\sum_{\nu} \lambda_{\mu,\nu} \sqrt{c_\nu} \ket{\overline{\nu}}\right)\nonumber\\
        &= \sum_{\nu} \lambda_{\mu,\nu} c_\nu\braket{\overline{\nu}|\overline{\nu}}\nonumber\\
        &= \sum_{\nu} \lambda_{\mu,\nu} c_\nu\abs{\omega_\nu}\nonumber\\
         &= \sum_{\nu} A_{\mu,\nu} c_\nu=d_\mu = \delta_{\mu,0},
    \end{align}
    where the first equality follows from Proposition \ref{prop:RG-props}. By Lemma \ref{thm:red-criteria}, $\ket{\psi_{\tilde{G}}}$ is a TS state. 

    ``$\implies$'' direction: Suppose there exists a TS state. By Lemma \ref{thm:red-criteria}, there exists a $\tilde{\mathcal{G}}$-invariant TS state $\ket{\psi_{\tilde{\mathcal{G}}}}$ satisfying Eq. (\ref{eq:red-criteria}) for all $\mu=0,\ldots,M_{\tilde{G}}-1$; by Proposition \ref{prop:HG-basis}, this state can be written as $\ket{\psi_{\tilde{\mathcal{G}}}} = \sum_{\nu=0}^{N_{\tilde{G}}-1}b_\nu\ket{\overline{\nu}}$ for some $b_\nu\in \mathbb{C}$. Now let $c_\nu = \abs{b_\nu}^2$. Then 
    \begin{align}
        \sum_{\nu} A_{\mu,\nu} c_\nu&= \sum_{\nu} \lambda_{\mu,\nu} \abs{b_\nu}^2\abs{\omega_\nu}= \bra{\psi_{\tilde{\mathcal{G}}}}\mathbf{R}_{\tilde{\mathcal{G}}}^{(\mu)}\ket{\psi_{\tilde{\mathcal{G}}}}
    \end{align}
    by Proposition \ref{prop:RG-props}. Since $\ket{\psi_{\tilde{\mathcal{G}}}}$ satisfies Eq. (\ref{eq:red-criteria}), it follows that 
    \begin{align}
         \sum_{\nu} A_{\mu,\nu} c_\nu=\delta_{\mu,0} = d_\mu,
    \end{align}
    or equivalently, $A\mathbf{c}=\mathbf{d}$. Note also that $\abs{b_\nu}^2\geq 0$ implies that $\mathbf{c}\geq 0$.
\end{proof}

The system in Eq. (\ref{eq:lin-prog}) defines a linear programming feasibility problem which seeks a solution vector $\mathbf{c}$ inside the feasible region given by Eq. (\ref{eq:lin-prog}). Theorem \ref{thm:lin-prog} then asserts that TS states exist if and only if this linear program is feasible. Furthermore, the proof of this theorem shows that any solution to the program gives the coefficients for a satisfactory TS state; that is, if $\mathbf{c}$ is a solution, then $\ket{\psi_{\tilde{\mathcal{G}}}} = \sum_\nu\sqrt{c_\nu}\ket{\overline{\nu}}$ is a TS state. Therefore, if a TS state exists, it can be constructed using pre-existing linear programming algorithms \cite{linprog} upon determining the matrix $A(\theta)$, which depends on $\mathcal{T}$ and consequently on the structure of the group $G$.

Lastly, we rigorously demonstrate that the system of Theorem \ref{thm:lin-prog} involves fewer equations and variables than the naive system obtained by substituting a general ansatz directly into Eq. (\ref{eq:ortho}). Recall that this naive system contains $\abs{\mathcal{T}}^2$ equations in $2^n$ variables, while the reduced system of Eq. (\ref{eq:lin-prog}) contains $M_{\tilde{G}}$ equations in $N_{\tilde{G}}$ variables. The following proposition then asserts that $M_{\tilde{G}}<\abs{\mathcal{T}}^2$ and $N_{\tilde{G}}<2^n$:
\begin{proposition}\label{prop:MNbounds}
    The following upper bounds hold on the size of $M_{\tilde{G}}$ and $N_{\tilde{G}}$ for any permutation group $G$ on $[n]$, assuming $\mathcal{T}$ is $G$-transitive:
    \begin{align}
       M_{\tilde{G}} \leq \abs{\nmg}\ \mathrm{and}\ N_{\tilde{G}}\leq 2^{n-1}
    \end{align}
\end{proposition}
\begin{proof}
    See Appendix \ref{appendix:MNbounds}. For completeness, lower bounds on $M_{\tilde{G}}$ and $N_{\tilde{G}}$ are provided in the same appendix.
\end{proof}
Typically, the upper bounds given by Proposition \ref{prop:MNbounds} are very loose, since they do not utilize any specific structure of the permutation group $G$. Consequently, $M_{\tilde{G}}$ and $N_{\tilde{G}}$ are generally much smaller than $\abs{\mathcal{T}}^2$ and $2^n$, respectively. For example, suppose $n$ is even and that $\mathcal{T}$ contains all trajectories of size $n/2$. Then $\abs{\mathcal{T}}^2$ grows asymptotically as $O(4^n)$. Noting that $\mathcal{T}$ is transitive under the symmetric group $\Sigma_n$ on $[n]$, we can compute $M_{\tilde{G}}=N_{\tilde{G}} = n/2+1$ for $G=\Sigma_n$. It follows that the numbers of equations and variables in Eq. (\ref{eq:lin-prog}) are only linear in $n$ as opposed to exponential. Hence, the system in Eq. (\ref{eq:lin-prog}) contains exponentially fewer equations and variables than the naive alternative under these conditions.

Theorem \ref{thm:lin-prog} provides a complete procedure for solving the TS problem when $\mathcal{T}$ is $G$-transitive for some permutation group $G$ on $[n]$. Recall that the inputs of a TS problem are $n$ and $\mathcal{T}$. To deploy this theorem, one must first compute the orbit sets $\mathcal{T}^2/\tilde{G}$ and $\mathbb{Z}_2^n/\tilde{G}$. From these orbit sets, the entries of the matrix $A(\theta)$ can then be calculated. Solving the TS problem then amounts to checking the feasibility of the linear program in Eq. (\ref{eq:lin-prog}) for every $\theta\in[0,\pi]$. In the next section, we use the results obtained above to solve the TS problem for arbitrary $n$ when $\mathcal{T}$ is transitive under symmetric and cyclic permutation groups.

\section{Main results for $G$-transitive trajectory sets}\label{sec:results}
In this section, we apply the machinery developed in Sections \ref{sec:theory} and \ref{sec:simp-criteria} to solve two general families of TS problems. Recall that a TS problem takes the total number of qubits $n$ and trajectory set $\mathcal{T}$ as inputs and asks for which interaction strengths $\theta\in[0,\pi]$ there exists a TS state yielding mutually orthogonal trajectory output states. We refer to this set of $\theta$ for which a TS state exists as the set of achievable $\theta$. Often, we are interested in the infinum of this set, which we call the \textit{minimum achievable $\theta$}, since this quantity represents the smallest particle-sensor interaction strength for which perfect single-shot trajectory discrimination is possible. Thus, in Sections \ref{sec:sym-group} and \ref{sec:cyc}, we derive bounds on the minimum achievable $\theta$ in terms of $n$ when $\mathcal{T}$ is transitive under the symmetric and cyclic permutation groups, respectively. The proofs of these bounds are constructive and provide explicit descriptions of TS states which exist at each $\theta$. 

Supposing that a given TS problem involves a $G$-transitive $\mathcal{T}$, the inputs of the problem can be equivalently redefined. Note that a $G$-transitive $\mathcal{T}$ is parameterized in terms of a permutation group $G$ and single generator trajectory $T_0$, per Eq. (\ref{eq:G-transitive}). In this section, we will thus choose $\mathcal{T}$ to be
\begin{align}
    \mathcal{T}_G(n,m) = \left \{\pi([m])\ :\ \pi\in G \rule{0pt}{2.4ex}  \right\}
\end{align}
for some given permutation group $G$ on $[n]$ and integer $m\leq n$, where $[m]=\{1,\ldots m\}$ is the generator trajectory. Because every trajectory in $\mathcal{T}_G(n,m)$ is a permutation of $[m]$, all trajectories contain exactly $m$ qubits. Consequently, $\mathcal{T}_G(n,m)$ is completely parameterized by $n$, $m$, and $G$, from which it follows that the TS problems considered here are fully defined by the three inputs $n$, $m$, and $G$.

Hence, we will assume in Sections \ref{sec:sym-group} and \ref{sec:cyc} that $\mathcal{T}=\mathcal{T}_G(n,m)$ and that $n,m$, and $G$ are given parameters. We first apply Theorem \ref{thm:lin-prog} to derive bounds on the minimum achievable $\theta$ when $n$ and $m$ are arbitrary and $G$ is the symmetric group. 

\subsection{$\mathcal{T}$ is transitive under the symmetric group}\label{sec:sym-group}

If $G$ is chosen to be the symmetric group $\Sigma_n$ on the set $[n]$, then the trajectory set $\mathcal{T}_G(n,m)$ includes all possible size-$m$ trajectories. We use the symbol $\Tsym$ to represent $\mathcal{T}_G(n,m)$ when $G=\Sigma_n$; Figure \ref{fig:mathcalT} then illustrates $\Tsym$ for small values of $n$ and $m$. Note that for a sensor of $n$ qubits, any trajectory set $\mathcal{T}'$ containing only size-$m$ trajectories must be a subset of $\Tsym$. Consequently, any TS state which can discriminate all trajectories in $\Tsym$ can also discriminate the trajectories in any such $\mathcal{T}'$. It follows that a solution to the TS problem for $\Tsym$ automatically generalizes to every TS problem involving equally-sized trajectories. 
\begin{figure}[htbp]
  \includegraphics[width = \linewidth]{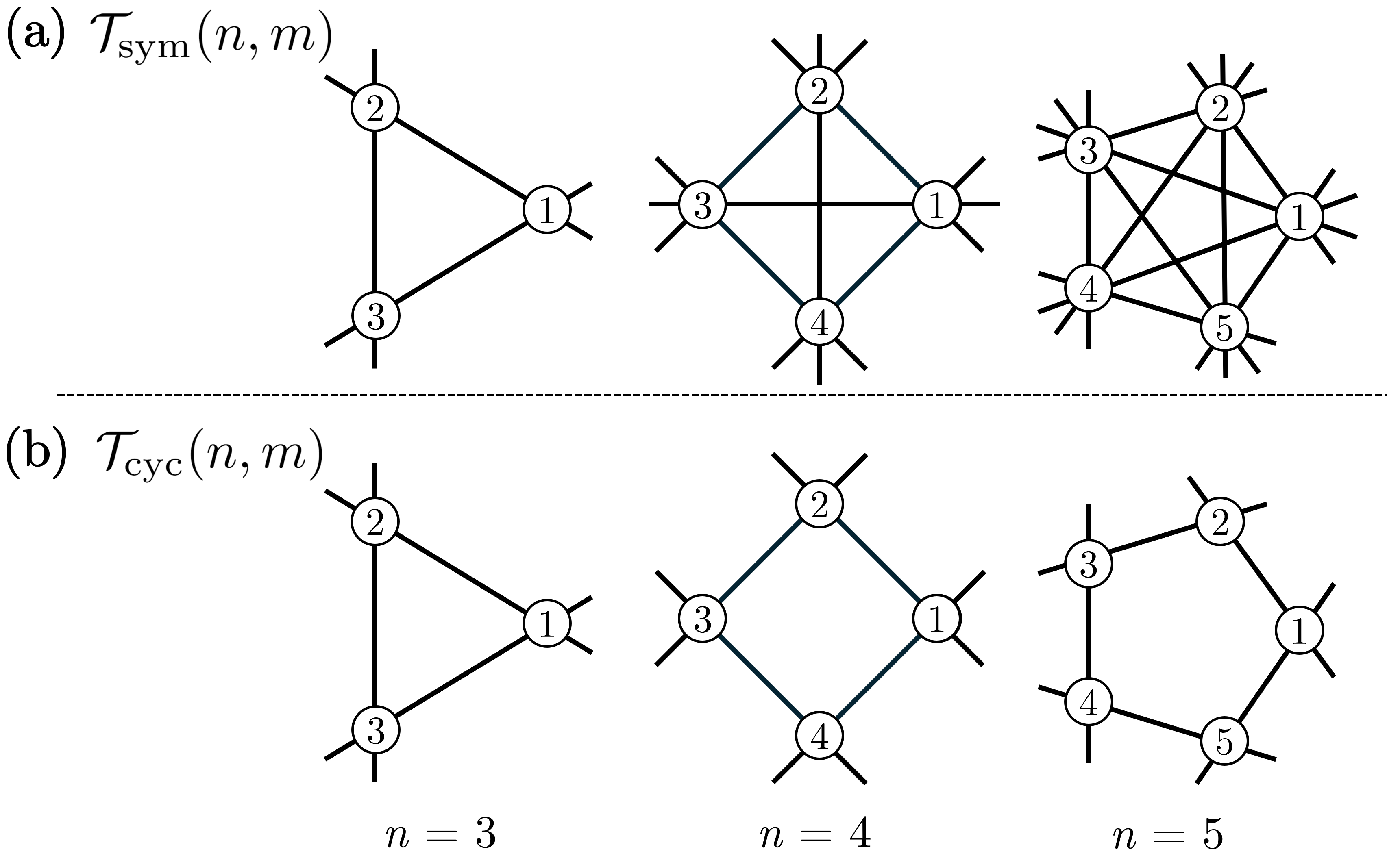}
  \caption{Sample trajectory sets (a) $\Tsym$ and (b) $\Tcyc$ for some $n$-qubit systems, with $m=2$ qubits per trajectory. Each circle represents a qubit, and solid black lines denote trajectories. Note that $\Tcyc\subseteq \Tsym$.}
  \label{fig:mathcalT}
\end{figure}

We now solve the TS problem for $\Tsym$ using the strategy outlined at the end of Section \ref{sec:simp-criteria}. Namely, we solve Eq. (\ref{eq:lin-prog}) of Theorem \ref{thm:lin-prog} to derive bounds on the achievable $\theta$. Let $\tilde{\Sigma}_n$ represent $\tilde{G}$ when $G=\Sigma_n$. Then to compute the entries of the matrix $A(\theta)$, we evaluate the orbits of  $\mathbb{Z}_2^n$ and $\Tsymsq$ under $\tilde{\Sigma}_n$.

We begin by determining the orbits of bit-strings in $\mathbb{Z}_2^n$ under $\tilde{\Sigma}_n$. Recall that the elements of $\tilde{\Sigma}_n$ act on a bit-string by permuting the bits and/or flipping every bit. If two bit-strings are related by a permutation, then they have the same weight, where the \textit{weight} of a bit-string $j_1\ldots j_n\in\mathbb{Z}_2^n$ is the number of 1s in the string (i.e., $\sum_k j_k$). Likewise, global bit-flips take bit-strings of weight $\nu$ to bit-strings of weight $n-\nu$. Now let $\mathcal{W}_\nu\subseteq\mathbb{Z}_2^n$ be the subset of bit-strings with weight $\nu$:
\begin{align}
    \mathcal{W}_\nu = \{j_1\ldots j_n\in \mathbb{Z}_2^n\ :\ j_1+\cdots+j_n=\nu\}.
\end{align}
Intuitively, it then follows that each orbit of $\mathbb{Z}_2^n/\tilde{\Sigma}_n$ is equal to ${\mathcal{W}_\nu\cup\mathcal{W}_{n-\nu}}$ for some $\nu\in\{0,\ldots,n\}$:
\begin{proposition}\label{prop:sym-orbits-strings}
    The set of orbits of bit-strings in $\mathbb{Z}_2^n$ under $\tilde{\Sigma}_n$ is $\mathbb{Z}_2^n/\tilde{\Sigma}_n = \{\omega_\nu\ :\ \nu = 0,\ldots,N_{\tilde{\Sigma}_n}-1\}$ where $\omega_\nu = \mathcal{W}_\nu\cup\mathcal{W}_{n-\nu}$. The number of such orbits equals
    \begin{align}\label{eq:N-sigma}
        N_{\tilde{\Sigma}_n}=\floor*{\frac{n}{2}}+1.
    \end{align}
\end{proposition}
\begin{proof}
    See Appendix \ref{appendix:orbs-pairs}.
\end{proof}

Now consider the orbits of trajectory pairs in $\Tsymsq$ under $\tilde{\Sigma}_n$. Recall that the elements of $\tilde{\Sigma}_n$ act on a trajectory pair $(T,T')$ by permuting the indices of qubits in $T$ and $T'$ and/or swapping the positions of $T$ and $T'$ within the pair. 
Note that each qubit in $T\cup T'$ either belongs to only one trajectory or to both. Since permutations are bijective, the number of qubits belonging exclusively to one of the trajectories does not change after a permutation.
Accordingly, define the \textit{degree} of a trajectory pair $(T,T')$ to be the number of qubits in one trajectory that are not in the other (i.e., $\abs{T\setminus T'}=\abs{T'\setminus T}$). It subsequently follows that if two trajectory pairs are related by a permutation, then their degrees are equal. Additionally, since $\abs{T\setminus T'}=\abs{T'\setminus T}$, the degree of a trajectory pair does not change if its trajectories are swapped. Now let $\mathcal{D}_\mu$ be the set of trajectory pairs with degree $\mu$:
\begin{align}
    \mathcal{D}_\mu=\{(T,T')\in\Tsymsq\ :\ \abs{T\setminus T'} = \mu\}.
\end{align}
Then, each orbit in $\Tsymsq/\tilde{\Sigma}_n$ is equal to $\mathcal{D}_\mu$ for some $\mu\in\{0,\ldots,m\}$:
\begin{proposition}\label{prop:sym-orbits-pairs}
    The set of orbits of trajectory pairs in $\Tsymsq$ under $\tilde{\Sigma}_n$ is $\Tsymsq/\tilde{\Sigma}_n = \{\Omega_\mu\ :\ \mu=0,\ldots,M_{\tilde{\Sigma}_n}-1\}$ where $\Omega_\mu = \mathcal{D}_\mu$. The number of such orbits equals 
    \begin{align}\label{eq:M-sigma}
    M_{\tilde{\Sigma}_n} = \begin{cases}
        m+1 &m\leq \floor*{\frac{n}{2}}\\
        n-m+1 &m>\floor*{\frac{n}{2}}.
    \end{cases}
\end{align}
\end{proposition}
\begin{proof}
    See Appendix \ref{appendix:orbs-pairs}.
\end{proof}

For convenience in the remainder of this subsection, we will denote $M_{\tilde{\Sigma}_n}$ and $N_{\tilde{\Sigma}_n}$ with $M$ and $N$, respectively. It is implicitly understood that $M$ and $N$ are functions of the given parameters $m$ and $n$. Comparing Eqs. (\ref{eq:N-sigma}) and (\ref{eq:M-sigma}), note also that $M\leq N$, just as $m\leq n$.

Equipped with the orbits of bit-strings and trajectory pairs under $\tilde{\Sigma}_n$, the next step in solving the TS problem is to compute the entries of the $A(\theta)$ matrix. Per Eq. (\ref{eq:A-entries-def}), these entries are determined by the eigenvalues of the  $\mathbf{R}^{(\mu)}_{\tilde{\mathcal{G}}}$ operators which correspond to the basis states $\ket{\overline{\nu}}$ of the $\tilde{\mathcal{G}}$-invariant subspace. For $G=\Sigma_n$ and $\mathcal{G} = P(G)$, we use the notation $\mathbf{R}^{(\mu)}_{\rm sym}$ and $\mathcal{H}_{\rm sym}$ to represent $\mathbf{R}^{(\mu)}_{\tilde{\mathcal{G}}}$ and $\mathcal{H}_{\tilde{\mathcal{G}}}$, respectively. By Proposition \ref{prop:HG-basis}, the bit-string orbits $\mathbb{Z}_2^n/\tilde{\Sigma}_n$ give an orthogonal basis for $\mathcal{H}_{\rm sym}$; in particular, each unnormalized basis vector $\ket{\overline{\nu}}$ is the sum of all $Z$-eigenbasis states with weight either $\nu$ or $n-\nu$:
\begin{align}\label{eq:sym-basis}
    \ket{\overline{\nu}}=\left(\sum_{j_1\ldots j_n\in \mathcal{W}_\nu}\ket{j_1\ldots j_n}\right) + \left(\sum_{j'_1\ldots j'_n\in \mathcal{W}_{n-\nu}}\ket{j'_1\ldots j'_n}\right)
\end{align}
for all $\nu=0,\ldots,N-1$, where the amplitude of $\ket{\overline{\nu}}$ is halved when $n$ is even and $\nu=N-1=n/2$ to avoid double-counting. For example, if $n=3$, the unnormalized basis vectors for  $\mathcal{H}_{\rm sym}$ are
\begin{align}
    \ket{\overline{0}}&=\ket{000}+\ket{111}\ \mathrm{and}\nonumber\\
    \ket{\overline{1}}&=\ket{100}+\ket{010}+\ket{001}+\ket{011}+\ket{101}+ \ket{110}.
\end{align}
On the other hand, the trajectory pair orbits $\Tsymsq/\tilde{\Sigma}_n$ give the operators $\mathbf{R}^{(\mu)}_{\rm sym}$. By applying each $\mathbf{R}^{(\mu)}_{\rm sym}$ to each $\ket{\overline{\nu}}$, the entries of $A(\theta)$ can be computed explicitly:
\begin{proposition}\label{prop:A-entries}
    Suppose $\mathcal{T}=\Tsym$ and $G=\Sigma_n$. Then the entries of the matrix $A(\theta)$ used in Theorem \ref{thm:lin-prog} are given by
    \begin{align}\label{eq:A-entries}
    A_{\mu,\nu}(\theta) = \alpha_\nu\sum_{i,i'=0}^\nu\binom{\mu}{i}\binom{\mu}{i'}\binom{n-2\mu}{\nu-(i+i')}\cos{[(i-i')\theta]}
\end{align}
    for $\mu=0,\ldots,M-1$ and $\nu=0,\ldots N-1$, where the symbol $\alpha_\nu$ is equal to 1 if $n$ is even and $\nu=N-1=n/2$; otherwise, $\alpha_\nu=2$.
\end{proposition}

The set of achievable $\theta$ can now be computed by determining for what values of $\theta$ the linear program of Eq. (\ref{eq:lin-prog}) admits a feasible solution of $\mathbf{c}$. Although this linear program can be solved numerically for any $\theta\in[0,\pi]$ and size parameters $n$ and $m$, we instead derive closed-form bounds on the achievable $\theta$ using the following two step strategy:

\begin{enumerate}
    \item Given $n$ and $m$, solve the equality constraint $A(\theta)\mathbf{c}=\mathbf{d}$ and express the entries of the solution $\mathbf{c}$ in terms of $\theta$.
    \item Apply the inequality constraint $\mathbf{c}\geq 0$ to each of these entries to find a bound on the $\theta$ for which the linear program is feasible.
\end{enumerate}
Step (1) above could be achieved by performing Gaussian elimination on the augmented matrix $(A|\mathbf{d})$. Note that since only the $\mu=0$ entry of $\mathbf{d}$ is nonzero, the $\mu=0$ row of $(A|\mathbf{d})$ constitutes a normalization condition; this row can thus be removed from the system with the understanding that $\mathbf{c}$ is to be normalized later. Hence, step (1) could also be fulfilled by performing Gaussian elimination on $(A|\mathbf{d})_{1:M-1}$, where $(A|\mathbf{d})_{1:M-1}$ denotes the matrix obtained by removing the $\mu=0$ row of $(A|\mathbf{d})$.

However, due to the fact that the entries of $A(\theta)$ may contain various high-order terms proportional to $\cos{k\theta}$ for $k>1$, it is very challenging to directly execute Gaussian elimination on $(A|\mathbf{d})_{1:M-1}$. The first practical improvement will be to write the matrix $A$ in terms of the new variable $t=\cos{\theta}$ instead of $\theta$. After this substitution, $\cos{k\theta}$ is replaced with $T_k(t)$, where $T_k$ is the $k$th Chebyshev polynomial of the first kind. The entries of $A(t)$ are then
\begin{align}\label{eq:A-entries-t}
    A_{\mu,\nu}(t) = \alpha_\nu\sum_{i,i'=0}^\nu\binom{\mu}{i}\binom{\mu}{i'}\binom{n-2\mu}{\nu-(i+i')}T_{|i-i'|}(t).
\end{align}
As a result, the new system contains only polynomial terms in $t$ instead of trigonometric functions of $\theta$. Nonetheless, it remains unclear whether elementary row operations can perform the cancellations needed to convert $(A|\mathbf{d})_{1:M-1}$ to row-echelon form.

Fortunately, a beautiful relationship between the entries of $A(t)$ and their derivatives ultimately allows these row-reduction challenges to be circumvented. It is shown in Appendix \ref{appendix:deriv-A} that for all $\mu=1,\ldots M-1$ and $\nu=0,\ldots,N-1$,
\begin{align}\label{eq:A-deriv-basecase}
    \mu\left(A_{\mu,\nu}(t)-A_{\mu-1,\nu}(t)\right)=(t-1)\frac{d}{dt}A_{\mu,\nu}(t).
\end{align}
This relation can be intuitively justified via the following argument. Recall that $A_{\mu,\nu}$ is proportional to the eigenvalue of $\mathbf{R}^{(\mu)}_{\rm sym}(\theta)$ with eigenvector $\ket{\overline{\nu}}$.  Note that $\mathbf{R}^{(\mu)}_{\rm sym}(\theta) = \Pi_{\rm sym}\rbf(\theta)\Pi_{\rm sym}$ for any $(T,T')$ in the orbit $\Omega_\mu$, where $\Pi_{\rm sym}$ is the projector onto $\mathcal{H}_{\rm sym}$. In the tensor product decomposition of $\rbf(\theta)$, each qubit in $T\setminus T'$ receives $R^\dagger_Z(\theta)$, while each qubit in $T'\setminus T$ receives $R_Z(\theta)$. Observe that one $R^\dagger_Z(\theta)$ and one $R_Z(\theta)$ together contribute at most a factor of $e^{i\theta}$ to the eigenvalues of $\rbf(\theta)$. It follows that the eigenvalues of $\rbf(\theta)$ do not contain terms of higher order than $e^{i\theta\abs{T\setminus T'}}$, which eventually implies that the entry $A_{\mu,\nu}(t)$ also does not contain terms of higher order than $t^{\abs{T\setminus T'}}$, where $t=\cos{\theta}$. Since $\mu$ is the degree of the trajectory pair $(T,T')$, we have $\mu=\abs{T\setminus T'}$ and the entry $A_{\mu,\nu}(t)$ is subsequently a polynomial in $t$ with degree not exceeding $\mu$. If for the sake of intuition we suppose that the degree of $A_{\mu,\nu}(t)$ is exactly $\mu$, then Eq. (\ref{eq:A-deriv-basecase}) is conceptually supported by the following observation: differentiating $A_{\mu,\nu}(t)$ introduces a multiplicative factor of $\mu$ and appears to decrease its degree by $1$.

By repeatedly differentiating of each side of Eq. (\ref{eq:A-deriv-basecase}), a more general relation can be found relating the $j$th derivative of $A_{\mu,\nu}$ to its $(j+1)$th derivative. Let $A_{\mu,\nu}^{(j)}(t)=\frac{d^{j}}{dt^{j}}A_{\mu,\nu}(t)$ represent the $j$th derivative of $A_{\mu,\nu}(t)$ with respect to $t$. We then obtain the following lemma:
\begin{lemma}\label{lemma:deriv-A}
    Suppose $\mathcal{T}=\Tsym$ and $G=\Sigma_n$. The following relationship holds for all $\mu=1,\ldots,M-1$, $\nu=0,\ldots,N-1$ and integers $j$ such that $0\leq j\leq \mu-1$:
\begin{align}\label{eq:deriv-A}
    (\mu-j)A_{\mu,\nu}^{(j)}(t)-\mu A_{\mu-1,\nu}^{(j)}(t)=(t-1)A_{\mu,\nu}^{(j+1)}(t).
\end{align}
\end{lemma}
\begin{proof}
    See Appendix \ref{appendix:deriv-A}.
\end{proof}

Lemma \ref{lemma:deriv-A} is a crucial result---without the relationship of Eq. (\ref{eq:deriv-A}), it currently appears extremely difficult to derive general closed-form bounds on the achievable $\theta$ when $\mathcal{T}=\Tsym$. In particular, this lemma provides a way to construct an equivalent linear system to $(A|\mathbf{d})_{1:M-1}$ which is in row-echelon form, as desired. However, rather than using elementary row operations to derive this equivalent system from $(A|\mathbf{d})_{1:M-1}$, the relationship of Eq. (\ref{eq:deriv-A}) allows us to simply replace the $\mu$th row of $(A|\mathbf{d})_{1:M-1}$ with its $(\mu-1)$th derivative. The resulting equivalent system is presented by the following theorem:
\begin{theorem}\label{thm:red-lin-prog}
    Suppose $\mathcal{T}=\Tsym$ and $G=\Sigma_n$. Let $A(t)$ and $\mathbf{d}$ be the matrix and vector defined in Theorem \ref{thm:lin-prog}, where $t=\cos{\theta}$. Define the $M\times N$ real matrix $A'(t)$ as follows:
    \begin{align}
        A'_{\mu,\nu}(t) &= \begin{cases}
            A_{0,\nu} (t)&\mu=0\\
            A_{\mu,\nu}^{(\mu-1)} (t)&\mu = 1,\ldots, M-1
        \end{cases}
    \end{align}
    for all $\nu=0,\ldots,N-1$. Then for a given $\theta\in[0,\pi]$, a $\sts$ state exists if and only if there exists $\mathbf{c}\in\mathbb{R}^N$ such that 
    \begin{align}\label{eq:red-lin-prog}
        A'(t)\mathbf{c}=\mathbf{d}\ \mathrm{and}\ \mathbf{c}\geq 0,
    \end{align}
    where $\mathbf{d}\in\mathbb{R}^{M}$ is the vector with entries $d_\mu = \delta_{\mu,0}$. Furthermore, for any $\mathbf{c}$ solving Eq. (\ref{eq:red-lin-prog}), $\ket{\psi}=\sum_{\nu=0}^{N-1}\sqrt{c_\nu}\ket{\overline{\nu}}$ is a satisfactory TS state.
    Note that the entries of $A'$ contain only terms which are linear in $t$, and $A'_{\mu,\nu}=0$ if $\nu < \mu-1$.
\end{theorem}
\begin{proof}
    See Appendix \ref{appendix:red-lin-prog}. 
\end{proof}
The augmented matrix $(A'|\mathbf{d})_{1:M-1}$ representing the system $A'(t)\mathbf{c}=\mathbf{d}$ is in row-echelon form since $A'_{\mu,\nu}=0$ if $\nu < \mu-1$. The next step is to transform this $(A'|\mathbf{d})_{1:M-1}$ into reduced row-echelon form such that each of the solution coefficients $c_\nu$ can be cleanly expressed in terms of $\theta$. Note that the repeated differentiation of $A(t)$ has conveniently killed off all higher order powers of $t$ in the entries of $A'(t)$. Because $A'(t)$ is thus linear in $t$, elementary row operations on $(A'|\mathbf{d})_{1:M-1}$ are relatively straightforward to work with. Hence, we aim to convert $(A'|\mathbf{d})_{1:M-1}$ to reduced row-echelon form via standard Gaussian elimination.

Unfortunately, for many choices of $n$ and $m$, it is difficult to exhaustively describe the set of achievable $\theta$ since this further elimination still remains challenging. Note that our overall strategy of (1) solving $A'(t)\mathbf{c}=\mathbf{d}$ and then (2) enforcing $\mathbf{c}\geq 0$ to retrieve a bound on $\theta$ yields necessary and sufficient conditions on $\theta$ for the existence of a TS state. Such conditions represent a full solution to the TS problem, as they completely characterize the set of achievable $\theta$. Although we ultimately find a full solution for special values of $n$ and $m$, it is often infeasible to find general closed-form necessary conditions. However, we can still use our special full solutions to derive general sufficient conditions, even when necessary conditions are unobtainable. In particular, the following proposition explains when sufficient conditions for one TS problem also hold for another related problem. 
\begin{proposition}\label{lemma:smaller}
    Suppose $G=\Sigma_n$. Consider two TS scenarios with a sensor of $n$ qubits: one with $\mathcal{T}=\mathcal{T}_{\rm sym}(n,m_1)$ and another with $\mathcal{T}=\mathcal{T}_{\rm sym}(n,m_2)$. Let $M_j=m_j+1$ if $m_j\leq \floor{n/2}$ and $M_j=n-m_j+1$ otherwise, per Eq. (\ref{eq:M-sigma}). Now suppose for a given $\theta\in[0,\pi]$ that a TS state $\ket{\psi}$ exists discriminating the trajectories in $\mathcal{T}_{\rm sym}(n,m_1)$. Then for the same $\theta$, $\ket{\psi}$ also discriminates the trajectories in $\mathcal{T}_{\rm sym}(n,m_2)$ if $M_2\leq M_1$.
\end{proposition}
\begin{proof}
    Let $N$ be the quantity defined in Eq. (\ref{eq:N-sigma}). Define $A(\theta)$ to be the $M_1\times N$ matrix whose entries are given by Eq. (\ref{eq:A-entries}) for $\mu=0,\ldots M_1-1$ and $\nu=0,\ldots,N-1$. By Theorem \ref{thm:lin-prog}, a TS state exists at the given $\theta$ for the trajectory set $\mathcal{T}_{\rm sym}(n,m_1)$ if and only if $A(\theta)\mathbf{c}=\mathbf{d}_1$ for some $\mathbf{c}\geq 0$, where $\mathbf{d}_1\in\mathbb{R}^{M_1}$ has entries $d_\mu=\delta_{\mu,0}$. Now let $A_{0:M_2-1}(\theta)$ be the matrix formed by rows $0$ to $M_2-1$ of $A(\theta)$; $A_{0:M_2-1}(\theta)$ is well-defined since $M_2\leq M_1$. Then by Theorem \ref{thm:lin-prog}, a TS state exists at the given $\theta$ for the trajectory set $\mathcal{T}_{\rm sym}(n,m_2)$ if and only if $A_{0:M_2-1}(\theta)\mathbf{c}=\mathbf{d}_2$ for some $\mathbf{c}\geq 0$, where $\mathbf{d}_2\in\mathbb{R}^{M_2}$ has entries $d_\mu=\delta_{\mu,0}$. The result follows from the observation that a $\mathbf{c}\geq 0$ automatically satisfies $A_{0:M_2-1}(\theta)\mathbf{c}=\mathbf{d}_2$ if it satisfies $A(\theta)\mathbf{c}=\mathbf{d}_1$ and $M_2\leq M_1$.
\end{proof}

We now determine the set of achievable $\theta$ for the special case where $n$ and $m$ yield $M=N$ through Eqs. (\ref{eq:M-sigma}) and (\ref{eq:N-sigma}); note by Proposition \ref{lemma:smaller} that, for fixed $n$, this set of $\theta$ is also achievable for any $m$ such that $M\leq N$. Thus, assume $M=N$. If $N=1$, then the TS problem is trivial, so assume $N>1$. By using Gaussian elimination to convert the augmented matrix $(A'|\mathbf{d})_{1:M-1}$ into reduced row-echelon form, each $c_\nu$ can ultimately be expressed exclusively in terms of $t=\cos{\theta}$ and $c_{N-1}$ for all $\nu=0,\ldots, N-1$, per the following proposition:

\begin{proposition}\label{prop:rref}
    Suppose $\mathcal{T}=\Tsym$ and $G=\Sigma_n$. Then $\mathbf{c}$ satisfies $A'(t)\mathbf{c}=\mathbf{d}$ if and only if $\sum_\nu A'_{0,\nu} c_\nu = 1$ and 
    \begin{align}\label{eq:TSsym-coeffs}
        c_\nu=\begin{cases}
            (-1)^{N-1-\nu}T_{N-1-\nu}(t)c_{N-1} &n\ \mathrm{even}\\
        (-1)^{N-1-\nu}W_{N-1-\nu}(t)c_{N-1} &n\ \mathrm{odd},
        \end{cases}
    \end{align}
    where $W_k(t)$ is defined to be the $k$th Chebyshev polynomial of the fourth kind \cite{cheby}:
    \begin{align}
        W_k(\cos{\theta}) = \frac{\sin{((k+1/2)\theta)}}{\sin{(\theta/2)}}.
    \end{align}
\end{proposition}
\begin{proof}
    See Appendix \ref{appendix:TSsym-coeffs}.
\end{proof}
Since Eq. (\ref{eq:TSsym-coeffs}) expresses the solution to $A'(t)\mathbf{c}=\mathbf{d}$ simply in terms of $t=\cos{\theta}$, Theorem \ref{thm:red-lin-prog} implies that imposing the nonnegativity constraint $\mathbf{c}\geq 0$ should readily yield bounds on the achievable $\theta$. Note that $\mathbf{c}\geq 0$ requires $c_\nu/c_{N-1} \geq 0$ for all $\nu=0,\ldots,N-1$. Since these coefficient ratios depend solely on $\theta$, enforcing them to be nonnegative gives necessary and sufficient conditions on $\theta$ for the existence of a TS state, thereby fully solving the TS problem for $\mathcal{T}=\Tsym$ when $M=N$. Proposition \ref{lemma:smaller} then allows this particular full solution to be extended to the case of general $n$ and $m$ in the form of sufficient conditions:
\begin{theorem}\label{thm:S-TS-criterion}
    Suppose $\mathcal{T}=\Tsym$. For $\theta\in[0,\pi]$ and arbitrary $n>0$ and $m\geq 0$, a sufficient criterion for the existence of a $\sts$ state is 
    \begin{align}\label{eq:S-TS-criterion}
       \theta\geq\frac{(n-1)\pi}{n}.
    \end{align}
    Furthermore, when $m=\floor*{\frac{n}{2}}$ or $\ceil*{\frac{n}{2}}$, Eq. (\ref{eq:S-TS-criterion}) becomes a necessary criterion.
\end{theorem}
\begin{proof}
    Suppose $m=\floor*{\frac{n}{2}}$ or $\ceil*{\frac{n}{2}}$ so that $M=N = \floor*{\frac{n}{2}}$. Under these conditions, Proposition \ref{prop:rref} and Theorem \ref{thm:red-lin-prog} imply that for a particular $\theta$, a TS state exists if and only if there exists $\mathbf{c}\geq 0$ such that Eq. (\ref{eq:TSsym-coeffs}) holds. Such a $\mathbf{c}$ exists if and only if
    \begin{align}\label{eq:coeffs-ineq}
        0\leq \begin{cases}
        (-1)^{N-1-\nu}T_{N-1-\nu}(\cos{\theta}) &n\ \mathrm{even}\\
        (-1)^{N-1-\nu}W_{N-1-\nu}(\cos{\theta}) &n\ \mathrm{odd},
    \end{cases}
    \end{align}
    for all $\nu=0,\ldots N-1$. It can readily be verified that Eq. (\ref{eq:coeffs-ineq}) holds for both even and odd $n$ if and only if $\theta$ satisfies Eq. (\ref{eq:S-TS-criterion}) (see Appendix \ref{appendix:S-TS-criterion}). Hence, Eq. (\ref{eq:S-TS-criterion}) is a necessary and sufficient criterion for the existence of a $\sts$ state when $m=\floor*{\frac{n}{2}}$ or $\ceil*{\frac{n}{2}}$.

    When $m\neq \floor*{\frac{n}{2}}$ or $\ceil*{\frac{n}{2}}$, we have $M\leq N$. However, applying Proposition \ref{lemma:smaller} with $m_2=m$ and $m_1=\floor*{\frac{n}{2}}$, we deduce that Eq. (\ref{eq:S-TS-criterion}) is a sufficient condition for the existence of a $\sts$ state.
\end{proof}

Because the proof of Theorem \ref{thm:S-TS-criterion} is constructive, we can recover descriptions of the TS states which are guaranteed to exist. For any $n$ and $m$, if $\theta$ satisfies Eq. (\ref{eq:S-TS-criterion}), then Theorem \ref{thm:red-lin-prog} implies that $\ket{\psi_{\rm sym}} = \sum_{\nu=0}^{N-1}\sqrt{c_\nu}\ket{\overline{\nu}}$ is a TS state, where $c_{N-1}$ is chosen such that $\ket{\psi_{\rm sym}}$ is normalized and the remaining $c_\nu$ are determined by Eq. (\ref{eq:TSsym-coeffs}). We call a TS state of this form a $\sts$ state. For example, if $n=2m$ and $\theta$ satisfies Eq. (\ref{eq:S-TS-criterion}), then a satisfactory $\sts$ state is
\begin{align}
    \ket{\psi_{\rm sym}}= \sum_{\nu=0}^m\sqrt{\abs{\cos\left[(m-\nu)\theta\right]}}\ket{\overline{\nu}}
\end{align}
up to normalization, where $\ket{\overline{\nu}}$ is the unnormalized superposition over $Z$-eigenbasis states with weight $\nu$ or $n-\nu$.

Since Theorem \ref{thm:S-TS-criterion} does not provide necessary conditions on $\theta$ when $m\neq\floor{n/2}$ or $\ceil{n/2}$, the true minimum achievable $\theta$ may in general be lower than the threshold of Eq. (\ref{eq:S-TS-criterion}). However, when $m$ is chosen such that $M$ is small, it is in fact possible to sharpen the bound on the minimum achievable $\theta$ by exactly solving the linear programs of Theorems \ref{thm:lin-prog} or \ref{thm:red-lin-prog}. We provide these improved bounds for $M=0$ and $M=1$ in the following theorem:
\begin{theorem}\label{thm:small-sym}
    Suppose $\mathcal{T}=\Tsym$. Then the following are true for $\theta\in[0,\pi]$ and $n>0$:
    \begin{enumerate}[label=(\alph*), ref=\ref{thm:small-sym}\alph*]
        \item If $m=0$ or $m=n$, then any state is a TS state.
        \item If $n>1$ and either $m=1$ or $m=n-1$, then a necessary and sufficient criterion for the existence of a TS state is
        \begin{align}
            \theta \geq \arccos{\left(-1+\ceil*{\frac{n}{2}}^{-1}\right)}.
        \end{align}\label{thm:small-sym1}
    \end{enumerate}
\end{theorem}
\begin{proof}
    See Appendix \ref{appendix:small-cases}.
\end{proof}
While not explored here, we note that an improved bound for $M=2$ (i.e., $m=2$ or $m=n-2$) can also be found by directly solving the linear programs. Additionally, we remark that recent work \cite{gupta} has independently proved a separate result equivalent to Theorem \ref{thm:small-sym1}. 

From an experimental point of view, it is often desirable for sensors to be able to perfectly distinguish trajectories in a single-shot measurement even when the particle-sensor interaction strength $\theta$ is weak. However, Theorem \ref{thm:S-TS-criterion} suggests that when $\mathcal{T}=\Tsym$, the minimum achievable $\theta$ increases toward $\pi$ (i.e., the maximum possible interaction strength) as the number of qubits in the sensor increases. This loss of sensor sensitivity corresponds with the observation that the number of trajectories in $\Tsym$ generally increases extremely rapidly with $n$, as $\abs{\Tsym}=\binom{n}{m}$. In the next section, we address this shortcoming by showing for appropriately restricted trajectory sets that the minimum achievable $\theta$ does not increase as the number of sensor qubits grows.

\subsection{$\mathcal{T}$ is transitive under the cyclic group}\label{sec:cyc}
Because $\Tsym$ includes all possible size-$m$ trajectories, it may include many trajectories which are not physically relevant in a practical setting. For example, assuming the sensor qubits are spatially distributed in a 3D array, some trajectories in $\Tsym$ will consist of qubits which are are not localized together along a continuous curve. Because particle-sensing applications may be interested in only such ``continuous" trajectories, it is desirable to know how to restrict $\mathcal{T}$ to exclude all discontinuous trajectories.

In this section, we will again choose our trajectory set to be $\mathcal{T}_G(n,m)$, but we instead select $G$ to be the cyclic group $Z_n$ of order $n$, where $Z_n = \{z^j: j=0,\ldots,n-1\}$ with generator permutation $z=(1 \ldots n)$. Here, the notation $z^j$ represents the permutation $z$ applied $j$ times. Accordingly, define $\Tcyc$ to be $T_G(n,m)$ when $G=Z_n$. Evidently, the qubits within each trajectory in $\Tcyc$ must have consecutive indices modulo $n$ (see Figure \ref{fig:mathcalT}). Additionally, $\Tcyc$ is a subset of $\Tsym$, and  $\abs{\Tcyc} = n$ scales only linearly with the number of sensor qubits.

Using the same strategy as Section \ref{sec:sym-group}, the criteria of Theorem \ref{thm:lin-prog} can be applied to determine whether a TS state exists at a particular value of $\theta$ when $\mathcal{T}=\Tcyc$. Let $\tilde{Z}_n$ equal the group $\tilde{G}$ when $G=Z_n$. To calculate the entries of the $A(\theta)$ matrix in Eq. (\ref{eq:lin-prog}), it is necessary to compute the orbits of bit-strings in $\mathbb{Z}_2^n$ and trajectory pairs in $\Tcycsq$ under the action of $\tilde{Z}_n$. The orbits of trajectory pairs are straightforward to determine:
\begin{proposition}\label{prop:cyc-orbits}
    The set of orbits of trajectory pairs in $\Tcycsq$ under $\tilde{Z}_n$ is $\Tcycsq/\tilde{Z}_n = \left\{\Omega_\mu\ :\ \mu=0,\ldots,M_{\tilde{Z}_n}-1\right\}$ where $\Omega_\mu=\operatorname{Orb}_{\tilde{Z}_n}[([m],z^\mu([m]))]$ and $z=(1\ldots n)$. The number $M_{\tilde{Z}_n}$ of such orbits equals $\floor{n/2}+1$. 
\end{proposition}
\begin{proof}
    See Appendix \ref{appendix:cyclic-orbs}.
\end{proof}

However, the orbits of bit-strings under $\tilde{Z}_n$ are considerably more complex to describe, and they are provided by the set of distinct length-$n$ binary necklaces \cite{necklace} where two necklaces are considered equivalent if related by the $n$-bit bit-flip operation. The number $N_{\tilde{Z}_n}$ of such orbits is given by sequence A000013 of the OEIS \cite{oeis} with the formula
\begin{align}
    N_{\tilde{Z}_n}=\frac{1}{n}\sum_{d|n}2^{\frac{n}{d}-1}\phi(2d),
\end{align}
where $d|n$ is the set of positive integers $d$ that divide $n$ and $\phi(\cdot)$ is the Euler totient function. The first several values for $N_{\tilde{Z}_n}$ when $n=1,2,\ldots$ are:
\begin{align}
    N_{\tilde{Z}_n}=1, 2, 2, 4, 4, 8, 10, 20, 30,\ldots.
\end{align}

Although the linear program of Theorem \ref{thm:lin-prog} can be used to numerically determine whether a TS state exists for $\mathcal{T}=\Tcyc$ and some given $\theta$, it is challenging to use this strategy to analytically derive bounds on the achievable $\theta$ as we did previously in Section \ref{sec:sym-group}. Recall that the bound in Theorem \ref{thm:S-TS-criterion} was effectively computed by solving the system of Eq. (\ref{eq:lin-prog}) for the special case where $A$ is a square matrix---that is, when $M_{\tilde{\Sigma}_n}=N_{\tilde{\Sigma}_n}$. However, when $G=Z_n$, it is apparent that no choice of $n>3$ and $m$ will result in $M_{\tilde{Z}_n}=N_{\tilde{Z}_n}$; consequently, it is impossible for $A$ to be square in general, thereby prohibiting a similar approach. Moreover, the complicated scaling of $N_{\tilde{Z}_n}$ with $n$ presents a significant obstacle for finding a closed-form solution to Eq. (\ref{eq:lin-prog}) that generalizes to arbitrary $n$.

To circumvent this issue, we now show that if $n$ is a constant multiple of $m$, there exist TS states for $\Tcyc$ which decompose into the tensor product of several identical, smaller $\sts$ states. Because TS states which decompose in this manner admit a vastly simpler description, it is easier to determine the range of $\theta$ over which they are guaranteed to exist. Hence, suppose $n=\kappa m$ for some integer $\kappa>1$. To visualize this decomposition, it is easiest to imagine that the $n$ qubits are assembled in a $\kappa\times m$ rectangular array. The array is filled in row-major order (using one-indexing) such that position $(r,s)$ in the array is occupied by qubit $(r-1)m+s$. Now suppose there exists a $\kappa$-qubit $\sts$ state $\ket{\phi}$ that discriminates the trajectories in $\mathcal{T}_{\rm sym}(n'=\kappa,m'=1)$ at a particular value of $\theta$. We claim that the $n$-qubit state $\ket{\psi}$ obtained by preparing each column of the array to $\ket{\phi}$ is a TS state which discriminates the trajectories in $\Tcyc$ at the same value of $\theta$. We will call any TS state constructed in this manner a $\cts$ state.

Theorem \ref{lemma:SG-simp} can be used to show that $\ket{\psi}$ is indeed a TS state for $\Tcyc$ at the given value of $\theta$. For $G=Z_n$ and $\mathcal{G}=P(G)$, let $\mathcal{H}_{\rm cyc}$ denote the $\tilde{\mathcal{G}}$-invariant subspace $\mathcal{H}_{\tilde{\mathcal{G}}}$. To use this theorem, we must first show that $\ket{\psi}$ belongs in $\mathcal{H}_{\rm cyc}$ , which requires that $\ket{\psi}$ be invariant under both $n$-qubit bit-flips and permutations in $Z_n$. Since the $\kappa$-qubit state $\ket{\phi}$ constituting each column of $\ket{\psi}$ is a $\sts$ state, it is already invariant under global bit-flips and any permutations. It follows that $\ket{\psi}$ is also bit-flip invariant. To show that $\ket{\psi}$ is also invariant under permutations in $Z_n$, it suffices to show that $\ket{\psi}$ is unchanged by the cyclic permutation $z=(1\ldots n)$. Note that $z$ can be implemented on the array by ($\pi_c$) cyclically permuting the collection of columns and then ($\pi_r$) cyclically permuting just the qubits within the last column. For example, consider the result of successively applying $\pi_c$ and $\pi_r$ on the array of qubit indices when $\kappa=2$ and $m=3$:
\begin{align}
    \begin{bmatrix}
        1&2&3\\
        4&5&6
    \end{bmatrix}\xrightarrow[]{\pi_c}
    \begin{bmatrix}
        2&3&1\\
        5&6&4
    \end{bmatrix}&\xrightarrow[]{\pi_r}
    \begin{bmatrix}
        2&3&4\\
        5&6&1
    \end{bmatrix}\nonumber\\
    &=\begin{bmatrix}
        z(1)&z(2)&z(3)\\
        z(4)&z(5)&z(6)
    \end{bmatrix}.
\end{align}
Since all the columns of $\ket{\psi}$ are identically in the state $\ket{\phi}$, permuting the columns with $\pi_c$ leaves $\ket{\psi}$ unchanged. Additionally, since the state $\ket{\phi}$ of each column is permutation-invariant, permuting within a column using $\pi_r$ also leaves $\ket{\psi}$ unchanged. Since $\ket{\psi}$ does not change under permutations $\pi_c$ and $\pi_r$, $\ket{\psi}$ is invariant under the action of $z$ and therefore invariant under the group $Z_n$. Because $\ket{\psi}$ is invariant under bit-flips and $Z_n$, it belongs in $\mathcal{H}_{\rm cyc}$, as desired. 

To verify the criteria of Theorem \ref{lemma:SG-simp} for $\ket{\psi}$, it is first necessary to consider the application of $R^{(T)}(\theta)$ to the $\kappa\times m$ array of qubits for any $T\in\Tcyc$. Because all trajectories $T\in\mathcal{T}_{\rm cyc}$ consist of $m$ consecutive qubits, any $R^{(T)}(\theta)$ operator applies the rotation $R_Z(\theta)$ to exactly one qubit in each column of the array. For example, if $\kappa=2$, $m=3$, and $T=\{3,4,5\}$, then $R^{(T)}(\theta)$ takes the form
\begin{align}\label{eq:RT-array}
    \begin{bmatrix}
        I&I&R_Z\\
        R_Z&R_Z&I
    \end{bmatrix},
\end{align}
where the operator at entry $(k,l)$ of the above matrix is applied to the qubit at position $(k,l)$ in the sensor array; clearly, there is only one $R_Z$ applied per column of the sensor array. 

Now consider the application of $\rbf=R^{\dagger(T)}R^{(T')}$ to the array for any $T,T'\in\mathcal{T}_{\rm cyc}$. By the above argument, $R^{\dagger(T)}$ and $R^{(T')}$ individually deposit exactly one $R^\dagger_Z$ and one $R_Z$ per column, respectively. Using the same depiction as Eq. (\ref{eq:RT-array}), the operator $\rbf$ takes the following form when $T=\{1,2,3\}$ and $T'=\{3,4,5\}$:
\begin{align}\label{eq:decomp}
\begin{bmatrix}
        R^{\dagger}_Z&R^{\dagger}_Z&R^{\dagger}_ZR_Z\\
        R_Z&R_Z&I 
    \end{bmatrix}=\begin{bmatrix}
        R^{\dagger}_Z&R^{\dagger}_Z&I\\
        R_Z&R_Z&I
    \end{bmatrix}.
\end{align}
The perturbation induced on the state $\ket{\psi}$ by $\rbf$ is visualized in Figure \ref{fig:cyclic-decomp}.

\begin{figure}[htbp]
  \includegraphics[width = \linewidth]{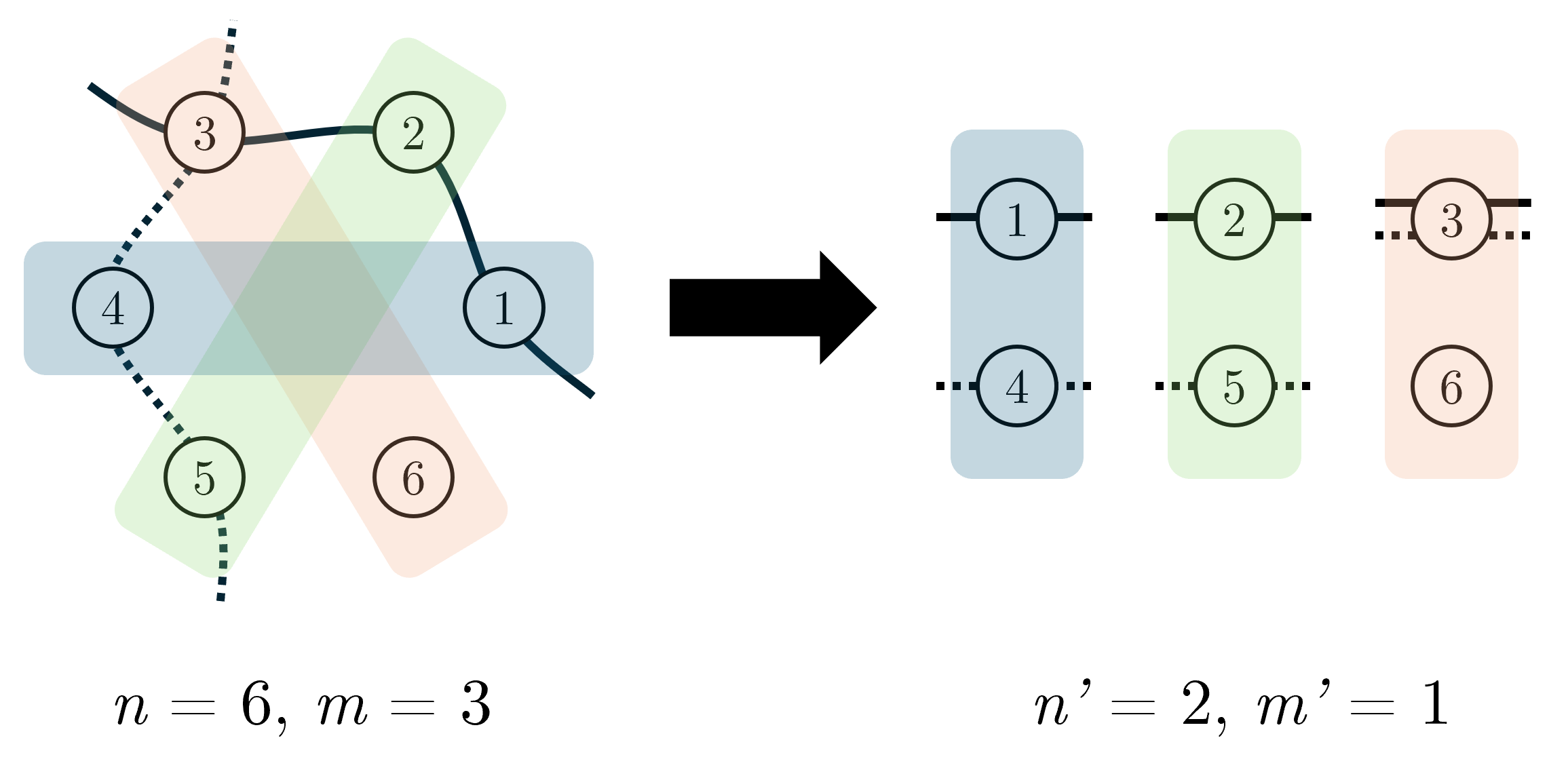}
  \caption{Decomposition of trajectories $T=\{1,2,3\}$ (solid line) and $T'=\{3,4,5\}$ (dotted line) applied to a $\cts$ state, as described in Eq. (\ref{eq:decomp}). $T$ and $T'$ belong to the set $\Tcyc$, where $n=6$ and $m=3$. Each pair of qubits grouped under a colored box is prepared to a $\sts$ state which can distinguish the trajectories in $\mathcal{T}_{\rm sym}(n'=2,m'=1)$. Note that $T$ and $T'$ locally look like distinct trajectories in $\mathcal{T}_{\rm sym}(n'=2,m'=1)$ when applied to the blue and green $\sts$ states. Thus, from a local viewpoint, $T$ and $T'$ map the each of the blue and green $\sts$ states to orthogonal outputs. This fact is sufficient to guarantee that the two full trajectories yield orthogonal outputs for the whole $\cts$ state in the global picture.}
  \label{fig:cyclic-decomp}
\end{figure}

To use Theorem \ref{lemma:SG-simp}, we must next prove that $\ket{\psi}$ satisfies Eq. (\ref{eq:ortho}) for at least one representative trajectory pair in each orbit of $\Tcycsq/\tilde{Z}_n$. Pick any $\mu\in\{0,\ldots, M_{\tilde{Z}_n}-1\}$. Then the trajectory pair $([m],z^\mu([m]))$ belongs to the orbit $\Omega_\mu\in\Tcycsq/\tilde{Z}_n$ by Proposition \ref{prop:cyc-orbits}. Because the two trajectories in the pair are equal if and only if $\mu=0$, we aim to show that
\begin{align}\label{eq:cyc-red-criteria}
    \bra{\psi}\mathbf{R}^{([m],z^\mu([m]))}(\theta)\ket{\psi} = \delta_{\mu,0}.
\end{align}
The inner product in Eq. (\ref{eq:cyc-red-criteria}) can be decomposed into $m$ inner products over the $m$ columns of the array. Noting from above that each column receives exactly one $R^\dagger_Z$ and one $R_Z$, we can write
\begin{align}
    \bra{\psi}\mathbf{R}^{([m],z^\mu([m]))}(\theta)\ket{\psi} &= \prod_{\mathrm {column}\ c=1}^m\bra{\phi}\mathbf{R}^{(T_c,T'_c)}\ket{\phi}
\end{align}
where $T_c = \{k_c\}$ and $T'_c=\{k'_c\}$ for some indices $k_c,k'_c$ of qubits that belong in the $c$th column. Alternately, $T_c$ and $T'_c$ can be interpreted as single-qubit trajectories within the $c$th column. Because the state $\ket{\phi}$ of each column is a $\sts$ state capable of discriminating all size-$1$ trajectories within the column, Eq. (\ref{eq:ortho}) implies that $\bra{\phi}\mathbf{R}^{(T_c,T'_c)}\ket{\phi}=\delta_{T_c,T'_c}$. Since $[m]=z^\mu([m])$ (or equivalently, $\mu=0$) if and only if $S_c=S'_c$ for all columns $c$,
\begin{align}
   \prod_{\mathrm {column}\ c=1}^m\bra{\phi}\mathbf{R}^{(T_c,T'_c)}\ket{\phi} &= \prod_{\mathrm{column}\ c}\delta_{T_c, T'_c}=\delta_{\mu,0}.
\end{align}
It follows that Eq. (\ref{eq:ortho}) holds for one representative trajectory pair in each orbit $\Omega_\mu$ for $\mu = 0,\ldots,M_{\tilde{Z}_n}-1$. Because $\ket{\psi}$ is in $\mathcal{H}_{\rm cyc}$, Theorem \ref{lemma:SG-simp} guarantees that $\ket{\psi}$ is a $\cts$ state for $\Tcyc$ at the given value of $\theta$, as desired.

Therefore, if a $\sts$ state $\ket{\phi}$ exists for $\mathcal{T}_{\rm sym}(n'=\kappa,m'=1)$ at a particular $\theta$, then a $\cts$ state $\ket{\psi}$ exists for $\mathcal{T}_{\rm cyc}(n=\kappa m,m)$ at the same $\theta$. The interval of $\theta$ over which the larger $\cts$ state exists consequently includes the interval over which the smaller $\sts$ state exists, leading to the following theorem:
\begin{theorem}\label{thm:C-TS-criterion}
    Suppose $\mathcal{T}=\Tcyc$, where $n=\kappa m$ for some arbitrary $m>0$ and $\kappa>1$. Given $\theta\in[0,\pi]$, a sufficient condition for the existence of a TS state is
    \begin{align}\label{eq:C-TS-criterion}
        \theta \geq \arccos{\left(-1+\ceil*{\frac{\kappa}{2}}^{-1}\right)}.
    \end{align}
\end{theorem}
\begin{proof}
    As shown above, the desired TS state exists if a smaller $m'=1$, $n'=\kappa$ $\sts$ state exists. A sufficient criterion for the existence of this smaller $\sts$ state is given by Theorem \ref{thm:small-sym1}, and substituting $n=\kappa$ into this criterion gives Eq. (\ref{eq:C-TS-criterion}).
\end{proof}

Given that the form of $\sts$ states is known, it is straightforward to provide descriptions for the $\cts$ states guaranteed to exist by Theorem \ref{thm:C-TS-criterion}. For example, if $n=2m$ and $\theta$ obeys Eq. (\ref{eq:C-TS-criterion}), then a satisfactory TS state takes the form of
\begin{align}\label{eq:example-cts}
    \ket{\psi_{\rm cyc}}=\left[\sqrt{\abs{\cos{\theta}}}(\ket{00}+\ket{11}) + (\ket{01}+\ket{10})\right]^{\otimes m}
\end{align}
up to normalization and a permutation of the qubits.

Since Eq. (\ref{eq:C-TS-criterion}) is only a sufficient condition, it represents an upper bound to the minimum achievable $\theta$ when $\mathcal{T}=\Tcyc$. However, assuming $n$ remains a fixed multiple of $m$, then this upper bound is in fact constant with respect to $n$. It follows that, when $\mathcal{T}=\Tcyc$, the minimum achievable $\theta$ does not increase even in the large-sensor limit.

In Section \ref{sec:results}, we solved the TS problem assuming the trajectory set is transitive for some permutation group. As a result, although the permutation symmetry group $G$ associated with each TS problem has been highly nontrivial, the associated Pauli symmetry group $\mathcal{S}$ has remained comparatively simple, as we have chosen $\mathcal{S}=\{I^{\otimes n},X^{\otimes n}\}$ every time. In the next section, we shift the focus from $G$ to $\mathcal{S}$: after first exploring the relationship between TS states and quantum codes, we then investigate useful TS scenarios where the Pauli symmetry group may instead be highly nontrivial.

\section{Quantum codes for trajectory sensing}\label{sec:qecc}
Quantum error correcting codes enable the recovery of quantum information after 
an undesirable error process by providing a means for the possible errors to be diagnosed and reversed. Generally speaking, the information to be preserved is encoded into a subspace of the full Hilbert space. Presuming that the encoded state is then corrupted by some unknown error, the goal is to perform syndrome measurements on the state to identify the error without destroying the information contained within. If successful, the error may be undone, and the original encoded state may be recovered without loss of quantum information.

In fact, the setting of quantum error correction  is in many ways fundamentally equivalent to that of quantum trajectory sensing, but the TS problem does not require preservation of quantum information and its error model is very different. In TS, a sensor state is prepared, subjected to an unknown trajectory, and measured with the intent to identify the trajectory (see Section~\ref{sec:formalism}).  Unlike in QEC, the ``errors" in TS are desirable system-environment interactions, and the purpose is shifted away from information recovery and toward error diagnosis. Moreover, while QEC typically supposes that each qubit is subject to independent and identically distributed errors, in TS the trajectories instead constitute highly correlated, many-qubit perturbations of a kind that are rarely considered in the standard QEC setting.

Due to the trajectory error model, a TS state should constitute a special kind of quantum code. On the other hand, the diagnostic capability of existing quantum codes suggests that they could be repurposed for sensing rather than information recovery. Specifically, we expect that syndrome measurements on such codes may be used to distinguish a discrete set of perturbations of interest.

In Section \ref{sec:sensor-code}, we establish a correspondence between the criteria for general codes and the criteria for TS states. Then, in Section \ref{sec:TS-code}, we demonstrate that a subspace of TS states can encode a logical qubit. Afterwards, we discuss in Section \ref{sec:stab} how the symmetries used earlier to simplify the search for TS states closely relate to familiar symmetries of existing quantum codes, particularly Pauli stabilizers. We subsequently show how a number of known stabilizer codes can be alternatively deployed for trajectory sensing. Lastly, in Section \ref{sec:noise}, we illustrate how code concatenation can be used to enhance TS states by, for example, making them resilient to noise.

\subsection{Criteria for quantum codes and TS states}\label{sec:sensor-code}

We now review the mathematical criteria for a general quantum error correcting code and contrast them with the TS state criteria.
An $[[n,k]]$ quantum code $C$ is a $2^k$-dimensional subspace of the full $n$-qubit Hilbert space $\mathcal{H}$ and is spanned by an orthonormal basis of code states $\{\ket{\psi_i}\ :\ i=1,\ldots,2^k\}$. Suppose the system-environment interaction can be modeled with the quantum channel $\mathcal{E}(\rho) = E_a\rho E_a^\dagger$ where $\rho$ is the density matrix of the system and the $\{E_a\}$ are a set of Kraus operators satisfying $\sum_a E_a^\dagger E_a = I$. Intuitively, this interaction affects a code state $\ket{\phi}\in C$ by applying one of the perturbations $E_a$ to it with some probability, yielding the output $E_a\ket{\phi}$. A projective syndrome measurement on the output reveals information about which $E_a$ was applied to $\ket{\phi}$, and $\ket{\phi}$ can be recovered if the code satisfies the Knill-Laflamme (KL) error correction criteria \cite{knill}:
\begin{align}\label{eq:KL}
\bra{\psi_i}E^\dagger_a E_b \ket{\psi_j} = \gamma_{ab}\delta_{ij}
\end{align}
for some Hermitian matrix $\gamma$ which is independent of $\ket{\psi_i}$. When $\gamma$ is diagonal, Eq. (\ref{eq:KL}) expresses the requirement that the output states for every error must be orthogonal. Thus, any code with diagonal $\gamma$ can perfectly discriminate the given set of perturbations with a single projective measurement.

However, even though a code may be effective in recovering a state subject to a set of perturbations, it may not be able to \textit{unambiguously} identify each perturbation if $\gamma$ is not diagonal. In fact, degenerate codes (for which some $\bra{\psi_i}E_a^\dagger E_b\ket{\psi_i} = 1$ for $a\neq b$) are completely unable to distinguish some of the perturbations, despite the fact that in some scenarios they can recover errors under higher noise rates than can non-degenerate codes \cite{Smith_2007}. Hence, to successfully discriminate a discrete set of perturbations, a quantum code must prioritize syndrome diagnosis over only information recovery.

The KL error correction criteria of Eq. (\ref{eq:KL}) bear a striking resemblance to the TS state criteria of Eq. (\ref{eq:ortho}) and suggest that TS states may form a kind of quantum code. In fact, TS states indeed constitute quantum error-correcting codes under the appropriate error channel. Consider the channel 
\begin{align}
\mathcal{E}_{\mathrm{TS}}(\rho)= \frac{1}{\abs{\mathcal{T}}}\sum_{T\in\mathcal{T}} R^{(T)}(\theta)\rho R^{\dagger(T)}(\theta),
\end{align}
which corresponds to the action of a trajectory $T\in\mathcal{T}$ chosen uniformly at random. It is easy to see that a TS state $\ket{\psi_{\mathrm{TS}}}$ should be recoverable under $\mathcal{E}_{\mathrm{TS}}$ as follows. Because the outputs corresponding to each trajectory acting on $\ket{\psi_{\mathrm{TS}}}$ are orthogonal, the trajectory $T$ can be unambiguously determined in a single measurement. Afterwards, the state $\ket{\psi_{\mathrm{TS}}}$ can be restored by applying the recovery operator $R^{\dagger(T)}$ to the output.

More formally, we can show that the $[[n,0]]$ code $\{\ket{\psi_{\rm TS}}\}$ satisfies the KL criteria of Eq. (\ref{eq:KL}) under this channel. If each trajectory $T\in\mathcal{T}$ is assigned a unique integer index $a=1,\ldots,\abs{\mathcal{T}}$, then the Kraus operators for $\mathcal{E}_{\mathrm{TS}}$ are  
\begin{align}\label{eq:TS-kraus}
E_a=\frac{1}{\sqrt{\abs{\mathcal{T}}}}R^{\left(T_a\right)}(\theta).
\end{align}
Since $\ket{\psi_{\rm TS}}$ is a TS state, it follows immediately from Eq. (\ref{eq:ortho}) that Eq. (\ref{eq:KL}) holds with $\gamma_{ab} = \delta_{a,b}/\abs{\mathcal{T}}$. Hence, $\{\ket{\psi_{\rm TS}}\}$ is an error-correcting code for the channel $\mathcal{E}_{\mathrm{TS}}$. Moreover, since $\gamma$ is diagonal, a syndrome measurement can unambiguously determine $E_a$, or equivalently, the trajectory which perturbed the TS state.

To re-emphasize the distinction between information recovery and error diagnosis, note that the code $\{\ket{\psi}\}$ for \textit{any} $\ket{\psi}\in\mathcal{H}$ also satisfies the KL criteria under this channel; $\ket{\psi}$ can always trivially be recovered by discarding the output state and reinitializing $\ket{\psi}$. However, although any $\ket{\psi}$ can be recovered from this channel, the ``errors" (or trajectories) can only be unambiguously distinguished if $\ket{\psi}$ is a TS state and $\gamma$ is subsequently diagonal.

Although single TS states form $[[n,0]]$ quantum codes which can precisely identify trajectory perturbations from the channel $\mathcal{E}_{\rm TS}$, these $k=0$ codes cannot store a logical qubit. In the next section, we show how to construct codes from TS states which properly encode a logical qubit.

\subsection{Trajectory sensing codes}\label{sec:TS-code}

There also exist $k>0$ codes such that each codeword state within the code space is a TS state. Recall that any code satisfying the KL criteria of Eq. (\ref{eq:KL}) with diagonal $\gamma$ can unambiguously distinguish the trajectory errors in $\mathcal{E}_{\mathrm{TS}}$. An $[[n,k]]$ code with orthonormal basis $\{|\psi_i\rangle \ :\ i=1,\ldots, 2^k\}$ satisfies Eq. (\ref{eq:KL}) for $\mathcal{E}_{\rm TS}$ and $\gamma_{ab}=\delta_{a,b}/\abs{\mathcal{T}}$ if and only if
\begin{align}\label{eq:TS-code}
    \bra{\psi_i}\rbf\ket{\psi_j}=\delta_{T,T'}\delta_{i,j}
\end{align}
for all $T,T'\in\mathcal{T}$. We call any code satisfying Eq. (\ref{eq:TS-code}) a \textit{TS code} since all of the states within are TS states.

We now show how to construct a $[[n,1]]$ TS code which encodes a single logical qubit. A basis for this code will consist of two orthogonal TS states. Recall by Theorem \ref{thm:exist} that there exist TS states in the subspace invariant under a permutation matrix group $\mathcal{G}$ and a Pauli subgroup $\mathcal{S}=\{I^{\otimes n}, X^{\otimes n}\}$; we now show that there also exist TS states instead invariant under $\mathcal{G}$ and $\mathcal{S}=\{I^{\otimes n}, -X^{\otimes n}\}$. Subsequently, define $\tilde{\mathcal{G}}_+$ and $\tilde{\mathcal{G}}_-$ to be the group $\langle\mathcal{S},\mathcal{G}\rangle$ when $\mathcal{S}=\{I^{\otimes n}, X^{\otimes n}\}$ and $\{I^{\otimes n}, -X^{\otimes n}\}$, respectively. Importantly, we find that these $\tilde{\mathcal{G}}_+$-invariant and $\tilde{\mathcal{G}}_-$-invariant TS states are orthogonal and form the desired code basis. 

When $\mathcal{T}=\Tsym$ and $n$ is odd, a $\tilde{\mathcal{G}}_-$-invariant TS state exists for a particular $\theta$ if and only if a $\tilde{\mathcal{G}}_+$-invariant TS state exists, where $\mathcal{G}=P(\Sigma_n)$. We can understand this claim via the following argument. An orthogonal basis for the $\tilde{\mathcal{G}}_-$-invariant subspace is given by the vectors 
\begin{align}
    \ket{\overline{\nu}_-}=\left(\sum_{j_1\ldots j_n\in \mathcal{W}_\nu}\ket{j_1\ldots j_n}\right) -\left(\sum_{j'_1\ldots j'_n\in \mathcal{W}_{n-\nu}}\ket{j'_1\ldots j'_n}\right)
\end{align}
for all $\nu=0,\ldots,N-1$, noting that $X^{\otimes n}\ket{\overline{\nu}_-}=-\ket{\overline{\nu}_-}$. If $n$ is odd, then the $\tilde{\mathcal{G}}_+$-invariant basis vectors $\ket{\overline{\nu}}$ of Eq. (\ref{eq:sym-basis}) are related to the $\ket{\overline{\nu}_-}$ via a unitary transformation $\ket{\overline{\nu}}=Z_L\ket{\overline{\nu}_-}$, where $Z_L$ is the unitary operator that applies a $-1$ phase to each $Z$-eigenbasis state with a weight over $\floor{n/2}$. Specifically, 
\begin{align}
    Z_L\ket{j_1\ldots j_n} = (-1)^{\Theta(j_1\ldots j_n)}\ket{j_1\ldots j_n},
\end{align}
where $\Theta(j_1\ldots j_n)$ equals one if $j_1+\ldots+ j_n>\floor{n/2}$ and zero otherwise. It is furthermore easy to check that $\braket{\overline{\nu}|\overline{\nu}'_-}=0$ for all $\nu,\nu'$, which implies that the $\tilde{\mathcal{G}}_+$- and $\tilde{\mathcal{G}}_-$-invariant subspaces are orthogonal. However, if $n$ is even, then $\ket{\overline{\nu}_-}=0$ for $\nu=n/2$, and clearly this vector is not related to the nonzero $\ket{\overline{\nu}}$ via a unitary transformation. Thus, assuming that $n$ is odd, let $\ket{\psi_+} = \sum_{\nu}c_\nu\ket{\overline{\nu}}$ be a $\tilde{\mathcal{G}}_+$-invariant state, and define the $\tilde{\mathcal{G}}_-$-invariant state $\ket{\psi_-}=Z_L\ket{\psi_+}= \sum_{\nu}c_\nu\ket{\overline{\nu}_-}$. Then
\begin{align}\label{eq:ortho-antisym}
    \bra{\psi_-}\rbf\ket{\psi_-}&=\bra{\psi_+}Z_L^\dagger \rbf Z_L\ket{\psi_+}\nonumber\\
    &=\bra{\psi_+}\rbf\ket{\psi_+}
\end{align}
for all $T,T'\in\mathcal{T}$ since $Z_L$ and the $\rbf$ commute and $Z_L^\dagger Z_L=I$. It follows that $\ket{\psi_+}$ is a TS state if and only if $\ket{\psi_-}$ is. 

For odd $n$ and a given $\theta$, after defining $\ket{+_L}$ to be a $\tilde{\mathcal{G}}_+$-invariant TS state and letting $\ket{-_L}=Z_L\ket{+_L}$, the $[[n,1]]$ code $C_{\rm TS} = \{\ket{+_L},\ket{-_L}\}$ is a TS code. Since $\ket{+_L}$ and $\ket{-_L}$ are TS states by the above argument,
\begin{align}
    \bra{+_L}\rbf\ket{+_L} = \bra{-_L}\rbf\ket{-_L} =\delta_{T,T'}
\end{align}
for all $T,T'\in\mathcal{T}$. To show that the TS code criteria of Eq. (\ref{eq:TS-code}) hold, it remains to show that $\bra{+_L}\rbf\ket{-_L}=\bra{-_L}\rbf\ket{+_L}=0$ for all $T,T'\in\mathcal{T}$. Note for any $T,T'\in\mathcal{T}$ that there exists a bijection between the elements of $T$ and the elements of $T'$ since $\abs{T}=\abs{T'}=m$. Hence, there is a permutation $\pi\in\Sigma_n$ such that $\pi(T)=T'$ and $\pi(T')=T$. Then, 
\begin{align}
    \bra{+_L}\rbf\ket{-_L} &= \bra{+_L}P_\pi\rbf P^\dagger_\pi\ket{-_L}\nonumber\\
    &= \bra{+_L}\mathbf{R}^{\pi(T,T')}\ket{-_L}\nonumber\\
     &= \bra{+_L}\mathbf{R}^{(\pi(T),\pi(T'))}\ket{-_L}\nonumber\\
    &= \bra{+_L}\mathbf{R}^{(T',T)}\ket{-_L}\nonumber\\
    &= \bra{+_L}\mathbf{R}^{[n](T,T')}\ket{-_L}\nonumber\\
    &= \bra{+_L}X^{\otimes n}\rbf X^{\otimes n}\ket{-_L}\nonumber\\
    &=-\bra{+_L}\rbf\ket{-_L},
\end{align}
where in the first equality we use the fact that $\ket{+_L},\ket{-_L}$ are permutation-invariant and in the fifth equality we invoke $X^{\otimes n}\ket{-_L}=-\ket{-_L}$. Necessarily, $\bra{+_L}\rbf\ket{-_L} = 0$; by an identical argument, $\bra{-_L}\rbf\ket{+_L}=0$ as well. We thus conclude that the code $C_{\rm TS}$ is a TS code, which implies that any state in $C_{\rm TS}$ is a TS state.

We now examine the logical qubit encoded in this code space. Define the logical qubit states $\ket{0_L}=\frac{1}{\sqrt{2}}(\ket{+_L}+\ket{-_L})$ and $\ket{1_L}=\frac{1}{\sqrt{2}}(\ket{+_L}-\ket{-_L})$. As mentioned above, both of these states are TS states. The logical operators for this code are then $Z_L$ and $X_L=X^{\otimes n}$, and it is easy to check that $X_L$ exchanges the states $\ket{0_L}$ and $\ket{1_L}$ while $Z_L$ applies a $-1$ phase to $\ket{1_L}$ only. Additionally, $X_L$ and $Z_L$ anticommute.

Although all of the states within $C_{\rm TS}$ are invariant under permutations by $\mathcal{G}=P(\Sigma_n)$, the Pauli symmetry group of $C_{\rm TS}$ is trivial, containing only the identity matrix. Recall that we have constructed $C_{\rm TS}$ from the state $\ket{+_L}$, which in contrast has nontrivial Pauli symmetry group $\{I^{\otimes n},X^{\otimes n}\}$. Hence, to create $C_{\rm TS}$ from the $k=0$ code $\{\ket{+}\}$, we remove the operator $X^{\otimes n}$ from the symmetry group of the latter, thereby increasing the code dimension by one. Furthermore, the removed operator $X^{\otimes n}$ becomes a logical operator of the larger code $C_{\rm TS}$. These relationships between symmetry operators, code dimension, and logical operators are already familiar in the study of stabilizer codes; note, however, that the $C_{\rm TS}$ provided here is not a stabilizer code in general. Given these observations, we are now motivated to examine the connection between stabilizer codes and TS codes in more detail.

\subsection{Stabilizer codes for trajectory sensing}\label{sec:stab}
Due to the large number of variables and equations typically involved, it is generally challenging to find codes which correct for a given error set by directly solving the KL criteria of (\ref{eq:KL}). The stabilizer formalism of QEC serves well to overcome these difficulties, and it has yielded many families of useful codes \cite{gottesman1997stabilizer}. A stabilizer code $\mathcal{H}_\mathcal{S}$ is defined to be the simultaneous $+1$ eigenspace of some Pauli subgroup $\mathcal{S}\leq \mathcal{P}_n$; recall that for $\mathcal{H}_{\mathcal{S}}$ to be nontrivial, $\mathcal{S}$ must be abelian and not contain the operator $-I^{\otimes n}$. Because a stabilizer code of dimension $K$ can be fully characterized by the $\log{K}$ generators of its stabilizer group, these codes admit error correcting criteria which are generally much easier to verify.

As seen in Sections \ref{sec:theory}-\ref{sec:simp-criteria}, it is similarly difficult to naively solve the TS criteria of Eq. (\ref{eq:ortho}) without simplification. To surmount this obstacle, we have likewise invoked symmetry: analogously to the stabilizer formalism, Theorem \ref{lemma:SG-simp} allows the TS criteria to be simplified for TS states invariant under certain symmetries. The first relevant symmetry is a permutation group $G$, which arises naturally due to the assumption that each $R^{(T)}$ ``error" operator is a tensor product of $R_Z(\theta)$ operators. However, the second relevant symmetry is in fact also a Pauli stabilizer group $\mathcal{S}$.

The fact that the TS states constructed in Sections \ref{sec:theory}-\ref{sec:results} are (1) quantum codes and (2) invariant under a Pauli subgroup $\mathcal{S}$ strongly suggests that there should exist stabilizer codes which are also TS codes for some particular set of trajectories. It is important to note, however, that a TS state invariant under some $\mathcal{S}$ is not necessarily a stabilizer state, since it need not be the unique state in $\mathcal{H}_\mathcal{S}$. Nonetheless, for a given $\theta$ and $\mathcal{T}$, it is sometimes possible for $\mathcal{H}_{\mathcal{S}}$ to be completely contained within the subspace of TS states; in this scenario, the stabilizer code $\mathcal{H}_\mathcal{S}$ may indeed also constitute a TS code. To visualize this situation, compare Figure \ref{fig:venn} to Figure \ref{fig:venn2}. In the former, some of the states in $\mathcal{H}_\mathcal{S}$ are not TS states, which prohibits the stabilizer code $\mathcal{H}_\mathcal{S}$ from being a TS code. In contrast, in the latter, the stabilizer code $\mathcal{H}_\mathcal{S}$ contains only TS states and may therefore also be a TS code.

\begin{figure}[htbp]
  \includegraphics[width = 0.6\linewidth]{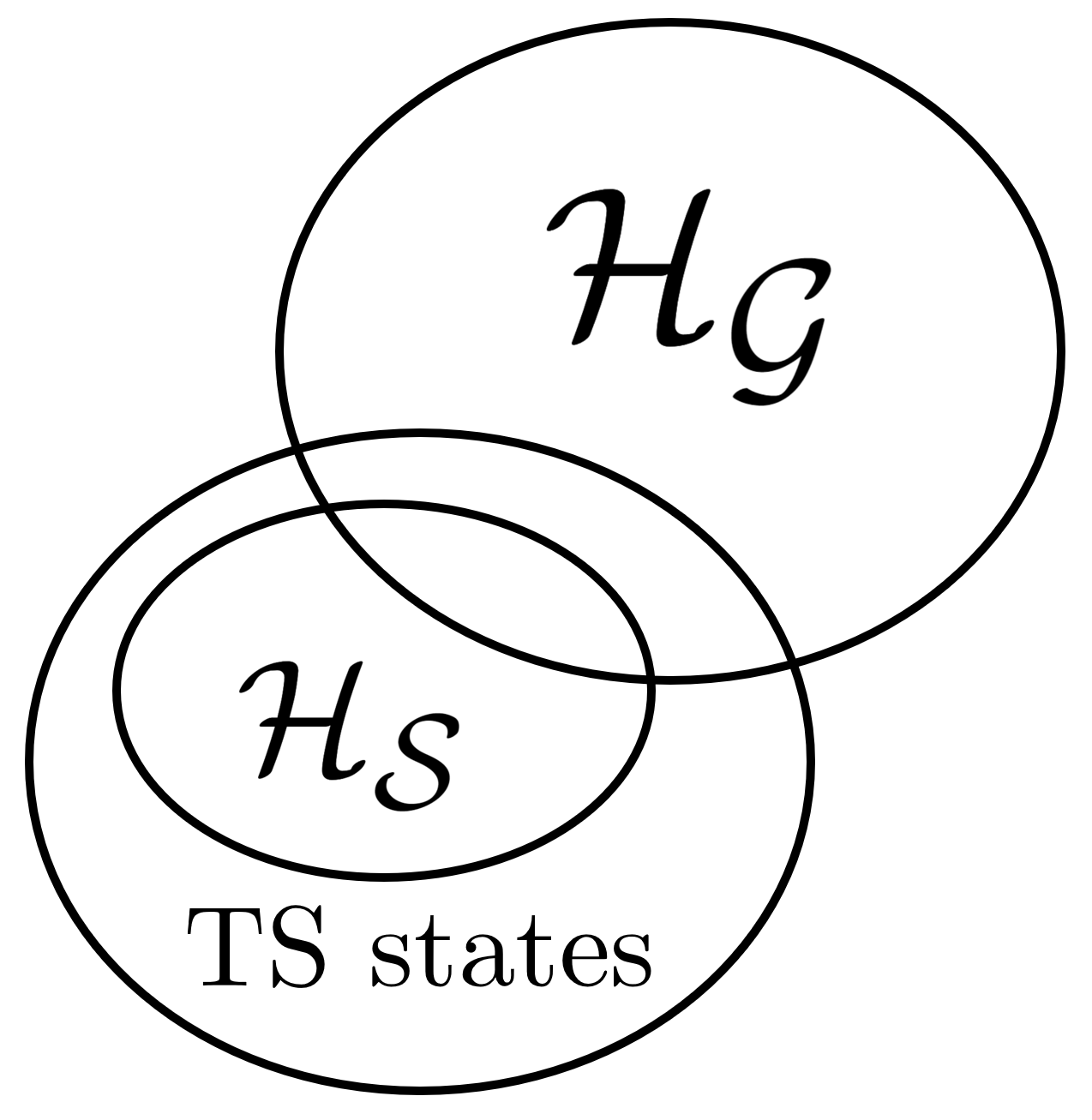}
  \caption{Venn diagram illustrating the space of TS states and the invariant subspaces $\mathcal{H}_\mathcal{S}$ and $\mathcal{H}_\mathcal{G}$ for a Pauli stabilizer $\mathcal{S}$ and permutation matrix group $\mathcal{G}$. In this situation, the space stabilized by $\mathcal{S}$ is completely contained within the space of TS states for the given $\theta$, which implies that $\mathcal{H}_\mathcal{S}$ may be a stabilizer TS code.}
  \label{fig:venn2}
\end{figure}

In the next three subsections, we first establish general criteria which a stabilizer code must satisfy for it to also serve as a TS code. We then show how these criteria may be applied by providing two small examples of stabilizer TS codes. Lastly, we demonstrate that certain code states of the well-known toric code \cite{toric} are TS states.

\subsubsection{Criteria for stabilizer TS codes}

Given a set of Pauli errors, the error correcting criteria for a stabilizer code can be succinctly stated in terms of the $k$ stabilizer generators. Consider a stabilizer code $\mathcal{H}_{\mathcal{S}}$ with stabilizer $\mathcal{S}$ and orthonormal basis $\{\ket{\psi_i}:i=1,\ldots,2^k\}$, where and define $\{U_i,U_j\}=U_iU_j+U_jU_i$ to be the anticommutator of two operators $U_i$ and $U_j$. Then, two Pauli errors $E_a,E_b\in\mathcal{P}_n$ are distinguishable via a single syndrome measurement if $\{D, E^\dagger_aE_b\} = 0$ for some generator $D$ of $\mathcal{S}$, since $E_a\ket{\psi_i}$ and $E_b\ket{\psi_j}$ are subsequently orthogonal for any basis vectors $\ket{\psi_i},\ket{\psi_j}$:
\begin{align}
    \bra{\psi_i}E^\dagger_aE_b\ket{\psi_j} &= \bra{\psi_i}DE^\dagger_aE_b\ket{\psi_j}\nonumber\\
    &=-\bra{\psi_i}E^\dagger_aE_bD\ket{\psi_j}\nonumber\\
    &=-\bra{\psi_i}E^\dagger_aE_b\ket{\psi_j} = 0.
\end{align}
Hence, if $E_a^\dagger E_b$ anticommutes with some generator of $\mathcal{S}$ for every $E_a, E_b$ in the given Pauli error set, then $\mathcal{H}_\mathcal{S}$ can unambiguously identify and correct for every error in the set. Verifying the error correcting criteria through these anticommutation conditions is typically much simpler than directly checking the KL conditions of Eq. (\ref{eq:KL}) since the number of stabilizer generators is only logarithmic in the number of code basis states. 

However, it is not generally possible to check whether some $\mathcal{H}_\mathcal{S}$ is a TS code using these anticommutation conditions because the trajectory ``errors" $R^{(T)}(\theta)$ are not Pauli operators when $\theta \in (0,\pi)$. In particular, for non-Pauli $R^{(T)}(\theta)$ operators, there cannot exist a generator of $\mathcal{S}$ which anticommutes with $R^{\dagger(T)}R^{(T')}=\rbf$. Nonetheless, there do indeed exist stabilizer TS codes for $\theta\in(0,\pi)$ even when these anticommutation conditions are not satisfied; we discuss multiple concrete examples in the next two subsections. Although these stabilizer codes can be verified to be TS codes by directly checking the criteria of Eq. (\ref{eq:TS-code}) for each code basis state, it still remains desirable to find simpler equivalent criteria in terms of the stabilizer generators.

Fortunately, the anticommutation conditions described above can be generalized to provide a sufficient condition for stabilizer TS codes even when $\theta\in(0,\pi)$. Given $\mathcal{T}$, it is too stringent to require that each $\rbf$ operator fully anticommutes with some stabilizer generator, as explained below. In fact, so long as each $\rbf$ operator anticommutes with some $D\in\mathcal{S}$ \textit{when both operators are projected to a subspace containing the code}, $\mathcal{H}_\mathcal{S}$ is indeed guaranteed to be a TS code, per the following theorem: 

\begin{theorem}\label{thm:stab}
    Let $\mathcal{H}_{\mathcal{S}}$ be a stabilizer code with Pauli stabilizer $\mathcal{S}$. Given a trajectory set $\mathcal{T}$ and $\theta\in[0,\pi]$, 
    then $\mathcal{H}_{\mathcal{S}}$ is a TS code if for every $T\neq T'$ in $\mathcal{T}$, there exists a subspace $V\supseteq\mathcal{H}_\mathcal{S}$ such that
    some $D\in\mathcal{S}$ satisfies
    \begin{align}\label{eq:anticom-proj}
        \{D,\rbf(\theta)\}\Pi_V=0,
    \end{align}
    where $\Pi_V$ is the projector onto $V$.
\end{theorem}
\begin{proof}
    Let $\{\ket{\psi_i}\}$ be an orthonormal basis for $\mathcal{H}_\mathcal{S}$. Clearly, $\bra{\psi_i}\rbf\ket{\psi_j} = \delta_{i,j}$ for any $i,j$ if $T=T'$. Thus, choose any $T,T'\in\mathcal{T}$ such that $T\neq T'$, and assume Eq. (\ref{eq:anticom-proj}) holds for some $D\in\mathcal{S}$ and $V\supseteq\mathcal{H}_\mathcal{S}$. We must show that $\bra{\psi_i}\rbf\ket{\psi_j}=0$ for all $i,j$. Eq. (\ref{eq:anticom-proj}) implies that $\rbf D\Pi_V = -D\rbf\Pi_V$. Also, note that every $\ket{\psi_i}\in V$ since $V\supseteq \mathcal{H}_\mathcal{S}$. Hence,
    \begin{align}
        \bra{\psi_i}\rbf\ket{\psi_j}&= \bra{\psi_i}\rbf D\ket{\psi_j}\nonumber\\
        &=\bra{\psi_i}\rbf D \Pi_V\ket{\psi_j}\nonumber\\
        &=-\bra{\psi_i}D \rbf  \Pi_V\ket{\psi_j}\nonumber\\
        &=-\bra{\psi_i}D \rbf \ket{\psi_j}\nonumber\\
        &=-\bra{\psi_i}\rbf \ket{\psi_j}\nonumber\\
        &=0,
    \end{align}
    as desired. We conclude that $\mathcal{H}_\mathcal{S}$ is a TS code.
\end{proof}

These anticommutation relations of Theorem \ref{thm:stab} in fact generalize the standard QEC criteria for a stabilizer code with Pauli errors. Specifically, when $\theta=\pi$ so that each ``error" $R^{(T)}(\theta)$ is a tensor product of Pauli operators, the standard criteria are recovered from the above theorem by choosing $V$ to be the whole Hilbert space for every $(T,T')$. However, we highlight the fact that $V$ may be chosen differently for each $(T,T')$, which is especially useful when $\theta<\pi$. Thus, $V$ may be understood as the decoding space for the corresponding pair of trajectory errors $R^{(T)}$ and $R^{(T')}$.

Theorem \ref{thm:stab} turns out to be particularly valuable for determining whether a CSS code \cite{css} is a TS code for a given $\mathcal{T}$ and $\theta$. Consider a CSS code with stabilizer $\mathcal{S}$. Note that a generating set for $\mathcal{S}$ can be found as the union of some sets $\mathcal{S}_X$ and $\mathcal{S}_Z$ whose elements are tensor products of only $\{I,X\}$ or $\{I,Z\}$, respectively. When a CSS code is a TS code, it is often possible in practice to find some $D\in\mathcal{S}_X$ which anticommutes with a given $\rbf$ when projected to a space stabilized by a subset of $\mathcal{S}_Z$. In the next subsection, we present two small CSS codes which are also TS codes and deploy Theorem \ref{thm:stab} to validate them.

\subsubsection{$C_4$ and $C_6$ codes}\label{sec:C4}
The $[[4,2]]$ $C_4$ code with stabilizer $\langle XXXX,ZZZZ\rangle$ is the smallest stabilizer code capable of detecting an arbitrary single-qubit Pauli error \cite{c4c6}. This CSS code encodes two logical qubits into the following code states 
\begin{align}
    \ket{00_L} &= \ket{0000} + \ket{1111}\nonumber\\
    \ket{01_L}&=\ket{0011}+\ket{1100}\nonumber\\
    \ket{10_L}&=\ket{0101}+\ket{1010}\nonumber\\
    \ket{11_L}&=\ket{0110}+\ket{1001},
\end{align}
where we have suppressed normalization for convenience. 

Remarkably, the $\cts$ state which discriminates $\mathcal{T}_{\rm cyc}(n=4,m=2)$ when $\theta=\frac{\pi}{2}$ (constructed in Section \ref{sec:cyc}) resides in this code space. Recall that this $\cts$ state $\ket{\psi_{\rm cyc}}$ is prepared by initializing both qubit pairs $\{1,3\}$ and $\{2,4\}$ into the two-qubit $\sts$ state  $\ket{01}+\ket{10}$, per Eq. (\ref{eq:example-cts}). It follows that 
\begin{align}
    \ket{\psi_{\rm cyc}}&=\ket{0011}+\ket{1100}+\ket{0110}+\ket{1001}\nonumber\\
    &=\ket{01_L}+\ket{11_L}.
\end{align}
The stabilizer of the one-dimensional subcode $\{\ket{\psi_{\rm cyc}}\}$ is generated by
\begin{align}\label{eq:C4}
    \langle -Z^{(13)},-Z^{(24)},X^{(13)},X^{(24)}\rangle,
\end{align}
where the notation $U^{(j_1\ldots j_k)}$ represents the operator which applies $U$ to qubits $j_1$ through $j_k$.
Besides its ability to discriminate trajectories, this subcode is further known to correct a single amplitude damping error \cite{amplitude-damp}.

Although we have already used Theorem \ref{lemma:SG-simp} to prove that $\ket{\psi_{\rm cyc}}$ is a TS state in Section \ref{sec:cyc}, we can now corroborate this result using Theorem \ref{thm:stab} instead. Specifically, we prove that the code $\{\ket{\psi_{\rm cyc}}\}$ is a TS code when $\mathcal{T}=\mathcal{T}_{\rm cyc}(n=4,m=2)$ and $\theta = \frac{\pi}{2}$. Suppose $T=\{1,2\}$ and $T'=\{2,3\}$, and let $\mathcal{S}$ be the stabilizer of $\mathcal{H}_\mathcal{S}= \{\ket{\psi_{\rm cyc}}\}$ given by Eq. (\ref{eq:C4}). Note that since $\rbf$ is not Pauli, it will not fully anticommute with any element of $\mathcal{S}$. We must consequently turn to the generalized criteria of Theorem \ref{thm:stab}, choosing $D=X^{(13)}$ and $V$ to be the $+1$ eigenspace of $-Z^{(13)}$; observe that $V\supseteq\mathcal{H}_\mathcal{S}$ since $-Z^{(13)}\in\mathcal{S}$. The projector onto $V$ is then given by $\Pi_V = (I-Z^{(13)})/2$. Then, to invoke Theorem \ref{thm:stab}, we must show that
\begin{align}\label{eq:C4-anticom}
    \{X^{(13)},\rbf\}(I-Z^{(13)})= 0.
\end{align}
For convenience, let $R=R_Z(\pi/2)$. Note that $R^2 = -iZ$ and $XRX = R^\dagger$, which implies that $RXZ = iXR$ and $R^\dagger X Z = -iXR^{\dagger}$. Furthermore, observe that $\rbf = R^{\dagger(1)}R^{(3)}$. It follows that  
\begin{align}
    &\{X^{(13)},\rbf\}Z^{(13)}\nonumber\\
    &= \left(XR^{\dagger}Z\right)^{(1)}\left(XRZ\right)^{(3)}+\left(R^{\dagger}XZ\right)^{(1)}\left(RXZ\right)^{(3)}\nonumber\\
    &= \left(RXZ\right)^{(1)}\left(R^{\dagger}XZ\right)^{(3)}+\left(R^{\dagger}XZ\right)^{(1)}\left(RXZ\right)^{(3)}\nonumber\\
    &= \left(iXR\right)^{(1)}\left(-iXR^{\dagger}\right)^{(3)}+\left(-iXR^{\dagger}\right)^{(1)}\left(iXR\right)^{(3)}\nonumber\\
    &= \left(XR\right)^{(1)}\left(XR^{\dagger}\right)^{(3)}+\left(XR^{\dagger}\right)^{(1)}\left(XR\right)^{(3)}\nonumber\\
    &= \left(R^\dagger X\right)^{(1)}\left(R X\right)^{(3)}+\left(XR^{\dagger}\right)^{(1)}\left(XR\right)^{(3)}\nonumber\\
    &=\{X^{(13)},\rbf\}.
\end{align}
Since $\{X^{(13)},\rbf\}Z^{(13)}=\{X^{(13)},\rbf\}$, Eq. (\ref{eq:C4-anticom}) must hold for this choice of $T,T'$. 
A similar argument applies for all other $T\neq T'$, so Theorem \ref{thm:stab} confirms that $\mathcal{H}_\mathcal{S}=\{\ket{\psi_{\rm cyc}}\}$ is indeed the desired TS code. 

The $[[6,2]]$ $C_6$ code, a close analogue of the $C_4$ code, is also a CSS code but with  stabilizer $\langle XIXXIX,IXXIXX,ZZIZZI,IZZIZZ\rangle$ up to a permutation of the qubits \cite{c4c6}. Similarly, the $\cts$ state which discriminates $\mathcal{T}_{\rm cyc}(n=6, m=3)$ when $\theta=\frac{\pi}{2}$ (prepared by initializing qubit pairs $\{1,4\}$, $\{2,5\}$, and $\{3,6\}$ to the smaller $\sts$ state $\ket{\phi}=\ket{01}+\ket{10}$) lies within the $C_6$ code. The stabilizer for the subcode containing this $\cts$ state is generated by 
\begin{align}\label{eq:C6}
    \langle -Z^{(14)},-Z^{(25)},-Z^{(36)},X^{(14)},X^{(25)},X^{(36)}\rangle.
\end{align}
Recall that in Section \ref{sec:cyc} we have already proven that the state stabilized by Eq. (\ref{eq:C6}) is a TS state using Theorem \ref{lemma:SG-simp}. In principle, it is possible to equivalently verify this result using Theorem \ref{thm:stab} as we did above for the $C_4$ TS state, but we will omit the details here.

Given that the $C_4$ and $C_6$ are small stabilizer codes, it is natural to ask whether there might exist larger stabilizer codes which also yield useful TS states. From a practical point of view, trajectory sensors built from surface codes \cite{surface} would be particularly desirable, as such codes have demonstrated particularly promising performance in recent experiments; namely, these codes have remarkably achieved decreasing error rates with increasing code size \cite{googlesurface,googlesurface2}. Subsequently, we next consider toric codes \cite{toric}, a kind of surface code with periodic boundary conditions, for trajectory sensing. Note that the $C_4$ code introduced above is in fact the smallest example of a toric code. We thus now show that a larger toric code can also support TS states.

\subsubsection{$8$-qubit toric code}\label{sec:toric}
The general toric code is defined on a two-dimensional square lattice with periodic boundary conditions such that a qubit is located on each edge \cite{toric}. The stabilizer operators for the code are generated by tensor products of Pauli operators on the qubits around each vertex and plaquette of the lattice:
\begin{align}
    A_v &= \prod_{i\in v} X_i,\ B_p = \prod_{j\in p}Z_j
\end{align}
where $i\in v$ designates the set of four qubits surrounding vertex $v$ and $j\in p$ denotes the set of four qubits on the boundary of plaquette $p$. The toric code is the simultaneous $+1$ eigenspace of the group generated by the $A_v$ and $B_p$ operators for all vertices and plaquettes $v$ and $p$ in the lattice. 

\begin{figure}[htbp]
  \includegraphics[width = 0.5\linewidth]{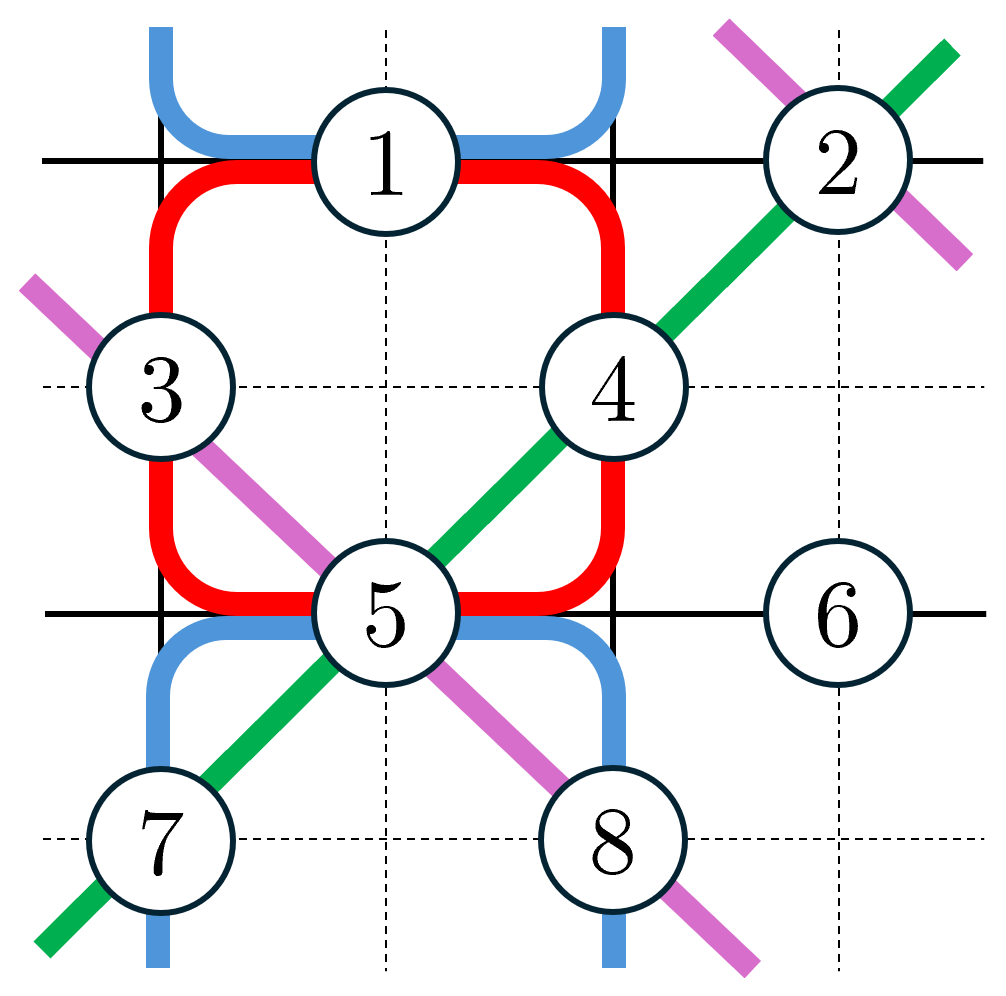}
  \caption{An $8$-qubit toric code contains a TS state which can discriminate the four trajectories in $\mathcal{T}_{\rm toric}=\{T_1,T_2,T_3,T_4\}$ when $\theta=\frac{\pi}{2}$, where $T_1=\{1,3,4,5\}$ (red), $T_2=\{5,7,8,1\}$ (blue), $T_3=\{2,4,5,7\}$ (green), and $T_4=\{3,5,8,2\}$ (purple).} 
  \label{fig:toric}
\end{figure}

An $8$-qubit toric code contains a TS state $\ket{\psi_{\rm toric}}$ which can distinguish the four trajectories in $\mathcal{T}_{\rm toric}$ (depicted in Figure \ref{fig:toric}) when $\theta=\frac{\pi}{2}$. The full toric code has the following $8$ stabilizer generators:
\begin{gather}
    \langle X^{(1237)}, X^{(1248)},X^{(3567)},X^{(4568)},\nonumber\\
    Z^{(1345)},Z^{(2346)},Z^{(5781)},Z^{(6782)}\rangle,\label{eq:toric-stab}
\end{gather}
where the first four are $A_v$ operators and the second four are $B_p$ operators.
Note that only $6$ of these generators are independent because the product of all $A_v$ and the product of all $B_p$ give the identity. The desired TS state $\ket{\psi_{\rm toric}}$ constitutes the one-dimensional subcode stabilized by the generators in Eq. (\ref{eq:toric-stab}) along with the additional generators $-Z^{(12)},-Z^{(37)}$. These additional generators can be verified to commute with all other generators of the code. 

To show that $\ket{\psi_{\rm toric}}$ can indeed distinguish the four trajectories of $\mathcal{T}_{\rm toric}$ when $\theta=\frac{\pi}{2}$, we now have a number of tools available. 
We could, for example, verify the anticommutation relations of Theorem \ref{thm:stab} for every pair of trajectories. However, supposing there exists a permutation group $G$ and Pauli subgroup $\mathcal{S}$ under which $\ket{\toric}$ and $\mathcal{T}^2_{\rm toric}$ are simultaneously invariant, we could also augment this approach with our earlier result of Theorem \ref{lemma:SG-simp}. In particular, instead of checking the anticommutation relations for every single trajectory pair, it would suffice to only validate them for one trajectory pair per orbit of $\mathcal{T}^2_{\rm toric}/(S_G\rtimes G)$, where $S=\sigma(\mathcal{S})$. This combined approach greatly reduces the number of anticommutation relations to evaluate, and we pursue this strategy in Appendix \ref{appendix:toric} to prove that $\ket{\toric}$ is indeed a TS state under the given conditions.

\subsection{Enhancement of TS states via code concatenation}\label{sec:noise}

Given the close connection between TS states and stabilizer codes, it is intriguing to consider what possibilities might be enabled through their combination. The principle of \textit{code concatenation} provides a meaningful way to amalgamate a TS state with another quantum code; two codes are said to be concatenated if one is constructed from the logical states of the other. In particular, suppose we have an $[[n,k]]$ code which encodes $k$ logical qubits into $n$ physical qubits. If we then initialize the $k$ logical qubits into a TS state, we can imagine that any perturbations on the physical qubits which realize trajectory-like perturbations on the encoded logical qubits might be distinguishable via a projective measurement. We now show how concatenating a TS state with another code in this manner may enhance some desirable properties of the TS state. 

For instance, by concatenating a TS state with a repetition code, it is possible in principle to push the achievable $\theta$ for the TS state arbitrarily close to zero. It would be practically useful for a TS state to be able to unambiguously discriminate trajectories even when the particle-sensor interaction strength $\theta$ is very weak. Let $\ket{\psi}$ be any $n$-qubit TS state that distinguishes a trajectory set $\mathcal{T}$ at a given $\theta$, where each trajectory in $\mathcal{T}$ is of size $m$. Furthermore, define an $n'$-qubit repetition code to be the code spanned by the logical qubits $\ket{0_L}=\ket{0}^{\otimes n'}$ and $\ket{1_L}=\ket{1}^{\otimes n'}$. If the logical qubits of $n$ different blocks of the repetition code are prepared to the state $\ket{\psi}$, then the resulting $nn'$-qubit concatenated state is a TS state for a new trajectory set $\mathcal{T}'$, but at the decreased interaction strength $\theta/n'$. Note that each trajectory in $\mathcal{T}'$ is of size $mn'$ and consists of $m$ whole code blocks. 

We illustrate this claim with the following simple example. Given some positive integer $n'$, suppose we have $2n'$ qubits that we prepare into two repetition code blocks: qubits $\{1,\ldots,n'\}$ form the first block and qubits $\{n'+1,\ldots 2n'\}$ form the second. These two code blocks encode two logical qubits, which we respectively label $\bar{1}$ and $\bar{2}$. If we prepare the logical qubits to the $\sts$ state $\ket{\psi_{\rm sym}}=\frac{1}{\sqrt{2}}(\ket{01_L}+\ket{10_L})$, then the logical trajectories $\bar{T} = \{\bar{1}\}$ and $\bar{T}'=\{\bar{2}\}$ can be perfectly distinguished with a single projective measurement when $\theta =\frac{\pi}{2}$. Importantly, note that applying $R_Z(\theta/n')$ to every physical qubit in a code block has the effect of rotating the corresponding logical qubit by $R_Z(\theta)$. Hence, the logical trajectories $R^{(\bar{T})}(\theta)$ and $R^{(\bar{T}')}(\theta)$ are respectively equivalent to the physical trajectories $R^{(T)}(\theta/n')$ and $R^{(T')}(\theta/n')$, where $T=\{1,\ldots,n'\}$ and $T'= \{n'+1,\ldots 2n'\}$. It follows that the concatenated TS state can distinguish the physical trajectories in $\mathcal{T}'=\{T,T'\}$ when the interaction strength is $\theta/n'=\frac{\pi}{2n'}$. 

Moreover, since $n'$ can be chosen arbitrarily large, the achievable interaction strength for this new TS state can in theory be pushed arbitrarily close to zero. However, in a real experiment, issues may arise when $n'$ becomes very large. In particular, note that while the incident particle perturbs the sensor qubits, the sensor qubits, in turn, must conversely perturb the particle. Hence, if the particle interacts with a vast number of qubits, the resulting back-action may measurably alter its trajectory in an undesirable way, thus corrupting the trajectory data. This phenomenon may thereby impose a practical limit on the minimum $\theta$ which could be achievable via concatenation with repetition codes.

On the other hand, if a TS state is instead concatenated with a code that has error-correction capabilities, then it may be possible to discriminate perturbations even if the physical qubits are subject to noise. However, it is not immediately clear for a given code what kinds of physical qubit perturbations could lead to distinguishable logical qubit ``trajectories." 

A good candidate error correcting code for concatenation with a TS state would be any code that implements the $R_Z(\theta)$ gate transversally on the logical space by applying $R_Z(\theta)$ to a subset of the physical qubits in the code. The Steane code, which encodes one logical qubit into 7 physical qubits, is a prime example, as applying $R_Z\left(\frac{\pi}{2}\right)$ to all 7 physical qubits has the effect of rotating the logical qubit by $R^\dagger_Z\left(\frac{\pi}{2}\right)$ \cite{steane}. Now let $\ket{\psi}$ be any $n$-qubit TS state that distinguishes a trajectory set $\mathcal{T}$ at $\theta=\frac{\pi}{2}$, where each trajectory in $\mathcal{T}$ is of size $m$. If we prepare the logical qubits of $n$ different Steane code blocks into the state $\ket{\psi}$, then the resulting $7n$-qubit concatenated code is a TS state for a new trajectory set $\mathcal{T}'$ at $\theta=\frac{\pi}{2}$, where each trajectory in $\mathcal{T}'$ is of size $7m$ and consists of $m$ whole code blocks.

We support this claim with an analogous toy example, illustrated in Figure \ref{fig:steane}. We prepare 14 physical qubits into two Steane code blocks such that qubits 1-7 constitute the first block and qubits 8-14 constitute the second. As before, the two code blocks respectively encode two logical qubits labeled $\bar{1}$ and $\bar{2}$ whose states are either $\ket{0_L}$ or $\ket{1_L}$. Like in the previous example, we prepare the logical qubits to the same state $\ket{\psi_{\rm sym}}$, which allows the logical trajectories $\bar{T} = \{\bar{1}\}$ and $\bar{T}'=\{\bar{2}\}$ to be perfectly distinguished when $\theta=\frac{\pi}{2}$. However, since the Steane code transversally implements the $R_Z\left(\frac{\pi}{2}\right)$ gate, these logical trajectories are again equivalent to physical trajectories:
\begin{align}
    R^{(\bar{T})}\left(\frac{\pi}{2}\right)\ket{\psi} = R^{\dagger({T})}\left(\frac{\pi}{2}\right)\ket{\psi}\ \mathrm{and}\\
    R^{(\bar{T}')}\left(\frac{\pi}{2}\right)\ket{\psi} = R^{\dagger({T}')}\left(\frac{\pi}{2}\right)\ket{\psi},
\end{align}
where $T=\{1,\ldots 7\}$ and $T'=\{8,\ldots 14\}$. Thus, this concatenated TS code can distinguish the two physical trajectories in $\mathcal{T}'=\{T, T'\}$ when $\theta=\frac{\pi}{2}$. 

\begin{figure}[htbp]
  \includegraphics[width = 0.8\linewidth]{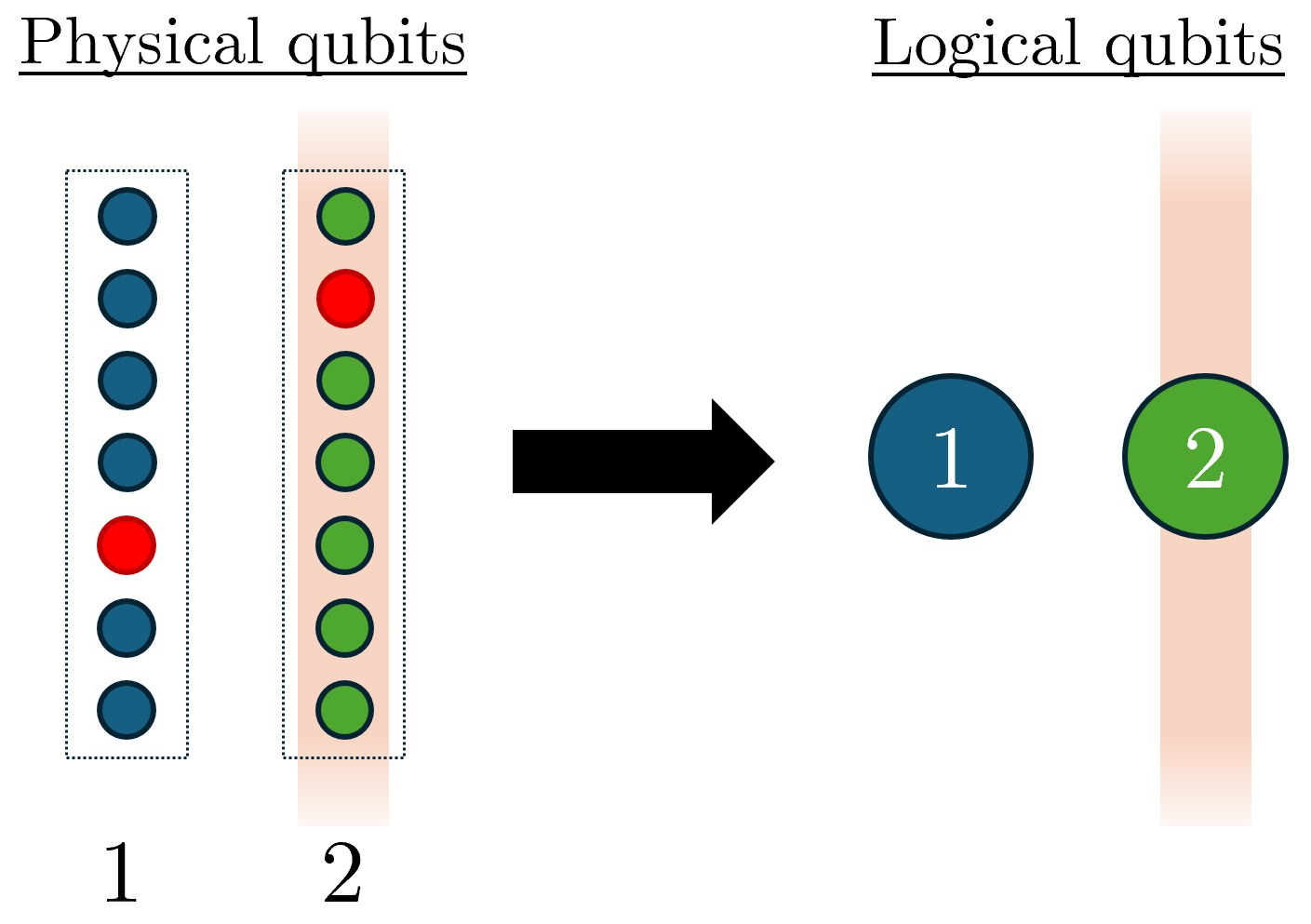}
  \caption{Noise-resilient TS state constructed by preparing the logical qubits of two Steane code blocks into the $\sts$ state $\ket{\psi_{\rm sym}}=\frac{1}{\sqrt{2}}(\ket{01_L} + \ket{10_L})$. A trajectory which passes through one block of physical qubits induces a trajectory through the corresponding logical qubit. Qubits in green have been rotated by $R_Z(\pi/2)$. Because the Steane code blocks can each correct arbitrary single-qubit errors, perfect trajectory sensing is possible even if one qubit per block decoheres (e.g., the red qubits).} 
  \label{fig:steane}
\end{figure}

Crucially, the Steane code furthermore allows arbitrary single-qubit errors to be perfectly corrected. It follows that if an undesired single qubit error occurs in either code block before (or after) the incident particle interacts with the sensor, then syndrome measurements on the blocks can be used to restore the TS state (or output state) so that the trajectories can still be distinguished with zero probability of failure. Hence, concatenating TS states with error-correcting codes allows for trajectory sensing that is robust to external noise.

\section{Conclusion}\label{sec:conclusion}
In this paper, we formally introduced the TS problem, developed a group-theoretic framework for solving the problem, and provided various families of solutions. In particular, we showed how permutation and Pauli group symmetries naturally arise within the TS problem, and we applied these symmetries to substantially simplify the general criteria for TS states. We subsequently used these simplified criteria to determine closed-form bounds on the interval of achievable $\theta$ as a function of sensor size when the trajectory set is transitive under the symmetric and cyclic permutation groups. Finally, we established a concrete link between trajectory sensing and quantum error correction, demonstrating how familiar stabilizer codes can be used as trajectory sensors under suitable conditions.

A number of important and interesting questions remain open within the TS formalism developed here. For example, recall that Theorem \ref{thm:exist} guarantees the existence of TS states invariant under $\langle\mathcal{S},\mathcal{G}\rangle$ where $\mathcal{S}$ is the particular Pauli subgroup $\{I^{\otimes n},X^{\otimes n}\}$; however, it is not known how this result might generalize to other choices of $\mathcal{S}$. Additionally, it remains unproven whether the bound on the achievable $\theta$ for $\Tcyc$ given in Theorem \ref{thm:C-TS-criterion} is a necessary condition (as opposed to just a sufficient one). Furthermore, although we provided individual examples of stabilizer codes which provide TS states, it is unclear how to systematically construct a general family of stabilizer codes (possibly surface or toric codes) which can support trajectory sensing.

The foundation for quantum trajectory sensing established here could be expanded in many meaningful and exciting ways. For instance, it may be possible to create TS codes with intrinsic error-correcting capabilities which do not rely on concatenation; such codes may subsequently be more resource-efficient. 
On the other hand, a natural extension of our TS scenario might replace the sensor qubits with qudits. Although qudits can be incorporated into a trajectory sensor by simply concatenating our current TS code with a qubit-into-oscillator bosonic code (such as a GKP code \cite{GKP}), it is unknown whether there exist more powerful TS architectures built from qudits directly. Lastly, given that the toric codes of Section \ref{sec:toric} are ground states of a many-body Hamiltonian \cite{toric}, we ask whether similar Hamiltonians may provide a natural way to describe or physically realize TS states in general.

\begin{acknowledgements}
Z.E.C. acknowledges support from the National Science Foundation Graduate Research Fellowship under Grant No. 2141064.  This project was supported in part by the NSF Q-SAIL National Quantum Virtual Laboratory under Grant No. 2410687, and by the U.S. Department of Energy, Office of Science, National Quantum Information Science Research Centers, Co-design Center for Quantum Advantage (C$^2$QA) under contract number DE-SC0012704.    
\end{acknowledgements}
 

\onecolumngrid
\appendix
\counterwithin{proposition}{section}
\renewcommand{\theproposition}{\thesection\arabic{proposition}}
\section{Miscellaneous proofs}
In this appendix, we provide a number of miscellaneous proofs which have been deferred from the main text. 
\subsection{Conjugation of operators by qubit permutation matrices}\label{appendix:perm-matrix-conjugation}
In Section \ref{sec:theory}, we consider how $\rbf$ operators and Pauli operators transform under conjugation by qubit permutation matrices. In essence, if an operator $U$ can be written as a tensor product of single-qubit operators, then conjugating $U$ by a qubit permutation matrix $P_\pi$ permutes the indices of its tensor factors by $\pi^{-1}$:
\begin{proposition}\label{prop:permute-factors}
    Let $G$ be any permutation group. Suppose $U\in\unit$ can be written as a tensor product of single qubit operators $U_k$, so that $U=\bigotimes_{k=1}^n U_k$. Then $P_\pi U P_\pi^\dagger = U'$ for any $\pi\in G$,
    where $U'=\bigotimes_{k=1}^n U_{\pi^{-1}(k)}$.
\end{proposition}
\begin{proof}
    Each $U_k$ can be written in the $Z$-eigenbasis as 
    \begin{align}
        U_k = \sum_{l',l\in\{0,1\}}\left(u_k\right)_{l'l}\ketbra{l'}{l}
    \end{align}
    for some scalars $\left(u_k\right)_{l'l}\in\mathbb{C}$. It follows that 
    \begin{align}
        U &= \bigotimes_{k=1}^n\sum_{l_k',l_k\in\{0,1\}}\left(u_k\right)_{l_k'l_k}\ketbra{l_k'}{l_k}\\
        &=\sum_{l'_1,\ldots, l'_n\in\{0,1\}}\sum_{l_1,\ldots, l_n\in\{0,1\}}\left(\prod_{k=1}^n\left(u_k\right)_{l'_kl_k}\right)\ketbra{l'_1\ldots l'_n}{l_1\ldots l_n}.
    \end{align}
    To compute $P_\pi U P_\pi^\dagger$, we first compute $UP_\pi^\dagger$:
    \begin{align}
        UP_\pi^\dagger &= \left(\sum_{l'_1,\ldots, l'_n\in\{0,1\}}\sum_{l_1,\ldots, l_n\in\{0,1\}}\left(\prod_{k=1}^n\left(u_k\right)_{l'_kl_k}\right)\ketbra{l'_1\ldots l'_n}{l_1\ldots l_n}\right)\left(\sum_{j_1,\ldots,j_n\in\{0,1\}}\ketbra{j_{\pi(1)}\ldots j_{\pi(n)}}{j_1\ldots j_n}\right)\\
        &=\sum_{l'_1,\ldots, l'_n\in\{0,1\}}\sum_{l_1,\ldots, l_n\in\{0,1\}}\sum_{j_1,\ldots,j_n\in\{0,1\}}\left(\prod_{k=1}^n\left(u_k\right)_{l'_kl_k}\right)\ketbra{l'_1\ldots l'_n}{l_1\ldots l_n}\ketbra{j_{\pi(1)}\ldots j_{\pi(n)}}{j_1\ldots j_n}
    \end{align}
    We keep only terms where $l_1\ldots l_n = j_{\pi(1)}\ldots j_{\pi(n)}$:
    \begin{align}
        UP_\pi^\dagger = \sum_{l'_1,\ldots, l'_n\in\{0,1\}}\sum_{j_1,\ldots, j_n\in\{0,1\}}\left(\prod_{k=1}^n\left(u_k\right)_{l'_kj_\pi(k)}\right)\ketbra{l'_1\ldots l'_n}{j_1\ldots j_n}.
    \end{align}
    We can now left-multiply by $P_\pi$:
    \begin{align}
        P_\pi U P^\dagger_\pi &= \left(\sum_{j'_1,\ldots,j'_n\in\{0,1\}}\ketbra{j'_1\ldots j'_n}{j'_{\pi(1)}\ldots j'_{\pi(n)}}\right)\left(\sum_{l'_1,\ldots, l'_n\in\{0,1\}}\sum_{j_1,\ldots, j_n\in\{0,1\}}\left(\prod_{k=1}^n\left(u_k\right)_{l'_kj_\pi(k)}\right)\ketbra{l'_1\ldots l'_n}{j_1\ldots j_n}\right)\\
        &=\sum_{j'_1,\ldots,j'_n\in\{0,1\}}\sum_{l'_1,\ldots, l'_n\in\{0,1\}}\sum_{j_1,\ldots, j_n\in\{0,1\}}\left(\prod_{k=1}^n\left(u_k\right)_{l'_kj_\pi(k)}\right)\ketbra{j'_1\ldots j'_n}{j'_{\pi(1)}\ldots j'_{\pi(n)}}\ketbra{l'_1\ldots l'_n}{j_1\ldots j_n}
    \end{align}
    and again keep only terms where $j'_{\pi(1)}\ldots j'_{\pi(n)} = l'_1\ldots l'_n$:
    \begin{align}
        P_\pi U P^\dagger_\pi&=\sum_{j'_1,\ldots,j'_n\in\{0,1\}}\sum_{j_1,\ldots, j_n\in\{0,1\}}\left(\prod_{k=1}^n\left(u_k\right)_{j'_{\pi(k)}j_\pi(k)}\right)\ketbra{j'_1\ldots j'_n}{j_1\ldots j_n}.
    \end{align}
    We finally change indices $k\rightarrow \pi^{-1}(k')$ to obtain
    \begin{align}
        P_\pi U P^\dagger_\pi&=\sum_{j'_1,\ldots,j'_n\in\{0,1\}}\sum_{j_1,\ldots, j_n\in\{0,1\}}\left(\prod_{k'=1}^n\left(u_{\pi^{-1}(k')}\right)_{j'_{k'}j_{k'}}\right)\ketbra{j'_1\ldots j'_n}{j_1\ldots j_n}\\
        &= \bigotimes_{k'=1}^n\sum_{j_k',j_k\in\{0,1\}}\left(u_{\pi^{-1}(k')}\right)_{j'_{k'}j'_{k'}}\ketbra{j'_{k'}}{j_{k'}}\\
        &=\bigotimes_{k'=1}^nU_{\pi^{-1}(k')},
    \end{align}
    as desired.
\end{proof}
We can directly apply this proposition to describe how $R^{(T)}$ operators transform when conjugated by qubit permutation matrices:
\begin{proposition}\label{prop:P-R-conj}
    Let $G$ be a permutation group on $[n]$. For any $\pi\in G$ and $T\subseteq[n]$, $P_\pi R^{(T)}P^\dagger_\pi = R^{(\pi(T))}$.
\end{proposition}
\begin{proof}
    We have 
    \begin{align}
        P_\pi R^{(T)}P^\dagger_\pi&= P_\pi \left(\bigotimes_{j=1}^n R_Z(\theta\cdot\mathbbm{1}_T(j))\right)P^\dagger_\pi.
    \end{align}
    Proposition \ref{prop:permute-factors} then implies that conjugation by $P_\pi$ permutes the indices of the tensor factors by $\pi^{-1}$:
    \begin{align}
        P_\pi R^{(T)}P^\dagger_\pi&= \bigotimes_{j=1}^n R_Z(\theta\cdot\mathbbm{1}_T(\pi^{-1}(j))).
    \end{align}
    Note that $\pi^{-1}(j)\in T$ if and only if $j\in \pi(T)$. Hence, $\mathbbm{1}_T(\pi^{-1}(j))=\mathbbm{1}_{\pi(T)}(j)$, so
    \begin{align}
        P_\pi R^{(T)}P^\dagger_\pi&=\bigotimes_{j=1}^n R_Z(\theta\cdot\mathbbm{1}_{\pi(T)}(j))\\
        &= R^{(\pi(T))},
    \end{align}
    as desired.
\end{proof}
\subsection{Proof of Proposition \ref{prop:S-conj}}\label{appendix:pauli-action}
In this section, we first show that $\sigma$ is a valid surjective homomorphism from a $\mathcal{S}\leq\mathcal{P}_n$ to $S=\sigma(\mathcal{S})$ and then prove Proposition \ref{prop:S-conj}. 
\begin{proposition}\label{prop:sigma-homo}
    The map $\sigma$ is a surjective homomorphism from any $\mathcal{S}\leq \mathcal{P}_n$ to $S = \sigma(\mathcal{S})$.
\end{proposition}
\begin{proof}
    To prove that $\sigma$ is a homomorphism, we must show that $\sigma_{DD'} =\sigma_D\symdif\sigma_{D'}$ for any $D,D'\in\mathcal{S}$, where $\symdif$ is the set symmetric difference. Write $D$ and $D'$ as $D = e^{i\phi}\bigotimes_{j=1}^nD_j$ and $D' = e^{i\phi'}\bigotimes_{j=1}^nD'_j$ for some $\phi,\phi'\in[0,2\pi)$, where each $D_j,D'_j\in\{X,Y,Z,I\}$. Then
    \begin{align}
        DD' = \bigotimes_{j=1}^nD_jD'_j.
    \end{align}
    We now make the following observations about products of single-qubit Paulis, which hold up to a phase:
    \begin{enumerate}
        \item Products of two (possibly identical) operators from the set $\{X,Y\}$ are equal to one of $\{Z,I\}$.
        \item Products consisting of one element of $\{X,Y\}$ and one element of $\{Z,I\}$ are equal to an element of $\{X,Y\}$. 
        \item Products of two (possibly identical) operators from the set $\{Z,I\}$ are equal to one of $\{Z,I\}$.
    \end{enumerate}
    By definition, $\sigma_{DD'} = \{j: D_jD'_j= X\ \mathrm{or}\ Y\}$. Due to the observations above, $j\in\sigma_{DD'}$ if and only if either $D_j\in\{X,Y\}\ \mathrm{and}\ D'_j\in\{Z,I\}$ or $D_j\in\{Z,I\}\ \mathrm{and}\ D'_j\in\{X,Y\}$. Equivalently, $j\in\sigma_{DD'}$ if and only if $j\in\sigma_D$ and $j\notin\sigma_{D'}$ or $j\notin\sigma_D$ and $j\in\sigma_{D'}$. This result can be summarized as follows: $j\in\sigma_{DD'}$ if and only if $j\in\sigma_D\symdif\sigma_{D'}$, which implies that $\sigma_{DD'}=\sigma_D\symdif \sigma_{D'}$, as desired.

    The surjectivity of $\sigma:\mathcal{S}\rightarrow S$ follows trivially from the fact that $S$ is the image of $\mathcal{S}$ under $\sigma$ by definition.
\end{proof}
We now prove Proposition \ref{prop:S-conj}.
\begin{proof}[Proof of Proposition \ref{prop:S-conj}]
    Part (a). For any $D\in \mathcal{S}$ and $T,T'\subseteq[n]$,
    \begin{align}
        D\rbf D^\dagger &= D R^{\dagger(T)}R^{(T')} D^\dagger\\
        &= \left(D R^{\dagger(T)}D^\dagger\right)\left(D R^{(T')} D^\dagger\right)
    \end{align}
    We will use the notation $R^{(j)}$ to represent $R^{(\{j\})}$, that is, the $n$-qubit operator which applies $R_Z$ to the $j$th qubit and the identity to all other qubits. Additionally, write $D$ as $D = e^{i\phi}\bigotimes_{j=1}^nD_j$ for some $\phi\in[0,2\pi)$, where each $D_j\in\{X,Y,Z,I\}$. Accordingly, define $\tilde{D}_j$ to be the $n$-qubit operator which applies $D_j$ to the $j$th qubit and the identity to all other qubits. Then
    \begin{align}
        D\rbf D^\dagger &= \left(\prod_{j\in T}\tilde{D}_jR^{\dagger(j)}\tilde{D}^\dagger_j\right)\left(\prod_{l\in T'}\tilde{D}_{l}R^{(l)}\tilde{D}^\dagger_{l}\right)
    \end{align}
     Note that if $j\in\sigma_D$, then $D_j = X$ or $Y$, so $D_j R^\dagger_Z D^\dagger_j = R_Z$ and $D_j R_Z D^\dagger_j = R^\dagger_Z$. If instead $j\notin\sigma_D$, then $D_j=Z$ or $I$, so $D_j R^\dagger_Z D^\dagger_j = R^\dagger_Z$ and $D_j R_Z D^\dagger_j = R_Z$. Thus,
    \begin{align}
        D\rbf D^\dagger &= \left(\prod_{j\in T\setminus\sigma_D}R^{\dagger(j)}\right)\left(\prod_{j'\in T\cap\sigma_D}R^{(j')}\right)\left(\prod_{l\in T'\setminus\sigma_D}R^{(l)}\right)\left(\prod_{l'\in T'\cap\sigma_D}R^{\dagger(l')}\right)\label{eq:d-conj01}\\
        &= \left(\prod_{j\in (T\setminus\sigma_D)\cup(T'\cap\sigma_D)}R^{\dagger(j)}\right)\left(\prod_{l\in (T'\setminus\sigma_D)\cup(T\cap\sigma_D)}R^{(l)}\right).\label{eq:d-conj02}
    \end{align}
    Eq. (\ref{eq:d-conj02}) follows from Eq. (\ref{eq:d-conj01}) because $(T\setminus\sigma_D)$ and $(T'\cap\sigma_D)$ are disjoint (likewise, $(T'\setminus\sigma_D)$ and $(T\cap\sigma_D)$ are disjoint). We can then rewrite this result as 
    \begin{align}
        D\rbf D^\dagger &= R^{\dagger((T\setminus\sigma_D)\cup(T'\cap\sigma_D))}R^{((T'\setminus\sigma_D)\cup(T\cap\sigma_D))}\\
        &= \mathbf{R}^{((T\setminus\sigma_D)\cup(T'\cap\sigma_D),(T'\setminus\sigma_D)\cup(T\cap\sigma_D))}\\
        &= \mathbf{R}^{\sigma_D(T,T')}
    \end{align}
    from which Eq. (\ref{eq:pauli-action}) follows, as desired. It remains to show that $\sigma_D(T,T')$ is a valid group action of $S=\sigma(\mathcal{S})$ on $[n]^2$. The identity group action axiom requires $e(T,T')=(T,T')$, where $e$ is the identity element of $S$. It is obvious that this axiom holds with $e=\{\}$. The compatibility axiom then requires that $\varsigma_1(\varsigma_2(T,T')) = (\varsigma_1\symdif\varsigma_2)(T,T')$ for any $\varsigma_1,\varsigma_2\in S$. We now prove this as follows. It will be convenient to write the set difference $T\setminus \varsigma$ as $T\cap\varsigma^c$, where $\varsigma^c$ is the complement of $\varsigma$ in $[n]$, that is, $[n]\setminus \varsigma$. Then
    \begin{align}
        \varsigma_1(\varsigma_2(T,T')) &= \varsigma_1((T\cap\varsigma_2^c)\cup(T'\cap\varsigma_2),(T'\cap\varsigma_2^c)\cup(T\cap\varsigma_2)) \\
        &= (\tilde{T},\tilde{T'}),
    \end{align}
    where 
    \begin{align}
        \tilde{T} &= [((T\cap\varsigma_2^c)\cup(T'\cap\varsigma_2))\cap\varsigma_1^c]\cup [((T'\cap\varsigma_2^c)\cup(T\cap\varsigma_2))\cap\varsigma_1]\ \mathrm{and}\nonumber\\
        \tilde{T'} &= [((T'\cap\varsigma_2^c)\cup(T\cap\varsigma_2))\cap\varsigma_1^c]\cup [((T\cap\varsigma_2^c)\cup(T'\cap\varsigma_2))\cap\varsigma_1].
    \end{align}
    We start by expanding the above expression for $\tilde{T}$. The first step is to distribute the $\varsigma_1^c$ and $\varsigma_1$ over the set unions with which they intersect:
    \begin{align}
        \tilde{T} &= ((T\cap\varsigma_2^c)\cap\varsigma_1^c)\cup((T'\cap\varsigma_2)\cap\varsigma_1^c)\cup ((T'\cap\varsigma_2^c)\cap\varsigma_1)\cup((T\cap\varsigma_2)\cap\varsigma_1)
    \end{align}
    We now apply associativity of the set intersection and De Morgan's laws:
    \begin{align}
        \tilde{T} &= (T\cap(\varsigma_2\cup\varsigma_1)^c)\cup(T'\cap(\varsigma_2\cap\varsigma_1^c))\cup (T'\cap(\varsigma_2^c\cap\varsigma_1))\cup(T\cap(\varsigma_2\cap\varsigma_1)).
    \end{align}
    Next, we use the commutativity of set unions and the distributive property to write
    \begin{align}
        \tilde{T} &= (T\cap(\varsigma_2\cup\varsigma_1)^c)\cup(T\cap(\varsigma_2\cap\varsigma_1))\cup(T'\cap(\varsigma_2\cap\varsigma_1^c))\cup (T'\cap(\varsigma_2^c\cap\varsigma_1))\\
        &=(T\cap[(\varsigma_2\cup\varsigma_1)^c\cup(\varsigma_2\cap\varsigma_1)])\cup(T'\cap[(\varsigma_2\cap\varsigma_1^c)\cup(\varsigma_2^c\cap\varsigma_1)]).
    \end{align}
    Letting $\symdif$ denote the set symmetric difference, note that $(\varsigma_2\cup\varsigma_1)^c\cup(\varsigma_2\cap\varsigma_1) = [(\varsigma_2\cup\varsigma_1)\cap(\varsigma_2\cap\varsigma_1)^c]^c = [\varsigma_1\symdif\varsigma_2]^c$. Furthermore, $(\varsigma_2\cap\varsigma_1^c)\cup(\varsigma_2^c\cap\varsigma_1) = (\varsigma_2\setminus\varsigma_1)\cup(\varsigma_1\setminus\varsigma_2)=\varsigma_1\symdif\varsigma_2$. Then,
    \begin{align}
        \tilde{T} &=(T\cap(\varsigma_1\symdif\varsigma_2)^c)\cup(T'\cap(\varsigma_1\symdif\varsigma_2))\\
        &=(T\setminus(\varsigma_1\symdif\varsigma_2))\cup(T'\cap(\varsigma_1\symdif\varsigma_2)).
    \end{align}
    By an identical argument, we also have 
    \begin{align}
        \tilde{T}' &=(T'\setminus(\varsigma_1\symdif\varsigma_2))\cup(T\cap(\varsigma_1\symdif\varsigma_2)).
    \end{align}
    Lastly, since evidently $(\varsigma_1\symdif\varsigma_2)(T,T') = (\tilde{T},\tilde{T'})$ as well, we have 
    \begin{align}
        \varsigma_1(\varsigma_2(T,T'))=(\varsigma_1\symdif\varsigma_2)(T,T'),
    \end{align}
    as desired.

    Part (b). Pick any $D\in\mathcal{S}$ and $j_1\ldots j_n\in\mathbb{Z}_2^n$. We will again write $D$ as $D = e^{i\phi}\bigotimes_{l=1}^nD_l$ for some $\phi\in[0,2\pi)$, where each $D_l\in\{X,Y,Z,I\}$. Then
    \begin{align}
        D\ket{j_1\ldots j_n} &= e^{i\phi}\bigotimes_{l=1}^n D_l\ket{j_l}.
    \end{align}
    Note that if $l\in\sigma_D$, then $D_l=X$ or $Y$, so $D_l\ket{j_l}\propto \ket{j_l\oplus 1}$ up to a phase. On the other hand, if $l\notin\sigma_D$, then $D_l=Z$ or $I$, so $D_l\ket{j_l}\propto \ket{j_l}$ up to a phase. It follows that
    \begin{align}
        D\ket{j_1\ldots j_n} &= e^{i\phi'}\bigotimes_{l=1}^n \ket{j'_l},
    \end{align}
    for some phase $\phi'\in [0,2\pi)$, where
    \begin{align}
        j'_l = \begin{cases}
            j_l\oplus 1 &l\in\sigma_D\\
            j_l &l\notin\sigma_D
        \end{cases}
    \end{align}
    for $l=1,\ldots,n$. Therefore,
    \begin{align}
        D\ket{j_1\ldots j_n} = e^{i\phi'}\ket{\sigma_D(j_1\ldots j_n)}
    \end{align}
    for some phase $\phi'\in [0,2\pi)$, as desired. It remains to show that $\sigma_D(j_1\ldots j_n)$ is a valid group action of $S$ on $\mathbb{Z}_2^n$. The identity axiom is clearly satisfied for the identity element $\{\}$ of $S$. We now prove that the compatibility axiom holds. We need $\varsigma_1(\varsigma_2(j_1\ldots j_n)) = (\varsigma_1\symdif\varsigma_2)(j_1\ldots j_n)$ for any $\varsigma_1,\varsigma_2\in S$. We have 
    \begin{align}
        \varsigma_1(\varsigma_2(j_1\ldots j_n)) = \varsigma_1(j'_1\ldots j'_n) = j''_1\ldots j''_n
    \end{align}
    where for $l=1,\ldots,n$,
    \begin{align}
        j'_l = \begin{cases}
            j_l\oplus 1 &l\in\varsigma_2\\
            j_l &l\notin\varsigma2
        \end{cases}\quad \mathrm{and}\quad 
        j''_l=\begin{cases}
            j'_l\oplus 1 &l\in\varsigma_1\\
            j'_l &l\notin\varsigma_1
        \end{cases}.
    \end{align}
    We would like to find expressions for $j''_l$ in terms of $j_l$, and we proceed by considering all possible cases:
    \begin{enumerate}
        \item Suppose $l\in\varsigma_1$ and $l\in\varsigma_2$. Then $j_l''=j'_l\oplus 1$ and $j'_l = j_l\oplus 1$, so $j_l'' = j_l$.
        \item Suppose $l\in\varsigma_1$ and $l\notin\varsigma_2$. Then $j_l'' = j'_l\oplus 1$ and  $j'_l = j_l$, so $j_l'' = j_l\oplus 1$.
        \item Suppose $l\notin\varsigma_1$ and $l\in\varsigma_2$. Then $j_l''=j_l'$ and $j'_l=j_l\oplus 1$, so $j_l''=j_l\oplus 1$.
        \item Suppose $l \notin \varsigma_1$ and $l\notin \varsigma_2$. Then $j_l''=j_l'$ and $j_l'=j_l$, so $j_l''=j_l$.
    \end{enumerate}
    In summary,
    \begin{align}
        j_l'' = \begin{cases}
            j_l\oplus 1 &l\in\varsigma_1\symdif\varsigma_2\\
            j_l &l\notin\varsigma_1\symdif\varsigma_2
        \end{cases}.
    \end{align}
    Evidently, $(\varsigma_1\symdif\varsigma_2)(j_1\ldots j_n) = (j''_1\ldots j''_n)$ as well, from which it follows that 
    \begin{align}
        \varsigma_1(\varsigma_2(j_1\ldots j_n)) = (\varsigma_1\symdif\varsigma_2)(j_1\ldots j_n),
    \end{align}
    as desired.
\end{proof}

\subsection{Proof of Proposition \ref{prop:PS-conj}}\label{appendix:semidirect}
We now establish some preliminary results which will later be used to prove statements about the semidirect product $S_G\rtimes G$. For Propositions \ref{prop:sigma-perm}-\ref{prop:phi-homo}, assume we are given an arbitrary permutation group $G$ on $[n]$ and Pauli subgroup $\mathcal{S}\leq\mathcal{P}_n$; let $\mathcal{G}=P(G)$ and $S=\sigma(\mathcal{S})$. We first describe how Pauli operators transform when conjugated by permutation matrices.
\begin{proposition}\label{prop:sigma-perm}
    For any $\pi\in G$ and $D\in\mathcal{S}$, define $D' = P_\pi D P_\pi^\dagger$. Then $\sigma_{D'} = \pi(\sigma_D)$.
\end{proposition}
\begin{proof}
    Write $D$ as $D = e^{i\phi}\bigotimes_{j=1}^nD_j$ for some $\phi\in[0,2\pi)$, where each $D_j\in\{X,Y,Z,I\}$. Then Proposition \ref{prop:permute-factors} implies that
    \begin{align}
        D' &= P_\pi D P_\pi'\\
        &= e^{i\phi} \bigotimes_{j=1}^n D_{\pi^{-1}(j)}.
    \end{align}
    Thus, $D' = e^{i\phi}\bigotimes_{j=1}^nD'_j$, where $D'_j = D_{\pi^{-1}(j)}$. Note  that $j\in\sigma_{D'}$ is equivalent to $D'_j= X$ or $Y$ by the definition of $\sigma_{D'}$. Since $D'_j=D_{\pi^{-1}(j)}$, $D'_j= X$ or $Y$ if and only if $\pi^{-1}(j)\in\sigma_D$. Since $\pi^{-1}(j)\in\sigma_D$ if and only if $j\in \pi(\sigma_D)$, we conclude that $\sigma_{D'}=\pi(\sigma_D)$. 
\end{proof}
We now show that the two definitions for $S_G$ given in the main text (i.e., $S_G = \sigma(\mathcal{S}_\mathcal{G})$ and $S_G = \langle \pi(\varsigma):\varsigma\in S, \pi\in G\rangle$) are equivalent. 
\begin{proposition}\label{prop:equiv-def}
    The two definitions for $S_G$ are equivalent: 
    \begin{align}
        \sigma(\mathcal{S}_\mathcal{G}) = \langle \pi(\varsigma)\ :\ \varsigma\in S, \pi\in G\rangle.
    \end{align}
\end{proposition}
\begin{proof}    
    Suppose $\varsigma\in\sigma(\mathcal{S}_\mathcal{G})$. Then $\varsigma = \sigma_D$ for some $D\in\mathcal{S}_\mathcal{G}$. According to Eq. (\ref{eq:SsubG-def}), $D$ can be written as $D=\prod_{j=1}^kP_{\pi_j} D'_j P_{\pi_j}^\dagger$ for some $P_{\pi_j}\in \mathcal{G}$, $D'_j\in \mathcal{S}$, and $k\in\mathbb{N}$. Since $\sigma$ is a homomorphism, we have $\sigma_D = \Symdif_{j=1}^k\sigma(P_{\pi_j} D'_j P_{\pi_j}^\dagger)$. By Proposition \ref{prop:sigma-perm}, $\sigma(P_{\pi_j} D'_j P_{\pi_j}^\dagger)=\pi_j(\sigma_{D'_j})$, which implies that $\sigma_D = \Symdif_{j=1}^k\pi_j(\sigma_{D'_j})$. Noting that each $\sigma_{D'_j}\in S$, it follows that $\sigma(\mathcal{S}_\mathcal{G})\subseteq \langle \pi(\varsigma'):\varsigma'\in S, \pi\in G\rangle$.

    Conversely, suppose $\varsigma\in\langle \pi(\varsigma'):\varsigma'\in S, \pi\in G\rangle$ such that $\varsigma = \Symdif_{j=1}^k\pi_j(\varsigma'_j)$ for some $\pi_j\in G$, $\varsigma'_j\in S$, and $k\in\mathbb{N}$. Since $S=\sigma(\mathcal{S})$, there exist $D'_j\in\mathcal{S}$ such that $\varsigma'_j=\sigma_{D'_j}$. Hence, $\varsigma = \Symdif_{j=1}^k\pi_j(\sigma_{D'_j})$, and it follows from Proposition \ref{prop:sigma-perm} that $\varsigma = \Symdif_{j=1}^k\sigma(P_{\pi_j}D'_j P^\dagger_{\pi_j})$. Because $\sigma$ is a homomorphism, we have $\varsigma = \sigma(\Symdif_{j=1}^kP_{\pi_j}D'_j P^\dagger_{\pi_j})$, which implies that $\varsigma\in\sigma(\mathcal{S}_\mathcal{G})$. It follows that $\sigma(\mathcal{S}_\mathcal{G})\supseteq \langle \pi(\varsigma'):\varsigma'\in S, \pi\in G\rangle$, thereby completing the proof.
\end{proof}
Next, we show that $S_G$ is invariant under permutations from $G$:
\begin{proposition}\label{prop:SG-perms}
    For any $\varsigma \in S_G$ and $\pi\in G$, $\pi(\varsigma)\in S_G$.
\end{proposition}
\begin{proof}
    Consider any $\varsigma\in S_G$ and $\pi\in G$. Because $\sigma:\mathcal{S_G}\rightarrow S_G$ is surjective, there exists $D\in\mathcal{S}_\mathcal{G}$ such that $\varsigma = \sigma_D$. Now define $D' = P_\pi D P_\pi^\dagger$, and note that $D'\in\mathcal{S}_\mathcal{G}$ since $\mathcal{G}\subseteq \mathcal{N}(\mathcal{S}_\mathcal{G})$ by Proposition \ref{prop:SG-decomp1}. It follows that $\sigma_{D'}\in S_G$, since $S_G = \sigma(\mathcal{S}_\mathcal{G})$. Furthermore, because $\sigma_{D'} = \pi(\sigma_D)=\pi(\varsigma)$ by Proposition \ref{prop:sigma-perm}, we conclude that $\pi(\varsigma)\in S_G$, as desired.
\end{proof}
We can use this result to show that permutations in $G$ act as automorphisms of $S_G$.
\begin{proposition}\label{prop:sigma-norm}
    For any $\pi\in G$, the map $\pi:S_G\rightarrow S_G$ is an automorphism of $S_G$.
\end{proposition}
\begin{proof}
    The result of Proposition \ref{prop:SG-perms} implies that $\pi(\cdot)$ is indeed a mapping from $S_G$ to itself. We must now show that $\pi(\cdot)$ is a homomorphism; specifically, we must demonstrate that $\pi(\varsigma_1\symdif\varsigma_2) = \pi(\varsigma_1)\symdif\pi(\varsigma_2)$ for all $\varsigma_1,\varsigma_2\in S_G$. Letting $\varsigma^c$ indicate the complement of $\varsigma$ in $[n]$ for any $\varsigma\subseteq[n]$, we have
    \begin{align}
        \pi(\varsigma_1\symdif\varsigma_2) &= \pi((\varsigma_1\cup\varsigma_2)\setminus (\varsigma_1\cap\varsigma_2)).
    \end{align}
    Note that $\pi(\varsigma\cup\varsigma') = \pi(\varsigma)\cup\pi(\varsigma')$ for any $\varsigma,\varsigma'\subseteq[n]$. Additionally, since $\pi$ is bijective, we have $\pi(\varsigma\cap\varsigma') = \pi(\varsigma)\cap\pi(\varsigma')$ and $\pi(\varsigma\setminus \varsigma') = \pi(\varsigma)\setminus \pi(\varsigma')$ as well. Thus,
    \begin{align}
        \pi(\varsigma_1\symdif\varsigma_2) &= \pi(\varsigma_1\cup\varsigma_2)\setminus\pi(\varsigma_1\cap\varsigma_2)\\
        &= [\pi(\varsigma_1)\cup\pi(\varsigma_2)]\setminus[\pi(\varsigma_1)\cap\pi(\varsigma_2)]\\
        &= \pi(\varsigma_1)\symdif\pi(\varsigma_2),
    \end{align}
    as desired. Because $\pi$ is bijective, we conclude that the homomorphism $\pi:S_G\rightarrow S_G$ is an automorphism.
\end{proof}
We are now equipped to show that $\Phi$ is a surjective homomorphism.
\begin{proposition}\label{prop:phi-homo}
    The map $\Phi$ is a surjective homomorphism from $\langle\mathcal{S},\mathcal{G}\rangle$ to $S_G\rtimes G$. 
\end{proposition}
\begin{proof}
    By Proposition \ref{prop:SG-decomp2}, any element of $\langle\mathcal{S},\mathcal{G}\rangle$ can be written uniquely as $DP_\pi$ for some $D\in\mathcal{S}_\mathcal{G}$ and $P_\pi\in\mathcal{G}$. To demonstrate that $\Phi$ is a homomorphism, we must show that 
    \begin{align}
        \Phi[D_1P_{\pi_1}D_2P_{\pi_2}] = \Phi[D_1P_{\pi_1}]\Phi[D_2P_{\pi_2}]
    \end{align}
    for any $D_1P_{\pi_1},D_2P_{\pi_2}\in\langle\mathcal{S},\mathcal{G}\rangle$. Hence, pick any arbitrary $D_1P_{\pi_1},D_2P_{\pi_2}\in\langle\mathcal{S},\mathcal{G}\rangle$. Because $\mathcal{G}\in\mathcal{N}(\mathcal{S}_\mathcal{G})$ by Proposition \ref{prop:SG-decomp1}, we have $P_{\pi_1}D_2 = D'_2 P_{\pi_1}$ for some $D'_2\in\mathcal{S}_\mathcal{G}$. Therefore,
    \begin{align}
        \Phi[D_1P_{\pi_1}D_2P_{\pi_2}] &= \Phi[(D_1D'_2)(P_{\pi_1}P_{\pi_2})]\\
        &=(\sigma_{D_1D'_2},P^{-1}(P_{\pi_1}P_{\pi_2}))\\
        &= (\sigma_{D_1}\symdif\sigma_{D'_2},\pi_1\pi_2),
    \end{align}
    where the last equality follows from the fact that $\sigma$ and $P^{-1}$ are homomorphisms. Because $D'_2 = P_{\pi_1}D_2P_{\pi_1}^\dagger$, Proposition \ref{prop:sigma-perm} implies that $\sigma_{D'_2} = \pi_1(\sigma_{D_2})$. Hence,
    \begin{align}
        \Phi[D_1P_{\pi_1}D_2P_{\pi_2}] &= (\sigma_{D_1}\symdif\pi_1(\sigma_{D_2}),\pi_1\pi_2)\\
        &= (\sigma_{D_1},\pi_1)\cdot(\sigma_{D_2},\pi_2)\\
        &= \Phi[D_1P_{\pi_1}]\Phi[D_2P_{\pi_2}],
    \end{align}
    as desired. The surjectivity of $\Phi$ follows from the surjectivity of $\sigma:\mathcal{S}_\mathcal{G}\rightarrow S_G$ and $P^{-1}:\mathcal{G}\rightarrow G$.
\end{proof}

Lastly, we prove Proposition \ref{prop:PS-conj}.
\begin{proof}[Proof of Proposition \ref{prop:PS-conj}]
    Part (a). Any $U\in\langle\mathcal{S},\mathcal{G}\rangle$ can be written uniquely as $DP_\pi$ for some $D\in\mathcal{S}_\mathcal{G}$ and $\pi\in G$ by Proposition \ref{prop:SG-decomp2}. Thus, by Propositions \ref{prop:P-conj} and \ref{prop:S-conj},
    \begin{align}
        U\rbf U^\dagger &= (DP_\pi)\rbf (DP_\pi)^\dagger\\
        &= D\left(P_\pi\rbf P_\pi^\dagger\right)D^\dagger\\
        &= D\mathbf{R}^{\pi(T,T')}D^\dagger\\
        &= \mathbf{R}^{\sigma_D[\pi(T,T')]}\\
        &= \mathbf{R}^{(\sigma_D,\pi)[(T,T')]},
    \end{align}
    where we note that $(\sigma_D,\pi) = \Phi(U)$.
    It remains to show that Eq. (\ref{eq:sympauli-action}) is a valid group action of $S_G\rtimes G$ on $[n]^2$. The identity axiom clearly holds. To verify the compatibility action, we will make use of the following result: for any $\pi\in G$, $\varsigma\in S_G$, and $T,T'\subseteq[n]$,
    \begin{align}\label{eq:semidirect-claim}
        \pi[\varsigma(T,T')] = \pi(\varsigma)[\pi(T,T')].
    \end{align}
    We now prove this claim. Note that
    \begin{align}
        \pi[\varsigma(T,T')] &= \pi((T\setminus\varsigma)\cup(T'\cap\varsigma),(T'\setminus\varsigma)\cup(T\cap\varsigma))\\
        &= (\pi[(T\setminus\varsigma)\cup(T'\cap\varsigma)],\pi[(T'\setminus\varsigma)\cup(T\cap\varsigma)])\\
        &= (\pi(T\setminus\varsigma)\cup\pi(T'\cap\varsigma),\pi(T'\setminus\varsigma)\cup\pi(T\cap\varsigma))\\
        &= ([\pi(T)\setminus\pi(\varsigma)]\cup[\pi(T')\cap\pi(\varsigma)],[\pi(T')\setminus\pi(\varsigma)]\cup[\pi(T)\cap\pi(\varsigma)]),\label{eq:sympauli-proof01}
    \end{align}
    where the last equality follows from the bijectivity of $\pi$. Now observe that
    \begin{align}
        \pi(\varsigma)[\pi(T,T')] &=  \pi(\varsigma)(\pi(T),\pi(T'))\\
        &= ([\pi(T)\setminus\pi(\varsigma)]\cup[\pi(T')\cap\pi(\varsigma)],[\pi(T')\setminus\pi(\varsigma)]\cup[\pi(T)\cap\pi(\varsigma)]).\label{eq:sympauli-proof02}
    \end{align}
    Comparing Eqs. (\ref{eq:sympauli-proof01}) and (\ref{eq:sympauli-proof02}), the claim follows. We now verify the compatibility axiom by showing that
    \begin{align}
         (\varsigma_1,\pi_1)[(\varsigma_2,\pi_2)[(T,T')]]=((\varsigma_1,\pi_1)\cdot(\varsigma_2,\pi_2))[(T,T')]
    \end{align}
    for any $(\varsigma_1,\pi_1),(\varsigma_2,\pi_2)\in S_G\rtimes G$. We expand the left side as
    \begin{align}
        (\varsigma_1,\pi_1)[(\varsigma_2,\pi_2)[(T,T')]] &= (\varsigma_1,\pi_1)[\varsigma_2[\pi_2(T,T')]]\\
        &=\varsigma_1[\pi_1[\varsigma_2[\pi_2(T,T')]]].
    \end{align}
    Using our earlier result of Eq. (\ref{eq:semidirect-claim}), we can rewrite this expression as
    \begin{align}
        (\varsigma_1,\pi_1)[(\varsigma_2,\pi_2)[(T,T')]]&=\varsigma_1[\pi_1(\varsigma_2)[\pi_1[\pi_2(T,T')]]]\\
        &= (\varsigma_1\symdif\pi_1(\varsigma_2))[(\pi_1\pi_2)(T,T')]
    \end{align}
    using the fact that $\varsigma(\cdot)$ and $\pi(\cdot)$ are individually valid group actions which satisfy the compatibility axiom. It then follows that
    \begin{align}
        (\varsigma_1,\pi_1)[(\varsigma_2,\pi_2)[(T,T')]]&=(\varsigma_1\symdif\pi_1(\varsigma_2),\pi_1\pi_2)[(T,T')]\\
        &=((\varsigma_1,\pi_1)\cdot(\varsigma_2,\pi_2))[(T,T')],
    \end{align}
    as desired. 

    Part (b). Again, we may write any $U\in\langle\mathcal{S},\mathcal{G}\rangle$ uniquely as $DP_\pi$ for some $D\in\mathcal{S}_\mathcal{G}$ and $\pi\in G$. Then, by Propositions \ref{prop:P-conj} and \ref{prop:S-conj},
    \begin{align}
        U\ket{j_1\ldots j_n}&=DP_\pi\ket{j_1\ldots j_n}\\
        &=D\ket{\pi(j_1\ldots j_n)}\\
        &=e^{i\phi}\ket{\sigma_D[\pi(j_1\ldots j_n)]}\\
        &=e^{i\phi}\ket{(\sigma_D,\pi)[(j_1\ldots j_n)]}
    \end{align}
    for some $\phi\in[0,2\pi)$, where we note that $(\sigma_D,\pi) = \Phi(U)$.
    It remains to show that Eq. (\ref{eq:sympauli-action-state}) is a valid group action of $S_G\rtimes G$ on $\mathbb{Z}_2^n$. Again, the identity axiom clearly holds. To verify the compatibility action, we will make use of an analogous result to Eq. (\ref{eq:semidirect-claim}): for any $\pi\in G$, $\varsigma\in S_G$, and $j_1\ldots j_n\in\mathbb{Z}_2^n$:
    \begin{align}\label{eq:semidirect-claim2}
        \pi[\varsigma(j_1\ldots j_n)] = \pi(\varsigma)[\pi(j_1\ldots j_n)].
    \end{align}
    We now prove this claim. Expand the left side as
    \begin{align}
        \pi[\varsigma(j_1\ldots j_n)] &= \pi(j'_1\ldots j'_n) \\
        &= j'_{\pi^{-1}(1)}\ldots j'_{\pi^{-1}(n)},
    \end{align}
    where $j'_l = j_l\oplus 1$ if $l\in\varsigma$ and $j'_l=j_l$ otherwise. It follows that 
    \begin{align}
        j'_{\pi^{-1}(l)}&=\begin{cases}
            j_{\pi^{-1}(l)}\oplus 1 & \pi^{-1}(l)\in\varsigma\\
            j_{\pi^{-1}(l)}& \mathrm{otherwise}
        \end{cases} 
    \end{align}
    or equivalently
    \begin{align}\label{eq:semidirect01}
        j'_{\pi^{-1}(l)}&=\begin{cases}
            j_{\pi^{-1}(l)}\oplus 1 & l\in\pi(\varsigma)\\
            j_{\pi^{-1}(l)}& \mathrm{otherwise}
        \end{cases}.
    \end{align}
    Eq. (\ref{eq:semidirect01}) then implies that $j'_{\pi^{-1}(1)}\ldots j'_{\pi^{-1}(n)} = \pi(\varsigma)(j_{\pi^{-1}(1)}\ldots j_{\pi^{-1}(n)})$, from which we deduce that
    \begin{align}
        \pi[\varsigma(j_1\ldots j_n)] &=\pi(\varsigma)(j_{\pi^{-1}(1)}\ldots j_{\pi^{-1}(n)})\\
        &=\pi(\varsigma)[\pi(j_1\ldots j_n)],
    \end{align}
    as desired. We are now equipped to prove the compatibility axiom, namely
    \begin{align}
        (\varsigma_1,\pi_1)[(\varsigma_2,\pi_2)(j_1\ldots j_n)]
        &=((\varsigma_1,\pi_1)\cdot(\varsigma_2,\pi_2))(j_1\ldots j_n),
    \end{align}
    for any $(\varsigma_1,\pi_1),(\varsigma_2,\pi_2)\in S_G\rtimes G$. The left side can be written as
    \begin{align}
        (\varsigma_1,\pi_1)[(\varsigma_2,\pi_2)(j_1\ldots j_n)]
        &=\varsigma_1[\pi_1[\varsigma_2[\pi_2(j_1\ldots j_n)]]],
    \end{align}
    and applying our result from Eq. (\ref{eq:semidirect-claim2}), we obtain
    \begin{align}
        (\varsigma_1,\pi_1)[(\varsigma_2,\pi_2)(j_1\ldots j_n)]
        &=\varsigma_1[\pi_1(\varsigma_2)[\pi_1[\pi_2(j_1\ldots j_n)]]]\\
        &=(\varsigma_1\symdif\pi_1(\varsigma_2))[(\pi_1\pi_2)(j_1\ldots j_n)]\\
        &=(\varsigma_1\symdif\pi_1(\varsigma_2),\pi_1\pi_2)(j_1\ldots j_n)\\
        &=((\varsigma_1,\pi_1)\cdot(\varsigma_2,\pi_2))(j_1\ldots j_n),
    \end{align}
    as desired.
\end{proof}

\subsection{Symmetries of $\mathcal{T}^2$}\label{appendix:pairs-sym}
In this section, we prove a supplemental result which is useful for Theorem \ref{lemma:SG-simp}. Namely, we show that if some qubit swap group $S\leq\mathbb{P}([n])$ and permutation group $G$ on $[n]$ are symmetries of $\mathcal{T}^2$, then $S_G\rtimes G$ is also a symmetry of $\mathcal{T}^2$:
\begin{proposition}\label{prop:pairs-sym}
    Given a permutation group $G$ on $[n]$ and a Pauli subgroup $\mathcal{S}\leq \mathcal{P}_n$, let $S=\sigma(\mathcal{S})$. Then $\mathcal{T}^2$ is invariant under $S$ and $G$ if and only if $\mathcal{T}^2$ is invariant under $S_G\rtimes G$.
\end{proposition}
\begin{proof}
    ``$\implies$" direction: Suppose $\mathcal{T}^2$ is invariant under $S$ and $G$. Consider any $(\varsigma,\pi)\in S_G\rtimes G$. By Eq. (\ref{eq:SG-equiv-def}), $\varsigma=\Symdif_{j=1}^k\pi_j(\varsigma_j)$ for some $\pi_j\in G$, $\varsigma_j\in S$, and $k\in\mathbb{N}$. Note for every $j=1,\ldots,k$ that
    \begin{align}
    (\pi_j(\varsigma_j),e)=(\varnothing,\pi_j)\cdot(\varsigma_j,\pi_j^{-1}),
    \end{align}
    where $e$ is the identity permutation. It follows that 
    \begin{align}
        (\varsigma,\pi) &= (\varsigma,e)\cdot(\varnothing,\pi)\\
        &=\left[\prod_{j=1}^k(\pi_j(\varsigma_j),e)\right]\cdot(\varnothing,\pi)\\
        &=\left[\prod_{j=1}^k(\varnothing,\pi_j)\cdot(\varsigma_j,\pi_j^{-1})\right]\cdot(\varnothing,\pi).
    \end{align}
    Now pick any $(T,T')\in\mathcal{T}^2$. Using the action of $S_G\times G$ on $[n]^2$ defined in Proposition \ref{prop:PS-conj}, observe that
    \begin{align}
        (\varsigma,\pi)[(T,T')] &= \left((\varnothing,\pi_1)\circ(\varsigma_1,\pi_1^{-1})\circ\cdots\circ(\varnothing,\pi_k)\circ(\varsigma_k,\pi_k^{-1})\circ(\varnothing,\pi)\right)[(T,T')]\\
        &= \left(\pi_1\circ\varsigma_1\circ\pi_1^{-1}\circ\cdots\circ\pi_k\circ\varsigma_k\circ\pi_k^{-1}\circ\pi\right)(T,T'),
    \end{align}
    where the ``$\circ$" symbol indicates function composition. Since $\mathcal{T}^2$ is invariant under the actions of $S$ and $G$ by assumption, it follows that $(\varsigma,\pi)[(T,T')]\in\mathcal{T}^2$. We thus conclude that $\mathcal{T}^2$ is invariant under $S_G\rtimes G$.

    ``$\impliedby$" direction: Suppose $\mathcal{T}^2$ is invariant under $S_G\rtimes G$. Pick any $(T,T')\in\mathcal{T}^2$. For any $\varsigma\in S$, we have $\varsigma(T,T')=(\varsigma,e)[(T,T')]\in\mathcal{T}^2$. Likewise, for any $\pi\in G$, we have $\pi(T,T') = (\varnothing,\pi)[(T,T')]\in\mathcal{T}^2$. It follows that $\mathcal{T}^2$ is invariant under $S$ and $G$.
\end{proof}

\subsection{Proof of Theorem \ref{thm:exist}}\label{appendix:exist}
In this section, we prove Theorem \ref{thm:exist}. However, it is first useful to compute the eigenvalues of the $R^{(T)}(\theta)$ operators. Note that since $R^{(T)}$ is a tensor product of single-qubit $R_Z$ operators, its eigenvectors will be the $Z$-eigenbasis vectors $\ket{j_1\ldots j_n}$ for $j_1,\ldots,j_n\in\{0,1\}$. From
\begin{align}
    R^{(T)}\ket{j_1\ldots j_n} &=  \left( \prod_{k\in T}R^{(\{k\})}_Z\right)\ket{j_1\ldots j_n}\label{eq:eigval01}\\
    &= \left(\prod_{k\in T}\exp\left[\frac{i\theta}{2}(-1)^{j_k+1}\right]\right)\ket{j_1\ldots j_n}\\
    &= \exp\left[\frac{i\theta}{2}\sum_{k\in T} (-1)^{j_k+1}\right]\ket{j_1\ldots j_n}.\label{eq:eigphase}
\end{align}
it follows that the eigenvalue of $R^{(T)}$ associated with eigenvector $\ket{j_1\ldots j_n}$ is $\exp(i\varphi_{j_1\ldots j_n}^{(T)}(\theta))$ where
\begin{align}\label{eq:R-eigphase-def}
    \varphi_{j_1\ldots j_n}^{(T)}(\theta) = \frac{\theta}{2}\sum_{k\in T} (-1)^{j_k+1}.
\end{align}

We are now prepared to prove Theorem \ref{thm:exist}.

\begin{proof}[Proof of Theorem \ref{thm:exist}]
The reverse direction is trivial, so we now prove the forward direction.

The first step is to prove that the existence of a TS state $\ket{\psi}$ implies the existence of another TS state $\ket{\psi_\mathcal{G}}$ such that $\ket{\psi_\mathcal{G}} = P_\pi\ket{\psi_\mathcal{G}}$ for all $\pi\in G$. $\ket{\psi}$ can be written in the $Z$-eigenbasis as 
\begin{align} \label{eq:psi-decomp}
    \ket{\psi} = \sum_{j_1\ldots j_n\in\mathbb{Z}_2^n} a_{j_1\ldots j_n}\ket{{j_1\ldots j_n}}
\end{align}
for some $a_{j_1\ldots j_n}\in\mathbb{C}$. Then we claim the desired state $\ket{\psi_\mathcal{G}}$ can be written as 
\begin{align} \label{eq:sym-phi}
    \ket{\psi_\mathcal{G}} = \sum_{j_1\ldots j_n\in\mathbb{Z}_2^n} \left(\sqrt{\frac{1}{\abs{ G}}\sum_{\pi\in G}\abs{a_{j_{\pi(1)}\ldots j_{\pi(n)}}}^2}\right)\ket{{j_1\ldots j_n}}.
\end{align}
To show that $\ket{\psi_\mathcal{G}} = P_\pi\ket{\psi_\mathcal{G}}$, note for any $\pi \in G$ that
\begin{align}
    P_\pi\ket{\psi_\mathcal{G}} = \sum_{j_1\ldots j_n\in\mathbb{Z}_2^n} \left(\sqrt{\frac{1}{\abs{G}}\sum_{\pi'\in G}\abs{a_{j_{\pi'(1)}\ldots j_{\pi'(n)}}}^2}\right)\ket{{j_{\pi^{-1}(1)}\ldots j_{\pi^{-1}(n)}}},
\end{align}
where we have invoked Eq. (\ref{eq:P-pi}). Since $\pi$ is a bijection, we can relabel the indices in the above with $j_k \rightarrow j_{\pi(k)}$ as follows, noting that a sum over $j_{\pi(1)},\ldots,j_{\pi(n)}$ is equivalent to a sum over $j_1,\ldots,j_n$:
\begin{align}
    P_\pi\ket{\psi_\mathcal{G}} &= \sum_{j_{\pi(1)}\ldots j_{\pi(n)}\in\mathbb{Z}_2^n} \left(\sqrt{\frac{1}{\abs{ G}}\sum_{\pi'\in G}\abs{a_{j_{\pi(\pi'(1))}\ldots j_{\pi(\pi'(n))}}}^2}\right)\ket{{j_1\ldots j_n}}\\
    &= \sum_{j_1\ldots j_n\in\mathbb{Z}_2^n} \left(\sqrt{\frac{1}{\abs{ G}}\sum_{\pi'\in G}\abs{a_{j_{\pi(\pi'(1))}\ldots j_{\pi(\pi'(n))}}}^2}\right)\ket{{j_1\ldots j_n}}\\
    &= \sum_{j_1\ldots j_n\in\mathbb{Z}_2^n} \left(\sqrt{\frac{1}{\abs{ G}}\sum_{\pi''\in G}\abs{a_{j_{\pi''(1)}\ldots j_{\pi''(n)}}}^2}\right)\ket{{j_1\ldots j_n}}.\label{eq:sym-phi2}
\end{align}
Comparing Eqs. (\ref{eq:sym-phi}) and (\ref{eq:sym-phi2}), we deduce that $\ket{\psi_\mathcal{G}} = P_\pi\ket{\psi_\mathcal{G}}$ for any $\pi\in G$, as desired. 

Now, we prove that $\bra{\psi_\mathcal{G}}\rbf\ket{\psi_\mathcal{G}} = \delta_{T,T'}$ for any $T, T'\in \nmg$, hence proving that $\ket{\psi_\mathcal{G}}$ is a TS state. As a prerequisite, note that for any $\pi\in G$, we can write the state $P_\pi\ket{\psi}$ as  
\begin{align}
    P_\pi\ket{\psi} &= \sum_{j_1\ldots j_n\in\mathbb{Z}_2^n} a_{j_1\ldots j_n}\ket{{j_{\pi^{-1}(1)}\ldots j_{\pi^{-1}(n)}}}\\
    &= \sum_{j_1\ldots j_n\in\mathbb{Z}_2^n} a_{j_{\pi(1)}\ldots j_{\pi(n)}}\ket{{j_1\ldots j_n}}, \label{eq:P-pi-psi}
\end{align}
where the first equality follows from the application of Eq. (\ref{eq:P-pi}) to Eq. (\ref{eq:psi-decomp}) and the second equality follows from the change of indices $j_k \rightarrow j_{\pi(k)}$. Additionally, define $\varphi_{j_1\ldots j_n}^{(T)}$ to be the eigenphase of $R^{(T)}$ given by Eq. (\ref{eq:R-eigphase-def}). Then for any $T,T'\in\mathcal{T}$,
\begin{align}
    \bra{\psi_\mathcal{G}}\rbf\ket{\psi_\mathcal{G}} &= \frac{1}{\abs{ G}}\sum_{j_1\ldots j_n\in\mathbb{Z}_2^n}\exp[i\left(\varphi_{j_1\ldots j_n}^{(T')}-\varphi_{j_1\ldots j_n}^{(T)}\right)]\sum_{\pi\in G}\abs{a_{j_{\pi(1)}\ldots j_{\pi(n)}}}^2\\
    &= \frac{1}{\abs{ G}}\sum_{\pi\in G}\sum_{j_1\ldots j_n\in\mathbb{Z}_2^n}\exp[i\left(\varphi_{j_1\ldots j_n}^{(T')}-\varphi_{j_1\ldots j_n}^{(T)}\right)]\abs{a_{j_{\pi(1)}\ldots j_{\pi(n)}}}^2\\
    &= \frac{1}{\abs{ G}}\sum_{\pi\in G} \bra{\psi}P^\dagger_\pi \rbf P_\pi\ket{\psi}\label{eq:symmetrize01}\\
    &= \frac{1}{\abs{ G}}\sum_{\pi\in G} \bra{\psi} \mathbf{R}^{\pi(T,T')} \ket{\psi}\label{eq:symmetrize02}
\end{align}
where Eq. (\ref{eq:symmetrize01}) follows from the application of Eq. (\ref{eq:P-pi-psi}) and Eq. (\ref{eq:symmetrize02}) follows from Proposition \ref{prop:P-conj}. By Proposition \ref{prop:G-transitive}, $\pi(T,T')\in\mathcal{T}^2$ because $\mathcal{T}$ is $G$-transitive. Consequently, since $\ket{\psi}$ is a TS state, $\bra{\psi} \mathbf{R}^{\pi(T,T')} \ket{\psi} = \bra{\psi} \mathbf{R}^{(\pi(T),\pi(T'))} \ket{\psi}=\delta_{\pi(T),\pi(T')}$ for all $\pi\in G$. Furthermore, noting that $\pi(\cdot)$ acts bijectively on trajectories, $\delta_{\pi(T),\pi(T')}=\delta_{T,T'}$ for all $\pi\in G$. Thus, for any $T,T'\in\mathcal{T}$,
\begin{align}
    \bra{\psi_\mathcal{G}}\rbf\ket{\psi_\mathcal{G}} &= \frac{1}{\abs{ G}}\sum_{\pi\in G} \delta_{\pi(T),\pi(T')}\\
    &= \frac{1}{\abs{ G}}\sum_{\pi\in G} \delta_{T,T'}\\
    &= \delta_{T,T'},
\end{align}
which implies that $\ket{\psi_{\mathcal{G}}}$ is a TS state, as desired.

The next step is to prove that the existence of a TS state $\ket{\psi_\mathcal{G}}$ satisfying $\ket{\psi_\mathcal{G}} = P_\pi\ket{\psi_\mathcal{G}}$ for all $\pi\in G$ implies the existence of another TS state $\ket{\psi_{\tilde{\mathcal{G}}}}$ satisfying $\ket{\psi_{\tilde{\mathcal{G}}}} = U\ket{\psi_{\tilde{\mathcal{G}}}}$ for all $U\in \tilde{\mathcal{G}}$. $\ket{\psi_\mathcal{G}}$ can be written in the $Z$-eigenbasis as 
\begin{align} \label{eq:phi-decomp}
    \ket{\psi_\mathcal{G}} = \sum_{j_1\ldots j_n\in\mathbb{Z}_2^n} b_{j_1\ldots j_n}\ket{{j_1\ldots j_n}}
\end{align}
for some $b_{j_1\ldots j_n}\in\mathbb{C}$. We now claim that the desired state $\ket{\psi_{\tilde{\mathcal{G}}}}$ can be written as 
\begin{align}\label{eq:psi_sym}
    \ket{\psi_{\tilde{\mathcal{G}}}} = \frac{1}{\sqrt{2}}\sum_{j_1\ldots j_n\in\mathbb{Z}_2^n}\left(\sqrt{\abs{b_{j_1\ldots j_n}}^2 + \abs{b_{[n](j_1\ldots j_n)}}^2}\right)\ket{{j_1\ldots j_n}}
\end{align}
where the action of $[n]$ on $j_1\ldots j_n$ is given by Eq. (\ref{eq:bit-flip-action}), i.e., $[n](j_1\ldots j_n) = (j_1\oplus 1)\ldots (j_n\oplus 1)$.

Recall that $\tilde{\mathcal{G}} = \mathcal{S}\mathcal{G}$, where $\mathcal{S}=\{I^{\otimes n},X^{\otimes n}\}$. Hence, to show that $\ket{\psi_{\tilde{\mathcal{G}}}} = U\ket{\psi_{\tilde{\mathcal{G}}}}$ for all $U\in\tilde{\mathcal{G}}$, it suffices to show that $\ket{\psi_{\tilde{\mathcal{G}}}} =  X^{\otimes n}\ket{\psi_{\tilde{\mathcal{G}}}}= P_\pi\ket{\psi_{\tilde{\mathcal{G}}}}$ for all $\pi\in G$. We first prove that $\ket{\psi_{\tilde{\mathcal{G}}}}=X^{\otimes n}\ket{\psi_{\tilde{\mathcal{G}}}}$. We have
\begin{align}
    X^{\otimes n}\ket{\psi_{\tilde{\mathcal{G}}}} &= \frac{1}{\sqrt{2}}\sum_{j_1\ldots j_n\in\mathbb{Z}_2^n}\left(\sqrt{\abs{b_{j_1\ldots j_n}}^2 + \abs{b_{[n](j_1\ldots j_n)}}^2}\right)X^{\otimes n}\ket{{j_1\ldots j_n}}\\
    &= \frac{1}{\sqrt{2}}\sum_{j_1\ldots j_n\in\mathbb{Z}_2^n}\left(\sqrt{\abs{b_{j_1\ldots j_n}}^2 + \abs{b_{[n](j_1\ldots j_n)}}^2}\right)\ket{{[n](j_1\ldots j_n)}}
\end{align}
due to Eq. (\ref{eq:X-def}). Changing the indices of summation from $j_k\rightarrow j_k \oplus 1$, we then have
\begin{align}
    X^{\otimes n}\ket{\psi_{\tilde{\mathcal{G}}}}&= \frac{1}{\sqrt{2}}\sum_{(j_1\oplus 1)\ldots(j_n \oplus 1)\in\mathbb{Z}_2^n}\left(\sqrt{\abs{b_{[n](j_1\ldots j_n)}}^2+\abs{b_{j_1\ldots j_n}}^2}\right)\ket{{j_1\ldots j_n}}\\
    &= \frac{1}{\sqrt{2}}\sum_{j_1\ldots j_n\in\mathbb{Z}_2^n}\left(\sqrt{\abs{b_{[n](j_1\ldots j_n)}}^2+\abs{b_{j_1\ldots j_n}}^2}\right)\ket{{j_1\ldots j_n}}\\
    &=\ket{\psi_{\tilde{\mathcal{G}}}},
\end{align}
as desired. We next prove that $P_\pi\ket{\psi_{\tilde{\mathcal{G}}}}=\ket{\psi_{\tilde{\mathcal{G}}}}$ for all $\pi\in G$. We first need a preliminary fact: analogously to Eq. (\ref{eq:P-pi-psi}), we can write $P_\pi\ket{\psi_\mathcal{G}}$ as 
\begin{align} \label{eq:phi-P}
    P_\pi\ket{\psi_\mathcal{G}} = \sum_{j_1\ldots j_n\in\mathbb{Z}_2^n} b_{j_{\pi(1)}\ldots j_{\pi(n)}}\ket{{j_1\ldots j_n}}.    
\end{align}
Because $\ket{\psi_{\tilde{\mathcal{G}}}} = P_\pi\ket{\psi_{\tilde{\mathcal{G}}}}$ for all $\pi\in G$, we deduce by comparing Eqs. (\ref{eq:phi-decomp}) and (\ref{eq:phi-P}) that $b_{j_1\ldots j_n} = b_{j_{\pi(1)}\ldots j_{\pi(n)}}$ for all $j_1\ldots j_n\in\mathbb{Z}_2^n$ and $\pi\in G$. The fact that $b_{j_1\ldots j_n} = b_{j_{\pi(1)}\ldots j_{\pi(n)}}$ readily implies that $b_{[n](j_1\ldots j_n)} = b_{[n](j_{\pi(1)}\ldots j_{\pi(n)})}$ as well. Consequently, for all $\pi\in G$,
\begin{align}
    P_\pi\ket{\psi_{\tilde{\mathcal{G}}}} &= \frac{1}{\sqrt{2}}\sum_{j_1\ldots j_n\in\mathbb{Z}_2^n}\left(\sqrt{\abs{b_{j_1\ldots j_n}}^2 + \abs{b_{[n](j_1\ldots j_n)}}^2}\right)\ket{{j_{\pi^{-1}(1)}\ldots j_{\pi^{-1}(n)}}}\\
    &= \frac{1}{\sqrt{2}}\sum_{j_1\ldots j_n\in\mathbb{Z}_2^n}\left(\sqrt{\abs{b_{j_{\pi(1)}\ldots j_{\pi(n)}}}^2 + \abs{b_{[n](j_{\pi(1)}\ldots j_{\pi(n)})}}^2}\right)\ket{{j_1\ldots j_n}}\\
    &= \frac{1}{\sqrt{2}}\sum_{j_1\ldots j_n\in\mathbb{Z}_2^n}\left(\sqrt{\abs{b_{j_1\ldots j_n}}^2 + \abs{b_{[n](j_1\ldots j_n)}}^2}\right)\ket{{j_1\ldots j_n}}\\
    &=\ket{\psi_{\tilde{\mathcal{G}}}}.
\end{align}

Now, to confirm that $\ket{\psi_{\tilde{\mathcal{G}}}}$ is a TS state, we prove that $\bra{\psi_{\tilde{\mathcal{G}}}}\rbf\ket{v} = \delta_{T,T'}$ for any $T,T'\in\nmg$. To do so, we will need the fact that 
\begin{align}
    X^{\otimes n}\ket{\psi_{\mathcal{G}}} &= \sum_{j_1\ldots j_n\in\mathbb{Z}_2^n} b_{j_1\ldots j_n}\ket{{[n](j_1\ldots j_n)}}\\
    &= \sum_{j_1\ldots j_n\in\mathbb{Z}_2^n} b_{[n](j_1\ldots j_n)}\ket{{j_1\ldots j_n}}.\label{eq:bitflip-decomp}
\end{align}
Then for any $T,T'\in\mathcal{T}$,
\begin{align}
    \bra{\psi_{\tilde{\mathcal{G}}}}\rbf\ket{\psi_{\tilde{\mathcal{G}}}}&=\frac{1}{2} \sum_{j_1\ldots j_n\in\mathbb{Z}_2^n}\exp[i\left(\varphi_{j_1\ldots j_n}^{(T')}-\varphi_{j_1\ldots j_n}^{(T)}\right)]\left(\abs{b_{j_1\ldots j_n}}^2 + \abs{b_{x(j_1\ldots j_n)}}^2\right)\\
    &= \frac{1}{2} \left(\bra{\psi_{\mathcal{G}}}\rbf\ket{\psi_{\mathcal{G}}}+\bra{\psi_{\mathcal{G}}}X^{\otimes n}\rbf X^{\otimes n}\ket{\psi_{\mathcal{G}}}\right)\label{eq:symmetrize03}\\
    &= \frac{1}{2} \left(\bra{\psi_{\mathcal{G}}}\rbf\ket{\psi_{\mathcal{G}}}+\bra{\psi_{\mathcal{G}}}\mathbf{R}^{[n](T,T')}\ket{\psi_{\mathcal{G}}}\right)\label{eq:symmetrize04},
\end{align}
where Eq. (\ref{eq:symmetrize03}) follows from the application of Eq. (\ref{eq:bitflip-decomp}) and Eq. (\ref{eq:symmetrize04}) follows from Proposition \ref{prop:S-conj}. By Proposition \ref{prop:S-bitflip3}, $[n](T,T')\in\mathcal{T}^2$. Consequently, since $\ket{\psi_{\mathcal{G}}}$ is a TS state, $\bra{\psi_{\mathcal{G}}}\mathbf{R}^{[n](T,T')}\ket{\psi_{\mathcal{G}}}=\bra{\psi_{\mathcal{G}}}\mathbf{R}^{(T',T)}\ket{\psi_{\mathcal{G}}} = \delta_{T',T}$. Since $\delta_{T',T}=\delta_{T,T'}$, it follows that
\begin{align}
    \bra{\psi_{\tilde{\mathcal{G}}}}\rbf\ket{\psi_{\tilde{\mathcal{G}}}}&= \frac{1}{2} \left(\delta_{T,T'}+\delta_{T',T}\right)\\
    &= \delta_{T,T'}
\end{align}
for all $T,T'\in\mathcal{T}$, as desired. We conclude that $\mathcal{\ket{\psi_{\tilde{\mathcal{G}}}}}$ is a $\tilde{\mathcal{G}}$-invariant TS state, thereby completing the proof.
\end{proof}
\subsection{Proof of Proposition \ref{prop:HG-basis}}\label{appendix:HGbasis}

\begin{proof}[Proof of Proposition \ref{prop:HG-basis}]
We now show that the set of vectors 
\begin{align}
    \left\{\ket{\overline{\nu}}=\sum_{j_1\ldots j_n\in {\omega_\nu}}\ket{{j_1\ldots j_n}}\ :\ \nu=0,\ldots,N_{\tilde{G}}-1\right\}
\end{align}
forms an orthogonal basis for $\mathcal{H}_{\tilde{\mathcal{G}}}$, where the $\omega_\nu$ are orbits in $\mathbb{Z}_2^n/\tilde{G}$. The key insight is that the orbits $\mathbb{Z}_2^n/\tilde{G}$ form a partition for $\mathbb{Z}_2^n$. Subsequently, since the orbits are disjoint and the $Z$-eigenbasis is orthogonal, $\braket{\overline{\nu}|\overline{\nu}'}=0$ for $\nu\neq \nu'$. Next, we show that the set of $\ket{\overline{\nu}}$ spans $\mathcal{H}_{\tilde{\mathcal{G}}}$. Any $\ket{\psi}\in\mathcal{H}_{\tilde{\mathcal{G}}}$ can be written in the $Z$-eigenbasis as $\ket{\psi}=\sum_{j_1\ldots j_n} a_{j_1\ldots j_n}\ket{j_1\ldots j_n}$ for some $a_{j_1\ldots j_n}\in\mathbb{C}$. Since $\ket{\psi}=\Pi_{\tilde{\mathcal{G}}}\ket{\psi}$,
\begin{align}
    \ket{\psi}&=\Pi_{\tilde{\mathcal{G}}}\sum_{j_1\ldots j_n\in\mathbb{Z}_2^n} a_{j_1\ldots j_n}\ket{j_1\ldots j_n}\\
    &= \sum_{j_1\ldots j_n\in\mathbb{Z}_2^n} a_{j_1\ldots j_n}\Pi_{\tilde{\mathcal{G}}}\ket{\psi}\ket{j_1\ldots j_n}\\
    &= \frac{1}{\abs{\tilde{\mathcal{G}}}}\sum_{j_1\ldots j_n\in\mathbb{Z}_2^n} a_{j_1\ldots j_n}\sum_{(\varsigma,\pi)\in\tilde{G}}\ket{(\varsigma,\pi)[j_1\ldots j_n]}
\end{align}
due to Eq. (\ref{eq:tilde-mathcal-G-action}). Invoking the definition of an orbit, we then have
\begin{align}
    \ket{\psi}&= \frac{1}{\abs{\tilde{\mathcal{G}}}}\sum_{j_1\ldots j_n\in\mathbb{Z}_2^n} a_{j_1\ldots j_n}\sum_{j'_1\ldots j'_n\in \operatorname{Orb}_{\tilde{G}}[j_1\ldots j_n]}\ket{{j'_1\ldots j'_n}}\\
    &= \frac{1}{\abs{\tilde{\mathcal{G}}}}\sum_{\nu=0}^{N_{\tilde{G}-1}} \left(\sum_{j_1 \ldots j_n\in \omega_\nu}a_{j_1\ldots j_n}\right)\ket{\overline{\nu}}.
\end{align}
It follows that any $\ket{\psi}\in\mathcal{H}_{\tilde{\mathcal{G}}}$ can be written as a linear combination of $\ket{\overline{\nu}}$, so the set of $\ket{\overline{\nu}}$ is a basis for $\mathcal{H}_{\tilde{\mathcal{G}}}$.
\end{proof}

\subsection{Proof of Proposition \ref{prop:RG-props}}\label{appendix:R-prop}
We first prove a useful intermediate result, namely, that $\tilde{G}$ and $\tilde{\mathcal{G}}$ are isomorphic:
\begin{proposition}\label{prop:phi-iso}
    Given a permutation group $G$ on $[n]$, let $\mathcal{G} = P(G)$. Then the map $\Phi(\cdot)$ defined in Eq. (\ref{eq:phi-def}) is an isomorphism from $\tilde{\mathcal{G}}$ to $\tilde{G}$.
\end{proposition}
\begin{proof}
    Let $\mathcal{S} = \{I^{\otimes n},X^{\otimes n}\}$ and let $S_G = \{\varnothing,[n]\}$. Observe that $\abs{\tilde{\mathcal{G}}} = \abs{\mathcal{S}\mathcal{G}} = 2\abs{\mathcal{G}}$. Similarly, $\abs{\tilde{G}} = \abs{S_G\times G} = 2\abs{G}$. Since $G$ and $\mathcal{G}$ are isomorphic, $\abs{\tilde{\mathcal{G}}}=\abs{\tilde{G}}$. Since $\abs{\tilde{\mathcal{G}}}=\abs{\tilde{G}}$ and $\Phi$ is a surjective homomorphism from $\tilde{\mathcal{G}}$ to $\tilde{G}$ by Proposition \ref{prop:phi-homo}, $\Phi$ is an isomorphism.
\end{proof}
We can now prove Proposition \ref{prop:RG-props}.
\begin{proof}[Proof of Proposition \ref{prop:RG-props}]
Part (a). To prove that the operator $\mathbf{R}^{(\mu)}_{\tilde{\mathcal{G}}}(\theta)$ is Hermitian for any $\mu=0,\ldots, M_{\tilde{G}}$, we show that $\mathbf{R}^{\dagger(\mu)}_{\tilde{\mathcal{G}}}=\mathbf{R}^{(\mu)}_{\tilde{\mathcal{G}}}$. Pick any $(T,T')\in\Omega_\mu$. Then $\mathbf{R}^{(\mu)}_{\tilde{\mathcal{G}}}=\mathbf{R}^{(T,T')}_{\tilde{\mathcal{G}}}$ and 
\begin{align}
    \mathbf{R}^{\dagger(T,T')}_{\tilde{\mathcal{G}}}&=\left(\Pi_{\tilde{\mathcal{G}}} R^{\dagger(T)}R^{(T')}\Pi_{\tilde{\mathcal{G}}}\right)^\dagger\\
    &=\Pi_{\tilde{\mathcal{G}}}R^{\dagger(T')}R^{(T)}\Pi_{\tilde{\mathcal{G}}}\\
    &=\Pi_{\tilde{\mathcal{G}}}\left(X^{\otimes n} R^{(T')}X^{\otimes n}\right)\left(X^{\otimes n}R^{\dagger(T)}X^{\otimes n}\right)\Pi_{\tilde{\mathcal{G}}}\\
    &=\Pi_{\tilde{\mathcal{G}}}X^{\otimes n} R^{(T')}R^{\dagger(T)}X^{\otimes n}\Pi_{\tilde{\mathcal{G}}}\\
    &=\Pi_{\tilde{\mathcal{G}}}R^{(T')}R^{\dagger(T)}\Pi_{\tilde{\mathcal{G}}}\\
    &=\Pi_{\tilde{\mathcal{G}}}R^{\dagger(T)}R^{(T')}\Pi_{\tilde{\mathcal{G}}} \\
    &= \mathbf{R}^{(T,T')}_{\tilde{\mathcal{G}}}
\end{align}
as desired, noting that the $R^{(T)}$ operators commute. The identity $X^{\otimes n}\Pi_{\tilde{\mathcal{G}}} = \Pi_{\tilde{\mathcal{G}}}$ follows from the fact that $X^{\otimes n}\in \tilde{\mathcal{G}}$.

Part (b). To show that the $\ket{\overline{\nu}}$ vectors are eigenvectors of $\mathbf{R}^{(\mu)}_{\tilde{\mathcal{G}}}(\theta)$, we first need to precisely describe how a given $Z$-eigenbasis vector $\ket{j_1\ldots j_n}$ projects onto $\mathcal{H}_{\tilde{\mathcal{G}}}$. For clarity, we will write elements 
$(\varsigma,\pi)$ of $\tilde{G}$ in the more compact form $\tilde{\pi}$. Then we have
\begin{align}
    \Pi_{\tilde{\mathcal{G}}}\ket{j_1\ldots j_n} &= \frac{1}{\abs{\tilde{\mathcal{G}}}} \sum_{U\in\tilde{\mathcal{G}}} U \ket{j_1\ldots j_n}\\
    &= \frac{1}{\abs{\tilde{\mathcal{G}}}} \sum_{U\in\tilde{\mathcal{G}}}  \ket{\Phi(U)(j_1\ldots j_n)}\\
    &= \frac{1}{\abs{\tilde{G}}}\sum_{\tilde{\pi}\in\tilde{G}}\ket{\tilde{\pi}(j_1\ldots j_n)},
\end{align}
where the second equality follows from Eq. (\ref{eq:tilde-mathcal-G-action}) and the third equality follows from the fact that $\Phi$ is an isomorphism (Proposition \ref{prop:phi-iso}). Now let $\mathfrak{S} ={\{\tilde{\pi}\in \tilde{G}\ :\ \tilde{\pi}(j_1\ldots j_n)= j_1\ldots j_n\}}$ be the stabilizer of the bit-string $j_1\ldots j_n$. Since $\mathfrak{S}$ is a subgroup of $\tilde{G}$, $\tilde{G}$ is partitioned by the left cosets of $\mathfrak{S}$. Hence, by separating the elements of $\tilde{G}$ into cosets, we can write
\begin{align}
    \Pi_{\tilde{\mathcal{G}}}\ket{j_1\ldots j_n} &= \frac{1}{\abs{\tilde{G}}}\sum_{\tilde{\pi}\in\tilde{G}}\ket{\tilde{\pi}(j_1\ldots j_n)}\\
    &= \frac{1}{\abs{\tilde{G}}}\sum_{\mathrm{cosets}\ \tilde{\pi}\mathfrak{S}}
    \left(\sum_{\tilde{\pi}'\in \tilde{\pi}\mathfrak{S}}\ket{\tilde{\pi}'(j_1\ldots j_n)}\right)\label{eq:orbit-stab01}
\end{align}
Note that all elements in a given coset $\tilde{\pi}\mathfrak{S}$ act identically on $j_1\ldots j_n$. For any $\tilde{\pi}'\in \tilde{\pi}\mathfrak{S}$, we can write $\tilde{\pi}'=\tilde{\pi}s$ for some $s\in \mathfrak{S}$, and
\begin{align}
    \tilde{\pi}'(j_1\ldots j_n) = \tilde{\pi}s(j_1\ldots j_n) = \tilde{\pi}(j_1\ldots j_n). \label{eq:orbit-stab02}
\end{align}
Thus, all terms in the inside the parentheses of Eq. (\ref{eq:orbit-stab01}) are identical, so 
\begin{align}
    \Pi_{\tilde{\mathcal{G}}}\ket{j_1\ldots j_n} &= \frac{1}{\abs{\tilde{G}}}\sum_{\mathrm{cosets}\ \tilde{\pi}\mathfrak{S}}\abs{\tilde{\pi}\mathfrak{S}}\ket{\tilde{\pi}(j_1\ldots j_n)}\\
    &= \frac{\abs{\mathfrak{S}}}{\abs{\tilde{G}}}\sum_{\mathrm{cosets}\ \tilde{\pi}\mathfrak{S}}\ket{\tilde{\pi}(j_1\ldots j_n)}
\end{align}
since the size of all cosets is $\abs{\mathfrak{S}}$. Now let $\omega_\nu=\operatorname{Orb}_{\tilde{G}}[j_1\ldots j_n]$ be the orbit of $j_1\ldots j_n$ under $\tilde{G}$. Then the map from cosets to $\omega$ defined by $f(\tilde{\pi}\mathfrak{S})=\tilde{\pi}(j_1\ldots j_n)$ is bijective. To see this, first note that $f$ is a valid map since if $\tilde{\pi}\mathfrak{S}=\tilde{\pi}'\mathfrak{S}$, then Eq. (\ref{eq:orbit-stab02}) guarantees that $f(\tilde{\pi}\mathfrak{S})=f(\tilde{\pi}'\mathfrak{S})$. Next, $f$ is surjective because every $j'_1\ldots j'_n\in\omega$ can be written as $\tilde{\pi}'(j_1\ldots j_n)$ for some $\tilde{\pi}'\in\tilde{G}$, and $\tilde{\pi}'$ belongs to the coset $\tilde{\pi}'\mathfrak{S}$. Lastly, if $f(\tilde{\pi}\mathfrak{S})=f(\tilde{\pi}'\mathfrak{S})$, then $\tilde{\pi}(j_1\ldots j_n)=\tilde{\pi}'(j_1\ldots j_n)$ which implies $\tilde{\pi}^{-1}\tilde{\pi}'\in \mathfrak{S}$; because $\tilde{\pi}'=\tilde{\pi}(\tilde{\pi}^{-1}\tilde{\pi}')$ and the cosets partition the group, the cosets $\tilde{\pi}\mathfrak{S}$ and $\tilde{\pi}'\mathfrak{S}$ are equal, so $f$ is injective. Since $f$ is injective and surjective, it is bijective. Then,
\begin{align}
    \Pi_{\tilde{\mathcal{G}}}\ket{j_1\ldots j_n} &= \frac{\abs{\mathfrak{S}}}{\abs{\tilde{G}}}\sum_{\mathrm{cosets}\ \tilde{\pi}\mathfrak{S}}\ket{f(\tilde{\pi}\mathfrak{S})}\\
    &= \frac{\abs{\mathfrak{S}}}{\abs{\tilde{G}}}\sum_{j'_1\ldots j'_n\in\omega_\nu}\ket{j'_1\ldots j'_n}\\
    &= \frac{\abs{\mathfrak{S}}}{\abs{\tilde{G}}}\ket{\overline{\nu}}\\
    &= \frac{1}{\abs{\omega_\nu}}\ket{\overline{\nu}}
\end{align}
where the last equality follows from the orbit-stabilizer theorem, which gives $\abs{\tilde{G}}/\abs{\mathfrak{S}}=\abs{\omega}$.

We now use this fact to show that the $\ket{\overline{\nu}}$ are eigenvectors of $\mathbf{R}^{(\mu)}_{\tilde{\mathcal{G}}}$. Pick any $(T,T')\in\Omega_\mu$. Then
\begin{align}
    \mathbf{R}^{(\mu)}_{\tilde{\mathcal{G}}}\ket{\overline{\nu}} &= \mathbf{R}^{(T,T')}_{\tilde{\mathcal{G}}}\ket{\overline{\nu}}\\
    &=\Pi_{\tilde{\mathcal{G}}}R^{\dagger(T)}R^{(T')}\Pi_{\tilde{\mathcal{G}}}\ket{\overline{\nu}}\\
    &= \Pi_{\tilde{\mathcal{G}}}R^{\dagger(T)}R^{(T')}\ket{\overline{\nu}}\\
    &=\sum_{j_1\ldots j_n\in {\omega_\nu}}\Pi_{\tilde{\mathcal{G}}}R^{\dagger(T)}R^{(T')}\ket{{j_1\ldots j_n}}\\
    &=\sum_{j_1\ldots j_n\in {\omega_\nu}}\exp(i\varphi_{j_1\ldots j_n}^{(T')}-i\varphi_{j_1\ldots j_n}^{(T)})\Pi_{\tilde{\mathcal{G}}}\ket{{j_1\ldots j_n}}\\
    &=\left[\frac{1}{\abs{\omega_\nu}}\sum_{j_1\ldots j_n\in {\omega_\nu}}\exp(i\varphi_{j_1\ldots j_n}^{(T')}-i\varphi_{j_1\ldots j_n}^{(T)})\right]\ket{\overline{\nu}},
\end{align}
where $\varphi_{j_1\ldots j_n}^{(T)}$ is the eigenphase of $R^{(T)}$ given by Eq. (\ref{eq:R-eigphase-def}).
Hence, $\ket{\overline{\nu}}$ is an eigenvector of $\mathbf{R}^{(\mu)}_{\tilde{\mathcal{G}}}$ with eigenvalue 
\begin{align}\label{eq:lambda-eigval}
    \lambda_{\mu,\nu}(\theta)=\frac{1}{\abs{\omega_\nu}}\sum_{j_1\ldots j_n\in {\omega_\nu}}\exp(i\varphi_{j_1\ldots j_n}^{(T')}(\theta)-i\varphi_{j_1\ldots j_n}^{(T)}(\theta)),
\end{align}
where $(T,T')$ can equivalently be chosen to be any trajectory pair in $\Omega_\mu$.
Since $\mathbf{R}^{(\mu)}_{\tilde{\mathcal{G}}}$ is Hermitian, the eigenvalues $\lambda_{\mu,\nu}(\theta)$ are real.
\end{proof}

\subsection{Proof of Proposition \ref{prop:MNbounds}}\label{appendix:MNbounds}
In this section, we first prove Proposition \ref{prop:MNbounds}. Afterward, we provide lower bounds on $M_{\tilde{G}}$ and $N_{\tilde{G}}$.

\begin{proof}[Proof of Proposition \ref{prop:MNbounds}]
To derive the upper bound on $M_{\tilde{G}}$, consider the orbit $\Omega$ of a pair $(T,T')$ under $\tilde{G}$. The set
\begin{align}
    \mathcal{X}&={\{(\varnothing,\pi)[(T,T')]\ :\ \pi\in G\}}\\
    &={\{(\pi(T),\pi(T'))\ :\ \pi\in G\}}
\end{align}
obtained by applying all permutations to $(T,T')$ is a subset of $\Omega$. Define the map $f:\mathcal{T}^2\rightarrow\mathcal{T}$ such that $f(T_1,T_2)=T_1$ for $T_1,T_2\in\mathcal{T}$. Because the image of $\mathcal{X}$ under $f$ is ${\{\pi(T):\pi\in G\}}=\mathcal{T}$, it must be true that $\abs{\mathcal{X}}\geq \abs{\mathcal{T}}$. It follows that each orbit has size at least $\abs{\mathcal{T}}$, and since the orbits partition $\mathcal{T}^2$, there cannot be more than $\abs{\mathcal{T}^2}/\abs{\mathcal{T}} = \abs{\mathcal{T}}$ orbits.

Let $e$ be the identity permutation. For the upper bound on $N_{\tilde{G}}$, note that all orbits have size at least two, since $\big\{j_1\ldots j_n,([n],e)[j_1\ldots j_n]\big\}\subseteq \operatorname{Orb}_{\tilde{G}}[j_1\ldots j_n]$ and $j_1\ldots j_n\neq([n],e)[j_1\ldots j_n]$ always. Since the orbits partition the set $\mathbb{Z}_2^n$, there cannot be more than $2^{n-1}$ orbits.

\end{proof}
We now provide lower bounds on $M_{\tilde{G}}$ and $N_{\tilde{G}}$.
\begin{proposition}
    The following lower bounds hold on the size of $N_{\tilde{G}}$ and $M_{\tilde{G}}$ for any permutation group $G$ on $[n]$, assuming $\mathcal{T}$ is $G$-transitive:
    \begin{align}
       M_{\tilde{G}} \geq \frac{1}{2\abs{G}}\left(\abs{\mathcal{T}}^2+\abs{\mathcal{T}}\right)\ \mathrm{and}\\
       N_{\tilde{G}}\geq \frac{1}{\abs{G}}\big(2^{n-1}+\abs{G}-1\big)
    \end{align}
\end{proposition}
\begin{proof}
    For the lower bound on $M_{\tilde{G}}$, we use Burnside's lemma, which states that for a group $H$ acting on a set $\mathcal{X}$, the number of orbits $\abs{\mathcal{X}/H}$ is given by
    \begin{align}
        \abs{\mathcal{X}/H} = \frac{1}{\abs{H}}\sum_{h\in H}\abs{\operatorname{Fix}_\mathcal{X}(h)},
    \end{align}
    where $\operatorname{Fix}_\mathcal{X}(h)$ is the subset of $\mathcal{X}$ fixed by the element $h$. Choosing $H=\tilde{G}$ and $\mathcal{X}=\mathcal{T}^2$, we can apply the lemma to compute $M_{\tilde{G}}$ as follows:
    \begin{align}
        M_{\tilde{G}} &= \frac{1}{\abs{\tilde{G}}}\sum_{(\varsigma,\pi)\in \tilde{G}}\abs{\operatorname{Fix}_{\nmg^2}((\varsigma,\pi))}\\
        &= \frac{1}{2\abs{G}}\big(\abs{\operatorname{Fix}_{\nmg^2}((\varnothing,e))}+\abs{\operatorname{Fix}_{\nmg^2}(([n],e))}+\cdots\big).
    \end{align}
    The observation that $\abs{\tilde{G}}=2\abs{G}$ follows from the fact that $\tilde{G}=\{\varnothing,[n]\}\times G$. Note that all trajectory pairs in $\nmg^2$ are fixed by $(\varnothing,e)$ and that pairs of the form $(T,T)$ are fixed by $([n],e)$. Thus,
    \begin{align}
        M_{\tilde{G}} \geq \frac{1}{2\abs{G}}\left(\abs{\mathcal{T}}^2+\abs{\mathcal{T}}\right).
    \end{align}
    
    The lower bound on $N_{\tilde{G}}$ also follows from Burnside's lemma as follows:
    \begin{align}
        N_{\tilde{G}} &= \frac{1}{\abs{\tilde{G}}}\sum_{(\varsigma,\pi)\in \tilde{G}}\abs{\operatorname{Fix}_{\mathbb{Z}_2^n}((\varsigma,\pi))}\\
        &= \frac{1}{2\abs{G}}\left[\abs{\operatorname{Fix}_{\mathbb{Z}_2^n}((\varnothing,e))}+\left(\sum_{\pi\in G,\pi\neq e}\abs{\operatorname{Fix}_{\mathbb{Z}_2^n}((\varnothing,\pi))}\right)+\cdots\right].
    \end{align}
    Note that all $2^n$ elements of $\mathbb{Z}_2^n$ are fixed by $(\varnothing,e)$, while at least two strings (namely, $0\ldots 0$ and $1\ldots 1$) are fixed by $(\varnothing,\pi)$ for all $\pi\neq e$. Thus,
    \begin{align}
        N_{\tilde{G}} &\geq \frac{1}{2\abs{G}}\left[2^n+\left(\sum_{\pi\in G,\pi\neq e}2\right)\right]\\
        &=\frac{1}{2\abs{G}}\big[2^n+2(\abs{G}-1)\big]\\
        &=\frac{1}{\abs{G}}\big(2^{n-1}+\abs{G}-1\big).
    \end{align}
\end{proof}

\subsection{Proof of Propositions \ref{prop:sym-orbits-strings} and  \ref{prop:sym-orbits-pairs}}\label{appendix:orbs-pairs}
In this section, we prove Propositions $\ref{prop:sym-orbits-strings}$ and  \ref{prop:sym-orbits-pairs}, beginning with Proposition $\ref{prop:sym-orbits-strings}$.
\begin{proof}[Proof of Proposition \ref{prop:sym-orbits-strings}]
    We must show that two strings $j_1\ldots j_n$ and $j'_1\ldots j'_n$ are in the same orbit of $\mathbb{Z}_2^n/\tilde{\Sigma}_n$ if and only if they belong to the same $\mathcal{W}_\nu\cup\mathcal{W}_{n-\nu}$, where $\nu\in\{0,\ldots, n\}$. By Eq. (\ref{eq:tilde-G-def}), we can write $\tilde{\Sigma}_n$ as $\tilde{\Sigma}_n=S\times \Sigma_n$, where $S=\{\varnothing,[n]\}$. 
    
    ``$\implies$" direction: Observe that a permutation $\pi\in\Sigma_n$ does not change the weight of a string and that the $[n]$ element of $S$ changes the weight of a string from $\nu$ to $n-\nu$. If $j'_1\ldots j'_n = (\varsigma,\pi)[j_1\ldots j_n]$ for some $(\varsigma,\pi)\in\tilde{\Sigma}_n$ (i.e., the strings are in the same orbit) and $j_1\ldots j_n\in \mathcal{W}_\nu$, then it follows that $j'_1\ldots j'_n$ must be in either $\mathcal{W}_\nu$ or $\mathcal{W}_{n-\nu}$. 
    
    ``$\impliedby$" direction: Assume $j_1\ldots j_n$ and $j'_1\ldots j'_n$ both belong to ${\mathcal{W}_\nu\cup\mathcal{W}_{n-\nu}}$ for some $\nu$. Without loss of generality, they are either both in $\mathcal{W}_\nu$, or one is in $\mathcal{W}_\nu$ while the other is in $\mathcal{W}_{n-\nu}$. In the former case, there exists a permutation $\pi\in\Sigma_n$ taking one string to the other because they have the same weight; in the latter case, some permutation $\pi\in\Sigma_n$ transforms one string into the bit-flipped version of the other. In either case, we conclude that some $(\varsigma,\pi)\in\tilde{\Sigma}_n$ maps one string to the other, so they must be in the same orbit. 
    
    Lastly, the observation that ${\mathcal{W}_\nu\cup\mathcal{W}_{n-\nu}}={\mathcal{W}_{n-\nu}\cup\mathcal{W}_{\nu}}$ implies that there are only $N_{\tilde{\Sigma}_n}=\floor{n/2}+1$ distinct orbits.
\end{proof}

We now prove Proposition \ref{prop:sym-orbits-pairs}.
\begin{proof}[Proof of Proposition \ref{prop:sym-orbits-pairs}]
    We must show that $(T_1,T_1')$ and $(T_2,T_2')$ are in the same orbit of $\Tsymsq/\tilde{\Sigma}_n$ if and only if they belong to the same $\mathcal{D}_\mu$. By Eq. (\ref{eq:tilde-G-def}), we can write $\tilde{\Sigma}_n$ as $\tilde{\Sigma}_n=S\times \Sigma_n$, where $S=\{\varnothing,[n]\}$. 
    
    ``$\implies$" direction: Note that $(T,T')$ and $[n][(T,T')]=(T',T)$ have the same degree since $m=\abs{T}=\abs{T'}$ implies $\abs{T\setminus T'}=\abs{T'\setminus T}$. Additionally, we now show that $(T,T')$ and $\pi(T,T') = (\pi(T),\pi(T'))$ must have the same degree because permutations are bijective. Note that if $j\in T\cap T'$, then $\pi(j)\in \pi(T)$ and $\pi(j)\in \pi(T')$, so $\pi(j)\in \pi(T)\cap\pi(T')$ and $\pi(T\cap T')\subseteq \pi(T)\cap\pi(T')$. Conversely, if $j'\in \pi(T)\cap\pi(T')$, then $\pi^{-1}(j)\in 
    T$ and $\pi^{-1}(j)\in T'$, so $\pi^{-1}(j)\in T\cap T'$. It follows that $j\in\pi(T\cap T')$, which implies $\pi(T)\cap\pi(T')\subseteq \pi(T\cap T')$ and therefore $\pi(T)\cap\pi(T')=\pi(T\cap T')$. Using this result, we can show that the degrees are equal:
    \begin{align}
        \abs{T\setminus T'} = \abs{T}-\abs{T\cap T'} = \abs{\pi(T)}-\abs{\pi(T\cap T')} =\abs{\pi(T)}-\abs{\pi(T)\cap \pi(T')} = \abs{\pi(T)\setminus\pi(T')},
    \end{align}
    where in the second equality we have used the bijectivity of $\pi$. Subsequently, if $(T_2,T_2')=(\varsigma,\pi)[(T_1,T_1')]$ for some $(\varsigma,\pi)\in\tilde{\Sigma}_n$, then the two pairs must have the same degree. 
    
    ``$\impliedby$" direction: We now show that if two pairs $(T_1,T_1')$ and $(T_2,T_2')$ have the same degree, then there exists $(\varsigma,\pi)\in\tilde{\Sigma}_n$ mapping one to the other, implying they belong to the same orbit. Note that since all four trajectories $T_1,T_1',T_2,T_2'$ are the same size, $\abs{T_1\setminus T_1'}=\abs{T_2\setminus T_2'}$ implies that $\abs{T_1'\setminus T_1}=\abs{T_2'\setminus T_2}$ and 
    \begin{align}
        \abs{T_1\cap T_1'} = \abs{T_1}-\abs{T_1\setminus T_1'} = \abs{T_2}-\abs{T_2\setminus T_2'} = \abs{T_2\cap T_2}.
    \end{align}
    Since a bijection exists sets of the same size, there exists bijections between $T_1\setminus T_1'$ and $T_2\setminus T_2'$, $T_1'\setminus T_1$ and $T_2'\setminus T_2$, and $T_1\cap T_1'$ and $T_2\cap T_2'$. Because the union of bijections with disjoint domains and disjoint codomains is a bijection, we can define a permutation $\pi\in\Sigma_n$ which is the union of these three bijections such that 
    \begin{align}
        \pi(T_1\setminus T_1')=T_2\setminus T_2',\ \pi(T_1'\setminus T_1)=T_2'\setminus T_2',\ \mathrm{and}\ \pi(T_1\cap T_1')=T_2\cap T_2'.
    \end{align}
    It follows that $\pi(T_1,T_1')=(T_2,T_2')$, which implies $(\varnothing,\pi)[(T_1,T_1')]=(T_2,T_2')$, as desired.

    The number of orbits $M_{\tilde{\Sigma}_n}$ is the number of possible degrees that an arbitrary allowed trajectory pair could have. If $m\leq \floor{n/2}$, a trajectory pair could have any degree between $0$ and $m$ inclusive, where these extreme values are respectively attained when the trajectories are either completely overlapping or completely disjoint. In contrast, if $m > \floor{n/2}$, then any two trajectories are guaranteed to overlap by at least one qubit, and the possible degrees are the values $0,\ldots,n-m$. In summary,
    \begin{align}
        M_{\tilde{\Sigma}_n} = \begin{cases}
            m+1 &m\leq \floor*{\frac{n}{2}}\\
            n-m+1 &m>\floor*{\frac{n}{2}}.
        \end{cases}
    \end{align}
\end{proof}

\subsection{Proof of Proposition \ref{prop:A-entries}}\label{appendix:A-entries}
In this section, we prove Proposition \ref{prop:A-entries}.
\begin{proof}[Proof of Proposition \ref{prop:A-entries}]
For given $\mu\in\{0,\ldots, M-1\}$ and $\nu \in\{0,\ldots N-1\}$, let $(T,T')$ be any trajectory pair in the orbit $\Omega_\mu = \mathcal{D}_\mu$. Using Eq. (\ref{eq:lambda-eigval}), we can write $A_{\mu,\nu}(\theta)=\lambda_{\mu,\nu}(\theta)\abs{\omega_\nu}$ as
\begin{align}
    A_{\mu,\nu}(\theta)&= \sum_{j_1\ldots j_n\in {\omega_\nu}}\exp(i\varphi_{j_1\ldots j_n}^{(T')}(\theta)-i\varphi_{j_1\ldots j_n}^{(T)}(\theta))\\
    &= \sum_{j_1\ldots j_n\in {\mathcal{W}_\nu\cup\mathcal{W}_{n-\nu}}}\exp(i\varphi_{j_1\ldots j_n}^{(T')}(\theta)-i\varphi_{j_1\ldots j_n}^{(T)}(\theta))
\end{align}
where $\varphi_{j_1\ldots j_n}^{(T)}$ is the eigenphase of $R^{(T)}$ given by Eq. (\ref{eq:R-eigphase-def}).
Note if $n$ is even and $\nu=n/2$, then $\mathcal{W}_\nu=\mathcal{W}_{n-\nu}$. Otherwise, $\mathcal{W}_\nu$ and $\mathcal{W}_{n-\nu}$ are disjoint. In the former case, the $[n]$ element of $S=\{\varnothing,[n]\}$ acts as a bijection from $\mathcal{W}_\nu$ to itself; in the latter, $[n]$ acts as a bijection between the elements of $\mathcal{W}_\nu$ and $\mathcal{W}_{n-\nu}$. These observations allow us to write the sum over $\mathcal{W}_\nu\cup\mathcal{W}_{n-\nu}$ as just a sum over $\mathcal{W}_\nu$:
\begin{align}
    A_{\mu,\nu}&= \frac{\alpha_\nu}{2}\sum_{j_1\ldots j_n\in \mathcal{W}_\nu}\left[\exp(i\varphi_{j_1\ldots j_n}^{(T')}-i\varphi_{j_1\ldots j_n}^{(T)}) 
    +\exp(i\varphi_{[n](j_1\ldots j_n)}^{(T')}-i\varphi_{[n](j_1\ldots j_n)}^{(T)})\right]
\end{align}
where $\alpha_\nu=1$ if $n$ is even and $\nu=N-1=n/2$, and $\alpha_\nu=2$ otherwise. This additional factor is necessary because the above summation double-counts the elements of $\mathcal{W}_\nu$ if $n$ is even and $\nu=n/2$. Now observe that 
\begin{align}
    \varphi_{[n](j_1\ldots j_n)}^{(T)} = \frac{\theta}{2}\sum_{k\in T} (-1)^{(j_k\oplus 1)+1} = -\frac{\theta}{2}\sum_{k\in T} (-1)^{j_k+1}=-\varphi_{j_1\ldots j_n}^{(T)}.
\end{align}
Hence,
\begin{align}
    A_{\mu,\nu}&= \frac{\alpha_\nu}{2}\sum_{j_1\ldots j_n\in \mathcal{W}_\nu}\left[\exp(i\varphi_{j_1\ldots j_n}^{(T')}-i\varphi_{j_1\ldots j_n}^{(T)}) 
    +\exp(-i\varphi_{j_1\ldots j_n}^{(T')}+i\varphi_{j_1\ldots j_n}^{(T)})\right]\\
    &= \alpha_\nu\sum_{j_1\ldots j_n\in \mathcal{W}_\nu}\cos(\varphi_{j_1\ldots j_n}^{(T')}-\varphi_{j_1\ldots j_n}^{(T)}). 
\end{align}
Substituting in Eq. (\ref{eq:R-eigphase-def}) for $\varphi_{j_1\ldots j_n}^{(T)}$, we can simplify the expression for $A_{\mu,\nu}$ using the intuition that any qubits in $T\cap T'$ do not acquire a phase because the $R^{\dagger}_Z$ they receive from trajectory $T$ cancels with the $R_Z$ from $T'$: that is, ${\bra{\psi}R^{\dagger(T)}R^{(T')}\ket{\psi}}={\bra{\psi}R^{\dagger(T\setminus T')}R^{(T'\setminus T)}\ket{\psi}}$. Thus,
\begin{align}
    A_{\mu,\nu}
    &=\alpha_\nu\sum_{j_1\ldots j_n\in \mathcal{W}_\nu}\cos\left[\frac{\theta}{2}\left(\sum_{k\in T'} (-1)^{j_k+1}-\sum_{l\in T} (-1)^{j_{l}+1}\right)\right]\\
    &=\alpha_\nu\sum_{j_1\ldots j_n\in \mathcal{W}_\nu}\cos\left[\frac{\theta}{2}\left(\sum_{k\in T'\setminus T} (-1)^{j_k+1}+\sum_{k'\in T'\cap T} (-1)^{j_{k'}+1}-\sum_{l'\in T\cap T'} (-1)^{j_{l'}+1}-\sum_{l\in T\setminus T'} (-1)^{j_{l}+1}\right)\right]\\
    &=\alpha_\nu\sum_{j_1\ldots j_n\in \mathcal{W}_\nu}\cos\left[\frac{\theta}{2}\left(\sum_{k\in T'\setminus T} (-1)^{j_k+1}-\sum_{l\in T\setminus T'} (-1)^{j_{l}+1}\right)\right].
\end{align}
Note that for $k\in T'\setminus T$,
\begin{align}
    \sum_{k} (-1)^{j_k+1} = \sum_{j_k=1}j_k -\sum_{j_k=0}j_k = \sum_{k}j_k - \left(\abs{T'\setminus T} - \sum_{k}j_k\right) = -\mu + 2\sum_k j_k,
\end{align}
since $\mu=\abs{T'\setminus T}$ is the degree of $(T,T')$. Substituting this result into the expression for $A_{\mu,\nu}$, we obtain
\begin{align}
    A_{\mu,\nu}&=\alpha_\nu\sum_{j_1\ldots j_n\in \mathcal{W}_\nu}\cos\left[\frac{\theta}{2}\left(-\mu + 2\sum_{k\in T'\setminus T} j_k-\left(-\mu + 2\sum_{l\in T\setminus T'} j_l\right)\right)\right]\\
    &=\alpha_\nu\sum_{j_1\ldots j_n\in \mathcal{W}_\nu}\cos\left[\theta\left( \sum_{k\in T'\setminus T} j_k-\sum_{l\in T\setminus T'} j_l\right)\right]\label{eq:A-matrix01}
\end{align}
Given $j_1\ldots j_n$, let $i=\sum_{k\in T'\setminus T} j_k$ and $i'=\sum_{l\in T\setminus T'} j_l$. Intuitively, $i$ (resp. $i'$) is the number of qubits in $T'\setminus T$ (resp. $T\setminus T'$) which are in the state $\ket{1}$ when the array is prepared to $\ket{j_1\ldots j_n}$; likewise, $\nu$ is the total number of $\ket{1}$-qubits in the whole array. The number of $\ket{1}$-qubits in $T'\setminus T$ can take any integer value between $0$ and $\min(\abs{T'\setminus T},\nu)$, so
\begin{align}\label{eq:i-bounds}
    0\leq i\leq \min(\mu,\nu).
\end{align}
To determine the possible values of $i'$, we must consider how the remaining $\nu-i$ $\ket{1}$-qubits can be distributed between $T\setminus T'$ and $T\cap T'$, which respectively hold $\mu$ and $n-2\mu$ qubits total. The following constraints apply to $i'$:
\begin{enumerate}
    \item $T\setminus T'$ cannot contain more than $\abs{T\setminus T'}$ $\ket{1}$-qubits, so $i'\leq \mu$.
    \item $T\setminus T'$ cannot contain more than the remaining number of $\ket{1}$-qubits, so $i'\leq \nu-i$.
    \item In the event that the remaining number of $\ket{1}$-qubits exceeds the size of $T\cap T'$ (i.e., $\nu-i >\abs{T\cap T'}$), then $T\setminus T'$ must at least contain those $\ket{1}$-qubits which do not fit in $T\cap T'$, so $i\geq \nu-i-(n-2\mu)$.
\end{enumerate}
Synthesizing the three above constraints, we deduce that 
\begin{align}\label{eq:j-bounds}
    \max[0,\nu-i-(n-2\mu)]\leq i' \leq \min(\mu,\nu-i).
\end{align}
We now convert the sum over $j_1\ldots j_n$ in Eq. (\ref{eq:A-matrix01}) to a sum over $i$ and $i'$. Recall that $i$ and $i'$ depend on $j_1\ldots j_n$, and it remains to be determined how many $j_1\ldots j_n\in\mathcal{W}_\nu$ are associated with a particular value of $i,i'$. For any $i,i'$ satisfying the above bounds, there are $\binom{\mu}{i}$ ways to arrange $i$ $\ket{1}$-qubits in $T'\setminus T$, $\binom{\mu}{i'}$ ways to arrange $i'$ $\ket{1}$-qubits in $T\setminus T'$, and $\binom{n-2\mu}{\nu-(i+i')}$ ways to arrange the remaining qubits in $T\cap T'$. Eq. (\ref{eq:A-matrix01}) then becomes
\begin{align}
    A_{\mu,\nu}(\theta) = \alpha_\nu\sum_{i,i'}\binom{\mu}{i}\binom{\mu}{i'}\binom{n-2\mu}{\nu-(i+i')}\cos{[(i-i')\theta]}
\end{align}
where $i$ and $i'$ are summed over the intervals given by Eqs. (\ref{eq:i-bounds}) and (\ref{eq:j-bounds}). In fact, $i$ and $i'$ can both be summed from $0$ to $\nu$ since at least one of the binomial coefficients becomes zero when $i$ or $i'$ exceeds the bounds of Eqs. (\ref{eq:i-bounds}) and (\ref{eq:j-bounds}).
\end{proof}

\subsection{Proof of Lemma \ref{lemma:deriv-A}}\label{appendix:deriv-A}
In this section, we prove Lemma \ref{lemma:deriv-A}.
\begin{proof}[Proof of Lemma \ref{lemma:deriv-A}]
Suppose $\mathcal{T}=\Tsym$. We now prove Lemma \ref{lemma:deriv-A}, which states that 
\begin{align}
    (\mu-j)A_{\mu,\nu}^{(j)}(t)-\mu A_{\mu-1,\nu}^{(j)}(t)=(t-1)A_{\mu,\nu}^{(j+1)}(t)\nonumber
\end{align}
for all $\mu=1,\ldots,M-1$, $\nu=0,\ldots,N-1$ and integers $j$ such that $0\leq j\leq \mu-1$, where $A_{\mu,\nu}^{(j)}(t)$ is the $j$-th derivative of $A_{\mu,\nu}(t)$ with respect to $t$. We proceed by induction on $j$. If Eq. (\ref{eq:deriv-A}) holds for one value of $j$, then it also holds for $j+1$. To see this, differentiate both sides of Eq. (\ref{eq:deriv-A}):
\begin{align}
    \frac{d}{dt}\left[(\mu-j)A_{\mu,\nu}^{(j)}(t)-\mu A_{\mu-1,\nu}^{(j)}(t)\right]&=\frac{d}{dt}\left[(t-1)A_{\mu,\nu}^{(j+1)}(t)\right]
\end{align}
which gives
\begin{align}
    (\mu-j)A_{\mu,\nu}^{(j+1)}(t)-\mu A_{\mu-1,\nu}^{(j+1)}(t)&=(t-1)A_{\mu,\nu}^{(j+2)}(t) + A_{\mu,\nu}^{(j+1)}(t)
\end{align}
and 
\begin{align}
    (\mu-(j+1))A_{\mu,\nu}^{(j+1)}(t)-\mu A_{\mu-1,\nu}^{(j+1)}(t)&=(t-1)A_{\mu,\nu}^{(j+2)}(t) 
\end{align}
upon rearrangement, as desired. It remains to show that Eq. (\ref{eq:deriv-A}) holds for the base case when $j=0$. 

The base case appears to be challenging to demonstrate, but it may be verified through the following approach. First, we make the change of variables $\abs{i-i'}\rightarrow k$ to rewrite the entries of $A(t)$ as
\begin{align}
    A_{\mu,\nu}(t) &= \alpha_\nu \sideset{}{'}\sum_{k=0}^\nu\sum_{i=0}^\nu \left[\binom{\mu}{i}\binom{\mu}{i+k}\binom{n-2\mu}{\nu-(2i+k)}+\binom{\mu}{i+k}\binom{\mu}{i}\binom{n-2\mu}{\nu-(2i+k)}\right]T_k(t)\\
    &= 2\alpha_\nu \sideset{}{'}\sum_{k=0}^{\nu}\sum_{i=0}^\nu\binom{\mu}{i}\binom{\mu}{i+k}\binom{n-2\mu}{\nu-(2i+k)}T_k(t)
\end{align}
where the primed summation indicates that the $k=0$ term should be halved to avoid double counting. For convenience, we will define
\begin{align}
    C_{\mu,k} = \sum_{i=0}^\nu\binom{\mu}{i}\binom{\mu}{i+k}\binom{n-2\mu}{\nu-(2i+k)}
\end{align}
such that 
\begin{align}
    A_{\mu,\nu} = 2\alpha_\nu\sideset{}{'}\sum_{k=0}^\nu C_{\mu,k} T_k(t).
\end{align}
The base case to be proven is:
\begin{align}\label{eq:base-case}
    \mu A_{\mu,\nu}(t) -\mu A_{\mu-1,\nu}(t) = (t-1)A^{(1)}_{\mu,\nu}(t).
\end{align}
The right-hand side (RHS) can be expanded as 
\begin{align}
    (t-1)A^{(1)}_{\mu,\nu}(t)=2\alpha_\nu\sideset{}{'}\sum_{k=0}^\nu C_{\mu,k} (t-1) \frac{d}{dt}T_k(t).
\end{align}
Using the identity $\frac{d}{dt}T_k(t) = k U_{k-1}(t)$, where $U_k(t)$ is the $k$th Chebyshev polynomial of the second kind defined by $U_k(\cos(\theta))\sin(\theta)=\sin((k+1)\theta)$, the RHS becomes
\begin{align}
    (t-1)A^{(1)}_{\mu,\nu}(t)&=2\alpha_\nu\sideset{}{'}\sum_{k=0}^\nu kC_{\mu,k} (t-1)U_{k-1}(t) \\
    &=2\alpha_\nu\sideset{}{'}\sum_{k=0}^\nu kC_{\mu,k} (tU_{k-1}(t)-U_{k-1}(t))\\
    &=2\alpha_\nu\sum_{k=1}^\nu kC_{\mu,k} (tU_{k-1}(t)-U_{k-1}(t)).
\end{align}
Note that we drop the $k=0$ term because $U_{-1}(t)=0$. Additionally, without the $k=0$ term, we drop the prime symbol on the summation. The identity $T_{k'}(t)U_{k}(t)=\frac{1}{2}(U_{k'+k}(t)+U_{k-k'}(t))$ (which holds for $k\geq k'-1$) implies that $tU_{k-1}(t) = \frac{1}{2}(U_k(t)+U_{k-2}(t))$ when $k\geq 0$ since $t=T_1(t)$. Making this substitution, we obtain
\begin{align}
    (t-1)A^{(1)}_{\mu,\nu}(t)&=\alpha_\nu\sum_{k=1}^\nu kC_{\mu,k} \left[(U_k(t)+U_{k-2}(t))-2U_{k-1}(t)\right]
\end{align}
Collecting like terms of $U_k(t)$, we have
\begin{align}
    (t-1)A^{(1)}_{\mu,\nu}(t)&=\alpha_\nu\sum_{k=0}^\nu \left[kC_{\mu,k} + (k+2)C_{\mu,k+2}- 2(k+1)C_{\mu,k+1}\right]U_k(t)
\end{align}
for the RHS. Similarly, using the identity $T_k(t) = \frac{1}{2}(U_{k}(t)-U_{k-2}(t))$, the left-hand side (LHS) can be expanded as
\begin{align}
    \mu A_{\mu,\nu}(t) -\mu A_{\mu-1,\nu}(t) &= 2\mu\alpha_\nu \sideset{}{'}\sum_{k=0}^\nu (C_{\mu,k}-C_{\mu-1,k})T_k(t)\\
    &= \mu\alpha_\nu \sideset{}{'}\sum_{k=0}^\nu (C_{\mu,k}-C_{\mu-1,k})(U_{k}(t)-U_{k-2}(t))\\
    &=  \mu\alpha_\nu \left[\frac{1}{2}(C_{\mu,0}-C_{\mu-1,0})(U_{0}(t)-U_{-2}(t))+\sum_{k=1}^\nu(C_{\mu,k}-C_{\mu-1,k})(U_{k}(t)-U_{k-2}(t))\right]
\end{align}
where we have separated out the $k=0$ term from the summation. Since $U_{-2}(t)=-U_0(t)$, we then have
\begin{align}
    \mu A_{\mu,\nu}(t) -\mu A_{\mu-1,\nu}(t) &=\mu\alpha_\nu \left[(C_{\mu,0}-C_{\mu-1,0})U_{0}(t)+\sum_{k=1}^\nu (C_{\mu,k}-C_{\mu-1,k})(U_{k}(t)-U_{k-2}(t))\right]\\
    &= \mu\alpha_\nu \sum_{k=0}^\nu \left[(C_{\mu,k}-C_{\mu-1,k})-(C_{\mu,k+2}-C_{\mu-1,k+2})\right]U_k(t).
\end{align}

By comparing the expanded expressions for the LHS and RHS and matching the $U_k(t)$ terms, we see that the base case of Eq. (\ref{eq:base-case}) holds if 
\begin{align}
    \mu \left[(C_{\mu,k}-C_{\mu-1,k})-(C_{\mu,k+2}-C_{\mu-1,k+2})\right] = kC_{\mu,k} + (k+2)C_{\mu,k+2}- 2(k+1)C_{\mu,k+1}\label{eq:coeffs-relation}
\end{align}
for $k=0,\ldots \nu$. However, the $C_{\mu,k}$ are sums of products of binomial coefficients, and a simple closed form for these coefficients has not been found. Let $\ell_k$ and $r_k$ equal the left and right sides of Eq. (\ref{eq:coeffs-relation}), respectively. To verify that $\ell_k=r_k$ for all $k=0,\ldots,\nu$, we will derive generating functions for the sequences $(\ell_k)_{k=0}^\nu$ and $(r_k)_{k=0}^\nu$ and show that the generating functions are equal. For a generating function $f(x)$, we will use the notation $[x^k]f(x)$ to represent the coefficient of $x^k$ in the formal power series expansion of $f$. By the binomial theorem, we have $\binom{n}{k}=[x^k](1+x)^n$. Thus, we can write each $C_{\mu,k}$ as
\begin{align}
    C_{\mu,k}&= \sum_{i=0}^\infty\binom{\mu}{i}\binom{\mu}{i+k}\binom{n-2\mu}{\nu-(2i+k)}\label{eq:infinite}\\
    &=\sum_{i=0}^\infty\binom{\mu}{i}[x^{i+k}](1+x)^\mu [w^{\nu-(2i+k)}](1+w)^{n-2\mu}\\
    &=\sum_{i=0}^\infty\binom{\mu}{i}[x^{k}]x^{-i}(1+x)^\mu [w^{\nu-k}]w^{2i}(1+w)^{n-2\mu}\\
    &=[x^{k}][w^{\nu-k}](1+x)^\mu(1+w)^{n-2\mu}\sum_{i=0}^\infty\binom{\mu}{i}\left(w^2x^{-1}\right)^i\\
    &=[x^{k}][w^{\nu-k}](1+x)^\mu(1+w)^{n-2\mu}\left(1+w^2x^{-1}\right)^\mu
\end{align}
where in the last equality we have employed the binomial theorem. Note that the upper bound of summation in Eq. (\ref{eq:infinite}) can be chosen to be infinity because the third binomial coefficient in the summand will evaluate to zero if $i>\nu$. For convenience, define 
\begin{align}
    g_k(x,w)=(1+x)^\mu(1+w)^{n-2\mu}\left(1+w^2x^{-1}\right)^\mu
\end{align}
such that $C_{\mu,k}=[x^k][w^{\nu-k}]g_k(x,w)$. We can now derive similar expressions for the remaining terms in $\ell_k$:
\begin{align}
    C_{\mu-1,k}&= [x^{k}][w^{\nu-k}](1+x)^{\mu-1}(1+w)^{n-2\mu+2}\left(1+w^2x^{-1}\right)^{\mu-1}\\
    &= [x^{k}][w^{\nu-k}]\frac{(1+w)^2}{(1+x)(1+w^2x^{-1})}g_k(x,w),\\
    C_{\mu,k+2} &= [x^{k+2}][w^{\nu-k-2}](1+x)^\mu(1+w)^{n-2\mu}\left(1+w^2x^{-1}\right)^\mu\\
    &= [x^{k}][w^{\nu-k}]w^2x^{-2}g_k(x,w),\ \mathrm{and}\\
    C_{\mu-1,k+2} &= [x^{k}][w^{\nu-k}]w^2x^{-2}\frac{(1+w)^2}{(1+x)(1+w^2x^{-1})}g_k(x,w).
\end{align}
Combining the above results, we can write $\ell_k$ as 
\begin{align}
    \ell_k &= [x^{k}][w^{\nu-k}] \mu\left(1-\frac{(1+w)^2}{(1+x)(1+w^2x^{-1})}\right)(1-w^2x^{-2}) g_k(x,w)\\
    &= [x^{k}][w^{\nu-k}] \mu\left(\frac{(1+ x + w^2x^{-1} +w^2)-(1+w)^2}{(1+x)(1+w^2x^{-1})}\right)(1-w^2x^{-2}) g_k(x,w)\\
    &= [x^{k}][w^{\nu-k}] \mu\left(\frac{x+w^2x^{-1}-2w}{(1+x)(1+w^2x^{-1})}\right)(1-w^2x^{-2}) g_k(x,w)\label{eq:lk}
\end{align}
For $r_k$, we need to find generating functions for expressions like $kC_{\mu,k}$. Fortunately, we can employ a clever trick to find these functions: note for a generating function $f(x)$ that $k[x_k]f(x) = [x_k]x\frac{d}{dx}f(x)$. Hence,
\begin{align}
    kC_{\mu,k} &= k[x_k][w^{\nu-k}]g_k(x,w)\\
    &= [x^k][w^{\nu-k}]x\frac{d}{dx}g_k(x,w)\\
    &= [x^k][w^{\nu-k}]x\left[\mu(1+x)^{\mu-1}(1+w)^{n-2\mu}(1+w^2x^{-1})^{\mu}-\mu w^2x^{-2}(1+x)^{\mu}(1+w)^{n-2\mu}(1+w^2x^{-1})^{\mu-1}\right]\\
    &= [x^k][w^{\nu-k}]\mu x\left(\frac{1}{1+x}-\frac{ w^2x^{-2}}{1+w^2x^{-1}}\right)g_k(x,w)\\
    &= [x^k][w^{\nu-k}] \frac{\mu x(1-w^2x^{-2})}{(1+x)(1+w^2x^{-1})}g_k(x,w)
\end{align}
as well as
\begin{align}
    (k+2)C_{\mu,k+2} &= [x^k][w^{\nu-k}] w^2x^{-2}\frac{\mu x(1-w^2x^{-2})}{(1+x)(1+w^2x^{-1})}g_k(x,w)\ \mathrm{and}\\
    (k+1)C_{\mu,k+1}&= [x^k][w^{\nu-k}]wx^{-1}\frac{\mu x(1-w^2x^{-2})}{(1+x)(1+w^2x^{-1})}g_k(x,w).
\end{align}
Combining the above, we obtain
\begin{align}
    r_k&= [x^k][w^{\nu-k}](1+w^2x^{-2}-2wx^{-1})\frac{\mu x(1-w^2x^{-2})}{(1+x)(1+w^2x^{-1})}g_k(x,w)\\
    &= [x^k][w^{\nu-k}]\mu \left(\frac{x+w^2x^{-1}-2w}{(1+x)(1+w^2x^{-1})}\right)(1-w^2x^{-2})g_k(x,w).\label{eq:rk}
\end{align}
Finally, comparing Eqs. (\ref{eq:lk}) and (\ref{eq:rk}), we see that $\ell_k$ and $r_k$ are coefficients of the same term in the power series expansion of the same function. It follows that $\ell_k=r_k$ for all $\nu=0,\ldots k$, which implies that the two sides of Eq. (\ref{eq:base-case}) are equal, thereby completing the proof.
\end{proof}
\subsection{Proof of Theorem \ref{thm:red-lin-prog}}\label{appendix:red-lin-prog}
In this section, we prove Theorem \ref{thm:red-lin-prog}. 
\begin{proof}[Proof of Theorem \ref{thm:red-lin-prog}]
We first compute explicit expressions for the entries of $A'(t)$ when $\mathcal{T}=\Tsym$. If $\mu=0$, then $A'_{\mu,\nu}(t)=A_{\mu,\nu}(t)$. If $\mu \geq 1$, then
\begin{align}
    A'_{\mu,\nu}(t) &= \frac{d^{\mu-1}}{dt^{\mu-1}} A_{\mu,\nu}\\
    &= \alpha_\nu\sum_{i,i'=0}^\nu\binom{\mu}{i}\binom{\mu}{i'}\binom{n-2\mu}{\nu-(i+i')}\frac{d^{\mu-1}}{dt^{\mu-1}}T_{\abs{i-i'}}(t).
\end{align}
Since $T_\abs{i-i'}(t)$ is a polynomial of degree $\abs{i-i'}$, the $(\mu-1)$th derivative of $T_\abs{i-i'}(t)$ will be zero if $\abs{i-i'}< \mu-1$. Furthermore, $\abs{i-i'}$ cannot exceed $\nu$, so if $\nu<\mu-1$, then $A'_{\mu,\nu}$ will automatically be zero. Subsequently, assume $\nu\geq\mu-1$. $A_{\mu,\nu}$ does not contain any Chebyshev polynomials with degree greater than $\mu$, since if $\abs{i-i'}>\mu$, at least one of the binomial coefficients will evaluate to zero. Thus, the only polynomials not sent to zero after differentiation will be $T_{\mu-1}(t)$ and $T_{\mu}(t)$. Keeping only terms proportional to these polynomials, and letting $T^{(\mu-1)}_k(t)$ be the $(\mu-1)$th derivative of $T_k(t)$, we obtain
\begin{align}
    A'_{\mu,\nu}(t)  &= 2\alpha_\nu\binom{\mu}{0}\binom{\mu}{\mu}\binom{n-2\mu}{\nu-\mu}T^{(\mu-1)}_{\mu}(t)+ 2\alpha_\nu\left[\binom{\mu}{0}\binom{\mu}{\mu-1}\binom{n-2\mu}{\nu-\mu+1}+\binom{\mu}{1}\binom{\mu}{\mu}\binom{n-2\mu}{\nu-\mu-1}\right]T^{(\mu-1)}_{\mu-1}(t) \\
    &= 2\alpha_\nu\binom{n-2\mu}{\nu-\mu}T^{(\mu-1)}_{\mu}(t)+ 2\alpha_\nu\mu\left[\binom{n-2\mu}{\nu-\mu+1}+\binom{n-2\mu}{\nu-\mu-1}\right]T^{(\mu-1)}_{\mu-1}(t).
\end{align}
Given that the leading coefficient of $T_k(t)$ is $2^{k-1}$, we deduce that $T_k^{(k)}(t)=2^{k-1}k!$. Furthermore, since $T_k$ is either even or odd, the degree $k-1$ term has coefficient zero, which implies that $T_k^{(k-1)}(t)=2^{k-1}k!t$. We use these facts to simplify the above expression to
\begin{align}
    A'_{\mu,\nu}(t)  &= 2\alpha_\nu\binom{n-2\mu}{\nu-\mu}2^{\mu-1}\mu!t+ 2\alpha_\nu\mu\left[\binom{n-2\mu}{\nu-\mu+1}+\binom{n-2\mu}{\nu-\mu-1}\right]2^{\mu-2}(\mu-1)!\\
    &= 2^\mu \mu! \alpha_\nu\left\{\binom{n-2\mu}{\nu-\mu}t+ \frac{1}{2}\left[\binom{n-2\mu}{\nu-\mu+1}+\binom{n-2\mu}{\nu-\mu-1}\right]\right\}.\label{eq:Ap-entries}
\end{align}

We are now equipped to prove Theorem \ref{thm:red-lin-prog}. We first consider the case where $\theta=0$ and $t=1$ so that the incident particle does not interact with the sensor at all. If $M=1$, then $A'(t)=A(t)$, so Theorem \ref{thm:lin-prog} and Theorem \ref{thm:red-lin-prog} are equivalent. Thus assume $M>1$. Then $1\leq m\leq n-1$, which means $\abs{\mathcal{T}}=\binom{n}{m}>1$. Clearly, it is not possible for $R^{(T)}(0)\ket{\psi}$ and $R^{(T')}(0)\ket{\psi}$ to be orthogonal for any $T\neq T'$ and $\ket{\psi}\in\mathcal{H}$, so no TS state exists. We now show by contradiction that Eq. (\ref{eq:red-lin-prog}) has no solution under these circumstances. Assume Eq. (\ref{eq:red-lin-prog}) has a solution when $M>1$ and $t=1$.  Note that all $A'_{\mu,\nu} \geq 0$ if $t=1$. For any $\mu\geq 1$, if $A'_{\mu,\nu} > 0$, then $c_\nu$ must be zero since Eq. (\ref{eq:red-lin-prog}) requires $\sum_{\nu}A'_{\mu,\nu}c_\nu = 0$ and $c_\nu\geq 0$. Inspecting Eq. (\ref{eq:Ap-entries}), we see that $A_{1,\nu} >0$ for all $\nu=0,\ldots, N-1$ (equivalently, all $\nu=0,\ldots, \floor{n/2}$) if $t=1$. Hence, all $c_\nu$ must be zero. However, this leads to a contradiction because the first row of the system $A'\mathbf{c}=\mathbf{d}$ cannot be satisfied, as $\sum_{\nu}A_{0,\nu} c_\nu\neq 1$. We conclude that Theorem \ref{thm:red-lin-prog} holds when $t=1$.

Now suppose $t\neq 1$. ``$\implies$" direction: Suppose there exists a TS state. Then Theorem \ref{thm:lin-prog} guarantees that there exists $\mathbf{c}\in \mathbb{R}^N$ such that $A(t)\mathbf{c}=\mathbf{d}$ and $\mathbf{c}\geq 0$. We show that this $\mathbf{c}$ also solves $A'(t)\mathbf{c}=\mathbf{d}$. Let $f_\mu(t)=\sum_\nu A_{\mu,\nu}c_\nu$ be the product of the $\mu$th row of $A(t)$ with $\mathbf{c}$. Then Lemma \ref{lemma:deriv-A} implies that the $j$th derivatives of $f$ with respect to $t$ satisfy
\begin{align}\label{eq:f-deriv}
    (\mu-j)f_{\mu}^{(j)}(t)+\mu f_{\mu-1}^{(j)}(t)=(t-1)f_{\mu}^{(j+1)}(t)
\end{align}
for all $j=0\ldots,\mu-1$. We now prove via induction on increasing $j$ that $f^{(j)}_\mu=0$ for all $j=0,\ldots, M-2$ and $\mu=j+1,\ldots,M-1$. For the base case, $A(t)\mathbf{c}=\mathbf{d}$ guarantees that $f^{(0)}_\mu=f_\mu=0$ for all $\mu\geq 1$. Now, if $f_\mu^{(j)}=0$ for all $\mu=j+1,\ldots, M-1$, then Eq. (\ref{eq:f-deriv}) directly implies that $f_\mu^{(j+1)}=0$ for $\mu=j+2,\ldots, M-1$ since $t\neq 1$; the claim follows. Since $f^{(j)}_{j+1}=0$ for all $j=0,\ldots, M-2$, we can change indices $j\rightarrow\mu-1$ to recover $f_\mu^{(\mu-1)}=0$ for all $\mu=1,\ldots, M-1$, so rows $1$ through $M-1$ of the system $A'(t)\mathbf{c}=\mathbf{d}$ must hold. Because the remaining first row of the system is identical to that of $A(t)\mathbf{c}=\mathbf{d}$, we conclude that $A'(t)\mathbf{c}=\mathbf{d}$.

``$\impliedby$" direction: Suppose there exists $\mathbf{c}\in\mathbb{R}^N$ such that $A'(t)\mathbf{c}=\mathbf{d}$ and $\mathbf{c}\geq 0$. We will show that this $\mathbf{c}$ also satisfies $A(t)\mathbf{c}=\mathbf{d}$, which guarantees the existence of a TS state by Theorem \ref{thm:lin-prog}. We similarly prove by induction that $f^{(j)}_\mu=0$ for $j=0,\ldots M-2$ and all $\mu=j+1\ldots, M-1$, but this time the procedure is more subtle. We perform a ``nested" induction: we induct on decreasing $j$ (starting at $j=M-2$), and at each $j$-step we induct again on increasing $\mu$ (starting at $\mu=j+1$). At the $j$-th outer step, assume we have shown that $f_\mu^{(j+1)}=0$ for $\mu=j+2,\ldots M-1$. To show the outer induction step, i.e., $f^{(j)}_\mu=0$ for all $\mu=j+1\ldots, M-1$, we induct on increasing $\mu$, assuming at the $\mu$th inner step that $f^{(j)}_{\mu-1}=0$. Eq. (\ref{eq:f-deriv}) implies the inner induction step, i.e., $f_\mu^{(j)}=0$, since the inner (resp. outer) induction hypothesis assures that $f_{\mu-1}^{(j)}=0$ (resp. $f_\mu^{(j+1)}=0$). The base case for both the inner and outer inductions is given by $A'(t)\mathbf{c}=\mathbf{d}$, which implies that $f_\mu^{(\mu-1)}=0$ for all $\mu=1,\ldots, M-1$ or equivalently $f^{(j)}_{j+1}=0$ for all $j=0,\ldots, M-2$. With the claim proven, we deduce that $f_\mu^{(0)}=f_\mu=0$ for all $\mu=1,\ldots,M-1$, which implies that rows $1$ through $M-1$ of the system $A(t)\mathbf{c}=\mathbf{d}$ must hold. Because the remaining first row of the system is identical to that of $A'(t)\mathbf{c}=\mathbf{d}$, we conclude that $A(t)\mathbf{c}=\mathbf{d}$, thereby completing the proof.
\end{proof}
\subsection{Proof of Proposition \ref{prop:rref}}\label{appendix:TSsym-coeffs}
In this section, we prove Proposition \ref{prop:rref}. 
\begin{proof}[Proof of Proposition \ref{prop:rref}]
Suppose $\mathcal{T}=\Tsym$. When $M=N=1$, the proposition is trivial. Hence, we will prove that for $M=N> 1$, the vector $\mathbf{c}\in\mathbb{R}^N$ satisfies $A'(t)\mathbf{c}=\mathbf{d}$ if and only if $\sum_\nu A'_{0,\nu} c_\nu = 1$ and all of its entries obey
\begin{align}
    c_\nu=\begin{cases}
        (-1)^{N-1-\nu}T_{N-1-\nu}(t)c_{N-1} &n\ \mathrm{even}\\
        (-1)^{N-1-\nu}W_{N-1-\nu}(t)c_{N-1} &n\ \mathrm{odd},
    \end{cases}\nonumber
\end{align}
where $W_k(t)$ is the $k$th Chebyshev polynomial of the fourth kind:
\begin{align}
    W_k(\cos{\theta}) = \frac{\sin{((k+1/2)\theta)}}{\sin{(\theta/2)}}
\end{align}
which follows the recursion relation
\begin{align}
    W_0(t)=1,\ W_1(t) = 2t+1,\ W_k(t) = 2tW_{k-1}(t)-W_{k-2}(t).
\end{align}
Furthermore, we define $W_{-1}(t)=-1$ by convention. We first prove the forward direction. Suppose that that $M=N=\floor{n/2}+1$, and first assume that $n$ is even. Let $g_\mu(t)=\sum_\nu A'_{\mu,\nu}c_\nu$ be the product of the $\mu$th row of $A'(t)$ with $\mathbf{c}$. Observe that $A'(t)\mathbf{c}=\mathbf{d}$ gives $g_{\mu}(t) = 0$ for $\mu = 1,\ldots, N-1$ and that $A'_{\mu, \nu}$ is nonzero only for $\nu\geq \mu-1$. We can then use this fact to write each $c_{\nu}$ in terms of the other $c_{\nu'}$ satisfying $\nu'>\nu$: 
\begin{align}
    0&=\frac{1}{2^\nu \nu!} g_{\nu}(t)\\
    &= \sum_{\nu'=\nu-1}^{N-1}\alpha_{\nu'}\left\{\binom{n-2\nu}{\nu'-\nu}t+ \frac{1}{2}\left[\binom{n-2\nu}{\nu'-\nu+1}+\binom{n-2\nu}{\nu'-\nu-1}\right]\right\}c_{\nu'}\\
    &= c_{\nu-1} + \sum_{\nu'=\nu}^{N-1}\alpha_{\nu'}\left\{\binom{n-2\nu}{\nu'-\nu}t+ \frac{1}{2}\left[\binom{n-2\nu}{\nu'-\nu+1}+\binom{n-2\nu}{\nu'-\nu-1}\right]\right\}c_{\nu'}
\end{align}
which upon rearrangement yields
\begin{align}
    c_{\nu-1} &= -\sum_{\nu'=\nu}^{N-1}\alpha_{\nu'}\left\{\binom{n-2\nu}{\nu'-\nu}t+ \frac{1}{2}\left[\binom{n-2\nu}{\nu'-\nu+1}+\binom{n-2\nu}{\nu'-\nu-1}\right]\right\}c_{\nu'}. \label{eq:TSsym-coeffs02}
\end{align}
We proceed to prove that Eq. (\ref{eq:TSsym-coeffs}) holds by induction on decreasing $\nu$, starting with $\nu=N-1$. There are two base cases to cover: $\nu=N-1$ and $\nu=N-2$. Since $T_0(t)=1$, Eq. (\ref{eq:TSsym-coeffs}) clearly holds for $\nu=N-1$. For $\nu = N-2$, we have
\begin{align}
    c_{N-2} = -\left\{\binom{0}{0}t+ \frac{1}{2}\left[\binom{0}{1}+\binom{0}{-1}\right]\right\}c_{N-1} = tc_{N-1} = T_1(t)c_{N-1},
\end{align}
as desired. For $\nu< N-2$, we will now show that if Eq. (\ref{eq:TSsym-coeffs}) holds for $\nu,\ldots,N-1$, then it also holds for $\nu-1$. Thus assume Eq. (\ref{eq:TSsym-coeffs}) holds for $\nu,\ldots,N-1$. Substituting the induction hypothesis into Eq. (\ref{eq:TSsym-coeffs02}), we obtain
\begin{align}
    c_{\nu-1} &= \sum_{\nu'=\nu}^{N-1}(-1)^{N-\nu'}\alpha_{\nu'}\left\{\binom{n-2\nu}{\nu'-\nu}t+ \frac{1}{2}\left[\binom{n-2\nu}{\nu'-\nu+1}+\binom{n-2\nu}{\nu'-\nu-1}\right]\right\}T_{N-1-\nu'}(t)c_{N-1}.
\end{align}
To prove the induction step for $\nu< N-2$, it is sufficient to show that
\begin{align}\label{eq:TSsym-coeffs01}
    \sum_{\nu'=\nu}^{N-1}(-1)^{N-\nu'}\alpha_{\nu'}\left\{\binom{n-2\nu}{\nu'-\nu}t+ \frac{1}{2}\left[\binom{n-2\nu}{\nu'-\nu+1}+\binom{n-2\nu}{\nu'-\nu-1}\right]\right\}T_{N-1-\nu'}(t) = (-1)^{N-\nu}T_{N-\nu}(t).
\end{align}
We now separate out the $\nu'=N-1=n/2$ term, noting that $\alpha_{\nu'}=1$ in when $\nu'=N-1$ and 2 otherwise. Additionally, we use the identity $T_k(t)T_{k'}(t) = \frac{1}{2}(T_{k+k'}(t)+T_{\abs{k-k'}}(t))$ (which holds for $k,k'\geq 0$) to evaluate the product $t T_{N-1-\nu'}(t)$, noting that $t=T_1(t)$. The left side of Eq. (\ref{eq:TSsym-coeffs01}) then becomes
\begin{align}
     &-\left\{\binom{n-2\nu}{\frac{n}{2}-\nu}t+\frac{1}{2}\left[\binom{n-2\nu}{\frac{n}{2}-\nu+1}+\binom{n-2\nu}{\frac{n}{2}-\nu-1}\right]\right\}\nonumber\\
    &+\sum_{\nu'=\nu}^{N-2}(-1)^{N-\nu'}\left\{\binom{n-2\nu}{\nu'-\nu}\left[T_{N-\nu'}(t)+T_{N-\nu'-2}(t)\right]+ \left[\binom{n-2\nu}{\nu'-\nu+1}+\binom{n-2\nu}{\nu'-\nu-1}\right]T_{N-1-\nu'}(t)\right\}.
\end{align}
To simplify the algebra, we change to the new variables $r= \nu'-\nu$ and $s = N-1-\nu=\frac{n}{2}-\nu$. The above expression is then
\begin{align}
     &-\left\{\binom{2s}{s}t+\frac{1}{2}\left[\binom{2s}{s+1}+\binom{2s}{s-1}\right]\right\}\nonumber\\
    &+\sum_{r=0}^{s-1}(-1)^{s-r+1}\left\{\binom{2s}{r}\left[T_{s-r+1}(t)+T_{s-r-1}(t)\right]+ \left[\binom{2s}{r+1}+\binom{2s}{r-1}\right]T_{s-r}(t)\right\}.
\end{align}
We now carefully combine like terms of $T_{s-r}(t)$ to obtain
\begin{align}
    &-\frac{1}{2}\left[\binom{2s}{s+1}+\binom{2s}{s-1}\right]+\binom{2s}{s-1}\nonumber\\
    &+\sum_{r=0}^{s-1}\left\{(-1)^{s-r+1} \left[\binom{2s}{r+1}+\binom{2s}{r-1}\right]+(-1)^{s-r+2}\binom{2s}{r-1}+(-1)^{s-r}\binom{2s}{r+1}\right\}T_{s-r}(t)\nonumber\\
    &+(-1)^{s+1}T_{s+1}(t).
\end{align}
Every term except the final one cancels in the expression above (the terms in the first row cancel because $\binom{2s}{s+1}=\binom{2s}{s-1}$). We thus conclude that the left side of Eq. (\ref{eq:TSsym-coeffs01}) equals $(-1)^{s+1}T_{s+1}(t)=(-1)^{N-\nu}T_{N-\nu}(t)$ as desired, which completes the proof of the induction step.  

It remains to show that Eq. (\ref{eq:TSsym-coeffs}) holds when $n$ is odd. The proof is much the same, except we now use polynomials $W_k(t)$ instead of $T_k(t)$. Since we will need to evaluate terms like $tW_k(t)$, we first prove the following identity:
\begin{align}\label{eq:Wk-identity}
    tW_k(t) = \frac{1}{2}(W_{k+1}(t)+W_{k-1}(t))
\end{align}
which holds for $k\geq 0$. To do so, we need the following preliminary fact:
\begin{align}\label{eq:Wk-identity2}
    W_k(t) = U_k(t)+U_{k-1}(t)
\end{align}
for $k\geq 0$, where $U_k(t)$ is the $k$th Chebyshev polynomial of the second kind. $U_k(t)$ follows the recursion relation $U_0(t) = 1$, $U_1(t)=2t$, and $U_{k+1}(t) = 2t U_{k}(t)-U_{k-1}$. Since $U_{-1}=0$ by convention, we can write $W_0(t)=U_0(t)+U_{-1}(t)$ and $W_1(t)=U_1(t)+U_0(t)$. Then, assuming Eq. (\ref{eq:Wk-identity2}) holds for $0,\ldots, k$, we see that for $k\geq 2$:
\begin{align}
    W_{k}(t) &= 2tW_{k-1}(t)-W_{k-2}(t)\\
    &= 2t(U_{k-1}(t)+U_{k-2}(t)) - (U_{k-2}(t)+U_{k-3}(t))\\
    &= (2tU_{k-1}(t)-U_{k-2}(t)+(2tU_{k-2}(t)-U_{k-3}(t))\\
    &= U_{k}(t)+U_{k-1}(t).
\end{align}
Eq. (\ref{eq:Wk-identity2}) then follows by induction. Returning to Eq. (\ref{eq:Wk-identity}), we use the identity $T_{k'}(t)U_{k}(t)=\frac{1}{2}(U_{k'+k}(t)+U_{k-k'}(t))$ (which holds for $k\geq k'-1$) that we employed earlier in Appendix \ref{appendix:deriv-A}, noting that $t=T_1(t)$:
\begin{align}
    tW_k(t) &= tU_{k}(t)+tU_{k-1}(t)\\
    &= \frac{1}{2}(U_{k+1}(t) + U_{k-1}(t)) + \frac{1}{2}(U_k(t)+U_{k-2}(t))\\
    &= \frac{1}{2}(U_{k+1}(t)+U_{k}(t))+\frac{1}{2}(U_{k-1}(t)+U_{k-2}(t))\\
    &= \frac{1}{2}(W_{k+1}(t)+W_{k-1}(t)),
\end{align}
as desired. Observe that Eq. (\ref{eq:Wk-identity}) holds for $k\geq 0$ using the conventions $U_{-2}(t)=-1$ and $W_{-1}(t)=-1$.

We continue now to show Eq. (\ref{eq:TSsym-coeffs}) holds for $n$ odd; accordingly, assume $n$ is odd. We induct again on decreasing $\nu$, starting with $\nu=N-1 = \frac{n-1}{2}$; the base case $\nu=N-1$ holds trivially. For $\nu=N-2= \frac{n-1}{2}-1$, we deduce from Eq. (\ref{eq:TSsym-coeffs02}) that 
\begin{align}
    c_{N-2} = -2\left\{\binom{1}{0}t+ \frac{1}{2}\left[\binom{1}{1}+\binom{1}{-1}\right]\right\}c_{N-1} = (2t+1)c_{N-1} = W_1(t)c_{N-1},
\end{align}
as desired. Note that $\alpha_{N-1}$ now equals 2 because $n$ is odd. Next assume Eq. (\ref{eq:TSsym-coeffs}) holds for $\nu,\ldots, N-1$, and substitute the induction hypothesis into Eq. (\ref{eq:TSsym-coeffs02}):
\begin{align}
    c_{\nu-1} &= -\sum_{\nu'=\nu}^{N-1}\left\{2\binom{n-2\nu}{\nu'-\nu}t+ \left[\binom{n-2\nu}{\nu'-\nu+1}+\binom{n-2\nu}{\nu'-\nu-1}\right]\right\}W_{N-1-\nu'}(t)c_{N-1}. 
\end{align}
To prove the induction step for $\nu<N-2$, it is sufficient to show that
\begin{align}\label{eq:TSsym-coeffs03}
    \sum_{\nu'=\nu}^{N-1}(-1)^{N-\nu'}\left\{2\binom{n-2\nu}{\nu'-\nu}t+ \left[\binom{n-2\nu}{\nu'-\nu+1}+\binom{n-2\nu}{\nu'-\nu-1}\right]\right\}W_{N-1-\nu'}(t) = (-1)^{N-\nu}W_{N-\nu}(t).
\end{align}
We separate out the $\nu'=N-1=\frac{n-1}{2}$ term and use the identity from Eq. (\ref{eq:Wk-identity}) to evaluate $tW_{N-1-\nu'}(t)$. The left side of Eq. (\ref{eq:TSsym-coeffs03}) becomes
\begin{align}
    &-\left\{2\binom{n-2\nu}{\frac{n-1}{2}-\nu}t+ \left[\binom{n-2\nu}{\frac{n-1}{2}-\nu+1}+\binom{n-2\nu}{\frac{n-1}{2}-\nu-1}\right]\right\}\\
    &+\sum_{\nu'=\nu}^{N-2}(-1)^{N-\nu'}\left\{\binom{n-2\nu}{\nu'-\nu}\left[W_{N-\nu'}(t)+W_{N-\nu'-2}(t)\right]+ \left[\binom{n-2\nu}{\nu'-\nu+1}+\binom{n-2\nu}{\nu'-\nu-1}\right]W_{N-1-\nu'}(t)\right\}
\end{align}
We change again to the new variables $r= \nu'-\nu$ and $s = N-1-\nu=\frac{n-1}{2}-\nu$:
\begin{align}
    &-\left\{\binom{2s+1}{s}(W_1(t)-1) +\left[\binom{2s+1}{s+1}+\binom{2s+1}{s-1}\right]\right\}\\
    &+\sum_{r=0}^{s-1}(-1)^{s-r+1}\left\{\binom{2s+1}{r}\left[W_{s-r+1}(t)+W_{s-r-1}(t)\right]+ \left[\binom{2s+1}{r+1}+\binom{2s+1}{r-1}\right]W_{s-r-1}(t)\right\}
\end{align}
Finally, we combine like terms:
\begin{align}
    &\binom{2s+1}{s}-\left[\binom{2s+1}{s+1}+\binom{2s+1}{s-1}\right] + \binom{2s+1}{s-1}\\
    &+\sum_{r=0}^{s-1}\left\{(-1)^{s-r+1} \left[\binom{2s+1}{r+1}+\binom{2s+1}{r-1}\right]+(-1)^{s-r+2}\binom{2s+1}{r-1}+(-1)^{s-r}\binom{2s+1}{r+1}\right\}W_{s-r}(t)\\
    &+(-1)^{s+1}W_{s+1}(t).
\end{align}
All terms except the last cancel (note that $\binom{2s+1}{s+1} = \binom{2s+1}{s}$). We conclude that the left side of Eq. (\ref{eq:TSsym-coeffs03}) equals $(-1)^{s+1}W_{s+1}(t) = (-1)^{N-\nu}T_{N-\nu}(t)$ as desired, completing the proof for the odd case. Thus, the forward direction is proven for both even and odd $n$.

The reverse direction is easy to understand as follows. Note that if Eq. (\ref{eq:TSsym-coeffs}) holds for some $\mathbf{c}\in\mathbb{R}^N$, then Eqs. (\ref{eq:TSsym-coeffs01}) and (\ref{eq:TSsym-coeffs03}) imply that Eq. (\ref{eq:TSsym-coeffs02}) holds as well, from which it follows that $g_\nu(t) = 0$ for all $\nu=1,\ldots N$. The normalization condition $\sum_\nu A'_{0,\nu} c_\nu = 1$ means that $g_0(t)=1$, so $A'(t)\mathbf{c}=\mathbf{d}$, as desired.
\end{proof}
\subsection{Proof of Theorem \ref{thm:S-TS-criterion}}\label{appendix:S-TS-criterion}
To complete the proof in the main text, it remains to be shown that for $n>0$, the inequalities
\begin{align}\label{eq:coeffs-ineq-appendix}
        0\leq \begin{cases}
        (-1)^{N-1-\nu}T_{N-1-\nu}(\cos{\theta}) &n\ \mathrm{even}\\
        (-1)^{N-1-\nu}W_{N-1-\nu}(\cos{\theta}) &n\ \mathrm{odd},
    \end{cases}
\end{align}
are satisfied for all $\nu=0,\ldots N-1$ if and only if $\theta\geq \frac{(n-1)\pi}{n}$. When $n=1$, the result is trivial, so assume $n>1$. First suppose $n$ is even. For even $n$, the change of variables $k=N-1-\nu$ allows the inequalities in Eq. (\ref{eq:coeffs-ineq-appendix}) to be rewritten as
\begin{align}
    (-1)^kT_k(\cos{\theta})\geq 0
\end{align}
for all $k=0,\ldots,\frac{n}{2}$. We can write $T_k(\cos{\theta})$ as 
\begin{align}
    T_k(\cos{\theta}) = \cos(k\theta).
\end{align}
Since $\cos{k\theta}=0$ when $k\theta=\frac{\pi}{2}+j\pi$ for integers $j$, it follows that the zeros of $T_k(\cos{\theta})$ over the interval $\theta\in[0,\pi]$ are given by
\begin{align}
    \theta = \frac{2j+1}{2k}\pi,\ j=0,\ldots,k-1.
\end{align}
For $k>0$, define $\theta_k=\frac{2k-1}{2k}\pi$ to be the largest such zero on the interval $[0,\pi]$; likewise, define $\theta_0=0$. Observe that for $\theta\in[\theta_k,\pi]$, $T_k(\cos{\theta})\geq0$ if $k$ is even and $T_k(\cos{\theta})\leq 0$ if $k$ is odd. Subsequently, $(-1)^kT_k(\cos{\theta})\geq 0$ for $\theta\in[\theta_k,\pi]$. Furthermore, for all $0\leq k' < k$, we have $(-1)^{k'}T_{k'}(\cos{\theta})\geq 0$ for $\theta\in[\theta_k,\pi]$ as well, since $\theta_{k'}<\theta_{k}$. Thus if $\theta\geq \frac{(n-1)\pi}{n}$, then $\theta \geq \theta_{\frac{n}{2}}$, and it follows by the above observations that all the inequalities must hold. 

To prove the other direction, we first show that if $\theta_{k-1}\leq \theta < \theta_k$ for some $k\geq 2$, then $(-1)^kT_k(\cos{\theta})<0$. Since $k\geq 2$, $T_k(\cos{\theta)}$ has more than one zero on the interval $[0,\pi]$. The second largest zero of $T_k(\cos{\theta)}$ is given by $\theta_k-\frac{1}{k}$. Evidently, $(-1)^kT_k(\cos{\theta})<0$ if $\theta_k-\frac{1}{k}<\theta<\theta_k $. However, we also have
\begin{align}
    \theta_{k-1} = \frac{2(k-1)-1}{2(k-1)}\pi = \frac{2k-3}{2k-2}\pi > \frac{2k-3}{2k}\pi = \theta_{k}-\frac{1}{k},
\end{align}
from which the claim follows. 

Now suppose that $\theta < \frac{(n-1)\pi}{n}$. Then either $0\leq \theta <\theta_1$ or there exists some $k\in\{2,\ldots, \frac{n}{2}\}$ such that $\theta_{k-1}\leq \theta < \theta_k$. In the former case, $-T_1(\cos{\theta}) = -\cos{\theta}< 0$ since $\theta_1=\frac{\pi}{2}$, while in the latter, the above result implies that some $(-1)^kT_l(\cos{\theta})<0$. Either way, not all of the inequalities in Eq. (\ref{eq:coeffs-ineq-appendix}) hold. 
It follows that Eq. (\ref{eq:coeffs-ineq-appendix}) is satisfied if and only if $\theta\geq\frac{(n-1)\pi}{n}$ when $n$ is even.

When $n$ is odd, the proof proceeds very similarly. For odd $n$, the change of variables $k=N-1-\nu$ allows the inequalities in Eq. (\ref{eq:coeffs-ineq-appendix}) to be rewritten as
\begin{align}
    (-1)^kW_k(\cos{\theta})\geq 0
\end{align}
for all $k=0,\ldots,\frac{n-1}{2}$. We can write $W_k(\cos{\theta})$ as 
\begin{align}
    W_k(\cos{\theta}) = \frac{\sin((k+\frac{1}{2})\theta)}{\sin{(\theta/2)}}
\end{align}
Since $\sin{(k+\frac{1}{2})\theta}=0$ when $(k+\frac{1}{2})\theta=j\pi$ for integers $j$, it follows that the zeros of $W_k(\cos{\theta})$ over the interval $\theta\in[0,\pi]$ are given by
\begin{align}
    \theta = \frac{2j}{2k+1}\pi,\ j=1,\ldots,k.
\end{align}
Note that $\theta=0$ is not a zero because of the $\sin(\theta/2)$ in the denominator of $W_k(\cos{\theta})$. For $k>0$, define $\theta_k=\frac{2k}{2k+1}\pi$ to be the largest such zero on the interval $[0,\pi]$; likewise, define $\theta_0=0$. We again have $(-1)^kW_k(\cos{\theta})\geq 0$ for $\theta\in[\theta_k,\pi]$, which implies that $(-1)^{k'}W_{k'}(\cos{\theta})\geq 0$ for all $0\leq k' < k$ and $\theta\in[\theta_k,\pi]$ since $\theta_{k'}<\theta_{k}$. Thus if $\theta\geq \frac{(n-1)\pi}{n}$, then $\theta \geq \theta_{\frac{n-1}{2}}$, so all the inequalities must hold. 

We can readily prove by exact analogy that if $\theta_{k-1}\leq \theta < \theta_k$ for some $k\geq 2$, then $(-1)^kW_k(\cos{\theta})<0$. As for the even case, this is sufficient to show that if $\theta<\frac{(n-1)\pi}{n}$, then not all of the inequalities in Eq. (\ref{eq:coeffs-ineq-appendix}) hold, thus completing the proof for the odd case. With the odd and even cases proven, the proof for Theorem \ref{thm:S-TS-criterion} is complete.

\subsection{Proof of Theorem \ref{thm:small-sym}}\label{appendix:small-cases}
In this section, we prove Theorem \ref{thm:small-sym}. 
\begin{proof}[Proof of Theorem \ref{thm:small-sym}]
First consider part (a). If $m=0$ or $m=n$, then $\abs{\mathcal{T}}=1$. Then $\ket{\psi}$ is a TS state if and only if $\bra{\psi}R^{\dagger(T)}(\theta)R^{(T)}(\theta)\ket{\psi}=\braket{\psi|\psi}=1$, which is true for any state $\ket{\psi}$.

Now consider part (b). If $m=1$ or $m=n-1$, then $M=1$. Also, assume $n>1$. Theorem \ref{thm:lin-prog} (or equivalently, Theorem \ref{thm:red-lin-prog}) implies that a TS state exists at a particular $\theta$ if and only if there exists $\mathbf{c}\geq 0$ such that $\sum_{\nu}A_{0,\nu}c_\nu=1$ and $\sum_{\nu}A_{1,\nu}(t)c_\nu=0$. Any appropriately normalized state can satisfy the first condition, so we focus on the second condition. We have
\begin{align}
    \sum_{\nu}A_{1,\nu}(t)c_\nu &= \sum_{\nu=0}^{N-1}\alpha_\nu\sum_{i,i'}\binom{1}{i}\binom{1}{i'}\binom{n-2}{\nu-(i+i')}T_{\abs{i-i'}}(t)c_\nu\\
    &= \sum_{\nu=0}^{N-1}\alpha_\nu\left[\binom{1}{0}\binom{1}{0}\binom{n-2}{\nu}+ \binom{1}{1}\binom{1}{1}\binom{n-2}{\nu-2}+2\binom{1}{1}\binom{1}{0}\binom{n-2}{\nu-1}t\right]c_\nu\\
    &= \sum_{\nu=0}^{N-1}\alpha_\nu\left[\binom{n-2}{\nu}+ \binom{n-2}{\nu-2}+ 2\binom{n-2}{\nu-1}t\right]c_\nu\\
    &= \sum_{\nu=0}^{N-1}\alpha_\nu\left[\binom{n-2}{\nu}+\binom{n-2}{\nu-2}\right]c_\nu+ \sum_{\nu'=0}^{N-1}2\alpha\binom{n-2}{\nu'-1}c_\nu't.
\end{align}
Setting $\sum_{\nu}A_{1,\nu}(t)c_\nu=0$ and solving for $t$, we obtain
\begin{align}
    t= \frac{-\sum_{\nu=0}^{N-1}\alpha_\nu\left[\binom{n-2}{\nu}+\binom{n-2}{\nu-2}\right]c_\nu}{\sum_{\nu'=0}^{N-1}2\alpha_\nu\binom{n-2}{\nu'-1}c_\nu'}
\end{align}
We would like to find the maximum possible value of $t$ as a function of the $c_\nu$ (as this corresponds to the minimum achievable $\theta$). We can cast this maximization problem as a linear fractional program:
\begin{align}
    \mathrm{maximize}\ t\ \mathrm{subject\ to} \sum_{\nu}A_{0,\nu}c_\nu=1\ \mathrm{and}\ c\geq 0.
\end{align}
Note that the $A_{0,\nu}$ are constants which do not depend on $t$. Because a linear fractional program admits a basic feasible solution \cite{linfrac}, the maximum will occur when exactly one of the $c_\nu$ is nonzero and the rest are zero. Thus, 
\begin{align}
    t_{\mathrm{max}} &= \max_{\nu}\frac{-\binom{n-2}{\nu}-\binom{n-2}{\nu-2}}{2\binom{n-2}{\nu-1}}.
\end{align}
The maximum occurs at $\nu=N-1$, giving
\begin{align}
    t_{\mathrm{max}} &= -1 + \ceil*{\frac{n}{2}}^{-1}.
\end{align}
Note that for $n>1$, $t_{\mathrm{max}}\in[-1,0]$. We now show that if $t\in[-1,t_{\mathrm{max}}]$, then there exists $\mathbf{c}\geq 0$ satisfying $\sum_\nu A_{1,\nu}(t)c_\nu=0$. Our solution will have $c_\nu=0$ for all $\nu$ except $\nu=0$ and $\nu=N-1$. We can then write
\begin{align}
    0&=\sum_\nu A_{1,\nu}(t)c_\nu\\
    &=c_0 + \alpha_{N-1}\left[\binom{n-2}{N-1}+ \binom{n-2}{N-1-2}+ 2\binom{n-2}{N-1-1}t\right]c_{N-1}\\
    &=c_0 + 2\alpha_{N-1}\binom{n-2}{N-1-1}(t-t_{\mathrm{max}})c_{N-1}
\end{align}
and rearrange to recover an expression for $c_0$:
\begin{align}
    c_0 = 2\alpha_{N-1}\binom{n-2}{N-2}(t_{\mathrm{max}}-t)c_{N-1}.
\end{align}
Since the coefficient of $c_{N-1}$ in the above expression is nonnegative, for any $t\in[-1,t_{\rm max}]$ it is possible to find $c_0\geq 0$ and $c_{N-1}\geq 0$ such that $\sum_\nu A_{1,\nu}(t)c_\nu=0$ with $c_\nu=0$ for $\nu=1,\ldots,N-2$. We conclude that $t\in[-1,t_{\rm max}]$ is a necessary and sufficient condition for the existence of a $\sts$ state, thereby proving the theorem.
\end{proof}
\subsection{Proof Proposition \ref{prop:cyc-orbits}}\label{appendix:cyclic-orbs}
In this section, we prove Proposition \ref{prop:cyc-orbits}.
\begin{proof}[Proof of Proposition \ref{prop:cyc-orbits}]
Assume $G=Z_n$. We can uniquely write any $T\in\mathcal{T}_{\rm cyc}$ as $T=z^j([m])$ for some $j\in\{0,\ldots n-1\}$. Then for $j_1,k_1,j_2,k_2\in\{0,\ldots n-1\}$, two trajectory pairs $(z^{j_1}([m]),z^{k_1}([m]))$ and $(z^{j_2}([m]),z^{k_2}([m]))$ are in the same orbit under $\tilde{Z}_n$ if and only if $j_1-k_1 = \pm(j_2-k_2)\pmod n$.

``$\implies$" direction: If $(z^{j_1}([m]),z^{k_1}([m]))=(e,z^l)(z^{j_2}([m]),z^{k_2}([m]))$ for some $l\in\{0,\ldots n-1\}$, then $j_1= j_2+l\pmod n$ and $k_1=k_2+l\pmod n$, so $j_1-k_1 = j_2-k_2\pmod n$. If instead $(z^{j_1}([m]),z^{k_1}([m]))=(x,z^l)(z^{j_2}([m]),z^{k_2}([m]))$, then $j_1=k_2+l\pmod n$ and $k_1= j_2+l\pmod n$, so $j_1-k_1 = -(j_2-k_2)\pmod n$. ``$\impliedby$" direction: Suppose $j_1-k_1 = j_2-k_2\pmod n$, and let $l=j_2-j_1$. Then $k_1+l = k_2\pmod n$, so $(z^{j_1}([m]),z^{k_1}([m]))=(e,z^l)(z^{j_2}([m]),z^{k_2}([m]))$. Now suppose $j_1-k_1 = -(j_2-k_2)\pmod n$, and let $l=k_2-j_1$. Then $k_1+l=j_2\pmod n$, so $(z^{j_1}([m]),z^{k_1}([m]))=(x,z^l)(z^{j_2}([m]),z^{k_2}([m]))$.

The orbits $\mathcal{T}_{\rm cyc}^2/\tilde{Z}_n$ are then
\begin{align}
    \mathcal{T}_{\rm cyc}^2/\tilde{Z}_n = \left\{\Omega_\mu\ :\ \mu = 0,\ldots, \floor*{\frac{n}{2}}\right\}
\end{align}
where 
\begin{align}
    \Omega_\mu &=\operatorname{Orb}_{\tilde{Z}_n}[([m],z^\mu([m]))]\\
    &=\left\{(z^{j}([m]),z^{k}([m]))\ :\ \pm(j-k) = \mu\pmod n\right\}.
\end{align}
\end{proof}
\subsection{Toric code provides a TS state}\label{appendix:toric}
In this section, we show that the state $\ket{\toric}$ constructed in Section \ref{sec:toric} is a TS state which can discriminate the four trajectories in Figure \ref{fig:toric} when $\theta=\frac{\pi}{2}$. These trajectories constitute the set $\mathcal{T}_{\rm toric}=\{T_1,\ldots,T_4\}$, where $T_1=\{1345\}$, $T_2=\{5781\}$, $T_3=\{2457\}$, and $T_4=\{3582\}$. Let $\mathcal{S}_{\rm toric}$ be the stabilizer group whose independent generators are 
\begin{align}\label{eq:toric-generators}
    \mathcal{S}_{\rm toric}=\langle X^{(1248)},X^{(3567)},X^{(4568)},Z^{(1345)},Z^{(2346)},Z^{(5781)},-Z^{(12)},-Z^{(37)}\rangle,
\end{align}
and recall that $\ket{\toric}$ is the unique state stabilized by $\mathcal{S}_{\rm toric}$.

To prove that $\ket{\toric}$ is the desired TS state, we will use a combination of Theorems \ref{lemma:SG-simp} and \ref{thm:stab}. Namely, we will first find an appropriate permutation group $G$ and Pauli subgroup $\mathcal{S}$ under which both $\mathcal{T}_{\rm toric}^2$ and $\ket{\toric}$ are invariant. Then, we will use the anticommutation relation of Theorem \ref{thm:stab} to demonstrate that Eq. (\ref{eq:ortho}) holds for at least one representative $(T,T')$ per orbit of $\mathcal{T}_{\rm toric}^2/(S_G\rtimes G)$; it will then follow that $\ket{\toric}$ is a TS state by Theorem \ref{lemma:SG-simp}.

We will not choose $\mathcal{S}=\mathcal{S}_{\rm toric}$, since $\mathcal{T}_{\rm toric}^2$ is not invariant under $\mathcal{S}_{\rm toric}$. Instead, we will choose $\mathcal{S} = \{I^{\otimes 8},X^{\otimes 8}\}$. By Proposition \ref{prop:S-bitflip3}, $\mathcal{T}_{\rm toric}^2$ is invariant under $S=\sigma(\mathcal{S})$. Additionally, since $X^{\otimes 8} = X^{(1248)}X^{(3567)}$, it follows that $\mathcal{S}\subseteq\mathcal{S}_{\rm toric}$, which implies that $\ket{\toric}$ is invariant under $\mathcal{S}$ as well.

We now choose a suitable permutation group. Using permutation cycle notation, define $\pi_1 = (1\ 2)(4\ 8)$ and $\pi_2=(3\ 7)(4\ 8)$, and let $G=\langle\pi_1,\pi_2\rangle$ be the group generated by these permutations. Note that $\pi_1$ interchanges $T_1\leftrightarrow T_2$ and $T_3\leftrightarrow T_4$; likewise, $\pi_2$ interchanges $T_1\leftrightarrow T_4$ and $T_2\leftrightarrow T_3$. It follows that $\mathcal{T}_{\rm toric}$ is $G$-transitive, which implies that $\mathcal{T}_{\rm toric}^2$ is invariant under $G$ by Proposition \ref{prop:G-transitive}. On the other hand, for $\ket{\toric}$ to be invariant under $\mathcal{G}=P(G)$, it suffices to show that $\mathcal{G}\in\mathcal{N}(\mathcal{S}_{\rm toric})$. To understand this claim, suppose $\mathcal{G}\in\mathcal{N}(\mathcal{S}_{\rm toric})$. Then for every $D\in\mathcal{S}_{\rm toric}$ and $\pi\in G$, $D(P_\pi\ket{\toric}) = P_\pi D'\ket{\toric} = P_\pi\ket{\toric}$ for some $D'\in\mathcal{S}_{\rm toric}$. It follows that every $P_\pi\ket{\toric}$ is in the stabilizer space of $\mathcal{S}_{\rm toric}$; however, since this stabilizer space has dimension one, it must be true that $P_\pi\ket{\toric}=\ket{\toric}$ for all $\pi\in G$. We thus proceed to show that $\mathcal{G}$ normalizes $\mathcal{S}_{\rm toric}$. It is easy to see that $P_{\pi_1}$ commutes with every generator of $\mathcal{S}_{\rm toric}$ except for $Z^{(1345)},Z^{(2346)}$, and $Z^{(5781)}$. However, conjugating these remaining generators by $P_{\pi_1}$ still produces stabilizer elements:
\begin{align}
    P_{\pi_1}Z^{(1345)}P^\dagger_{\pi_1} &= Z^{(2358)} = Z^{(5781)}\left(-Z^{(12)}\right)\left(-Z^{(37)}\right)\nonumber\\
    P_{\pi_1}Z^{(2346)}P^\dagger_{\pi_1} &= Z^{(1386)} = Z^{(1345)}Z^{(5871)}Z^{(2346)}\left(-Z^{(12)}\right)\left(-Z^{(37)}\right)\nonumber\\
    P_{\pi_1}Z^{(5781)}P^\dagger_{\pi_1} &= Z^{(5742)} = Z^{(1345)}\left(-Z^{(12)}\right)\left(-Z^{(37)}\right).
\end{align}
Similarly, $P_{\pi_2}$ evidently commutes with every generator of $\mathcal{S}_{\rm toric}$ except the same three. Nonetheless, conjugating these generators by $P_{\pi_2}$ still yields stabilizer elements:
\begin{align}
    P_{\pi_2}Z^{(1345)}P^\dagger_{\pi_2} &= Z^{(5871)} \nonumber\\
    P_{\pi_2}Z^{(5871)}P^\dagger_{\pi_2} &= Z^{(1345)} \nonumber\\
    P_{\pi_2}Z^{(2346)}P^\dagger_{\pi_2} &= Z^{(2678)} =Z^{(1345)}Z^{(5871)}Z^{(2346)}\nonumber\\.
\end{align}
It follows that $\mathcal{G}\in\mathcal{N}(\mathcal{S}_{\rm toric})$, which implies that $\ket{\toric}$ is invariant under $\mathcal{G}$, as desired. 

Having shown that $\mathcal{T}_{\rm toric}^2$ is invariant under both $G$ and $S=\sigma(\mathcal{S})$, we can now compute the orbits in $\mathcal{T}_{\rm toric}^2/(S_G\rtimes G)=\mathcal{T}_{\rm toric}^2/\tilde{G}$. There are four such orbits, labelled $\Omega_\mu$ for $\mu=0,\ldots 3$, as follows:
\begin{align}
    \Omega_0 &=\left\{(T_1,T_1),(T_2,T_2),(T_3,T_3),(T_4,T_4)\right\}\nonumber\\
    \Omega_1 &= \left\{(T_1,T_2),(T_2,T_1),(T_3,T_4),(T_4,T_3)\right\}\nonumber\\
    \Omega_2 &= \left\{(T_1,T_3),(T_3,T_1),(T_2,T_4),(T_4,T_2)\right\}\nonumber\\
    \Omega_3 &= \left\{(T_1,T_4),(T_4,T_1),(T_2,T_3),(T_3,T_2)\right\}.\label{eq:A-orbits}
\end{align}

To prove that $\ket{\toric}$ is a TS state using Theorem \ref{lemma:SG-simp}, we must demonstrate that Eq. (\ref{eq:ortho}) holds for one representative trajectory pair per orbit. Eq. (\ref{eq:ortho}) clearly holds for every pair in $\Omega_0$. Now note that every $(T,T')$ in $\Omega_1,\Omega_2,$ or $\Omega_3$ satisfies $T\neq T'$. Thus, to verify Eq. (\ref{eq:ortho}) for a representative $(T,T')$ in each of these remaining orbits, we will invoke the anticommutation relation of Theorem \ref{thm:stab}. Specifically, this theorem implies that $\bra{\toric}\rbf\ket{\toric}=0$ if there exists a subspace $V$ containing $\ket{\toric}$ such that some $D\in\mathcal{S}_{\rm toric}$ satisfies $\{D,\rbf\}\Pi_V=0$.

First consider $\Omega_1$, and pick $(T_1,T_2)\in\Omega_1$ as a representative trajectory pair. Then $V$ can be chosen as the $+1$ eigenspace of $-Z^{(37)}$; since $-Z^{(37)}$ is in $\mathcal{S}_{\rm toric}$, $\ket{\toric}$ is contained in $V$. The projector onto $V$ is consequently $\Pi_V = (I-Z^{(37)})/2$. We can then show that $X^{(3567)}$ satisfies the anticommutation relation of Eq. (\ref{eq:anticom-proj}). Letting $R=R_Z(\pi/2)$, recall the identities $RXZ = iXR$ and $R^\dagger X Z = -iXR^{\dagger}$ from Section \ref{sec:C4}. Then 
\begin{align}
    \{X^{(3567)},\mathbf{R}^{(T_1,T_2)}\}Z^{(37)} &= \left(X^{(3567)}R^{\dagger(34)}R^{(78)}+R^{\dagger(34)}R^{(78)}X^{(3567)}\right)Z^{(37)}\\
    &= X^{(56)}R^{\dagger(4)}R^{(8)}\left[\left(XR^{\dagger}\right)^{(3)}\left(XR\right)^{(7)}+\left(R^{\dagger}X\right)^{(3)}\left(RX\right)^{(7)}\right]Z^{(37)}\\
    &= X^{(56)}R^{\dagger(4)}R^{(8)}\left[\left(RX\right)^{(3)}\left(R^\dagger X\right)^{(7)}+\left(R^{\dagger}X\right)^{(3)}\left(RX\right)^{(7)}\right]Z^{(37)}\\
    &= X^{(56)}R^{\dagger(4)}R^{(8)}\left[\left(RXZ\right)^{(3)}\left(R^\dagger XZ\right)^{(7)}+\left(R^{\dagger}XZ\right)^{(3)}\left(RXZ\right)^{(7)}\right]\\
    &= X^{(56)}R^{\dagger(4)}R^{(8)}\left[\left(iXR\right)^{(3)}\left(-iXR^{\dagger}\right)^{(7)}+\left(-iXR^{\dagger}\right)^{(3)}\left(iXR\right)^{(7)}\right]\\
     &= X^{(56)}R^{\dagger(4)}R^{(8)}\left[\left(XR\right)^{(3)}\left(XR^{\dagger}\right)^{(7)}+\left(XR^{\dagger}\right)^{(3)}\left(XR\right)^{(7)}\right]\\
     &= X^{(56)}R^{\dagger(4)}R^{(8)}\left[\left(R^\dagger X\right)^{(3)}\left(R X\right)^{(7)}+\left(XR^{\dagger}\right)^{(3)}\left(XR\right)^{(7)}\right]\\
     &= X^{(56)}R^{\dagger(4)}R^{(8)}\left[\left(R^\dagger X\right)^{(3)}\left(R X\right)^{(7)}+\left(XR^{\dagger}\right)^{(3)}\left(XR\right)^{(7)}\right]\\
     &=\left(R^{\dagger(34)}R^{(78)}X^{(3567)}+X^{(3567)}R^{\dagger(34)}R^{(78)}\right)\\
     &=\{X^{(3567)},\mathbf{R}^{(T_1,T_2)}\}.
\end{align}
Since $\{X^{(3567)},\mathbf{R}^{(T_1,T_2)}\}Z^{(37)} = \{X^{(3567)},\mathbf{R}^{(T_1,T_2)}\}$, it follows that 
\begin{align}
    \{X^{(3567)},\mathbf{R}^{(T_1,T_2)}\}\Pi_V = \{X^{(3567)},\mathbf{R}^{(T_1,T_2)}\}\left(\frac{I-Z^{(37)}}{2}\right)
    &=0.
\end{align}
By Theorem \ref{thm:stab}, it must hence be true that $\bra{\toric}\mathbf{R}^{(T_1,T_2)}\ket{\toric}=0$.

For $\Omega_2$, we pick $(T_1,T_3)$ as a representative. Then identically choose $V$ to be the $+1$ eigenspace of $-Z^{(37)}$ so that $\Pi_V = (I-Z^{(37)})/2$. It can be verified by a similar argument that $X^{(3567)}$ again satisfies the anticommutation relation of Eq. (\ref{eq:anticom-proj}) with $\mathbf{R}^{(T_1,T_3)}$, from which it follows that $\bra{\toric}\mathbf{R}^{(T_1,T_3)}\ket{\toric}=0$. For $\Omega_3$, let $(T_1,T_4)$ be the representative and let $V$ instead be the $+1$ eigenspace of $-Z^{(48)}=Z^{(1345)}Z^{(1578)}\left(-Z^{(37)}\right)$. Then $X^{(4568)}$ can be shown to satisfy Eq. (\ref{eq:anticom-proj}) with $\mathbf{R}^{(T_1,T_4)}$, which implies $\bra{\toric}\mathbf{R}^{(T_1,T_4)}\ket{\toric}=0$.

Since $\ket{\toric}$ is invariant under $\mathcal{G}$ and $\mathcal{S}$ and Eq. (\ref{eq:ortho}) holds for one representative trajectory pair per orbit in $\mathcal{T}_{\rm toric}^2/\tilde{G}$, we conclude by Theorem \ref{lemma:SG-simp} that $\ket{\toric}$ is a TS state for the trajectory set $\mathcal{T}_{\rm toric}$ when $\theta = \frac{\pi}{2}$.

\bibliography{References} 
\bibliographystyle{apsrev4-1}
\end{document}